%% file: main.tex
\documentclass[conference]{IEEEtran}
\IEEEoverridecommandlockouts
\ifCLASSINFOpdf
\else
\fi
\usepackage{amsmath,amsthm}

\newtheorem{definition}{Definition}
\newtheorem{proposition}[definition]{Proposition}
\newtheorem{theorem}[definition]{Theorem}
\newtheorem{corollary}[definition]{Corollary}
\newtheorem{lemma}[definition]{Lemma}
\newtheorem{remark}[definition]{Remark}

\usepackage[parfill]{parskip}    
\usepackage{verbatim}
\usepackage{graphicx}
\usepackage{epstopdf}
\usepackage{bussproofs}
\usepackage{color} 
\usepackage{dsfont}
\usepackage{stmaryrd}
\usepackage{proof}
\usepackage{mathtools}
\usepackage{amssymb}
\usepackage{hyperref}
\usepackage{enumerate}
\usepackage{xcolor}
\usepackage{todonotes}
\usepackage{framed}
\usepackage{comment}
\usepackage{tikz-cd}
\usepackage{wrapfig}
\definecolor{shadecolor}{gray}{0.95}
\usepackage{faktor} 
\specialcomment{ndr}{\begin{shaded*}}{\end{shaded*}}
\usepackage{scalerel}
\input{macro}

\newcommand{\vvm}
{}
	
\usepackage{caption}
\usepackage{subcaption}

\begin{document}
%
\title{Combining Nondeterminism, Probability, and Termination: Equational and Metric Reasoning}



\author{\IEEEauthorblockN{Matteo Mio}
\IEEEauthorblockA{CNRS \& LIP, ENS Lyon, 
France}
\and
\IEEEauthorblockN{Ralph Sarkis}
\IEEEauthorblockA{ENS Lyon, 
France}
\and
\IEEEauthorblockN{Valeria Vignudelli}
\IEEEauthorblockA{CNRS \& LIP, ENS Lyon, 
France}
\thanks{This is an extended version, with appendix, of a paper accepted at LICS 2021.}
}



%



\IEEEoverridecommandlockouts

\maketitle

\begin{abstract}
We study monads resulting from the combination of nondeterministic and probabilistic behaviour with the possibility of termination, which is essential in program semantics. Our main contributions are presentation results for the monads, providing  
equational reasoning tools for establishing equivalences and distances of programs.
\end{abstract}


%
\IEEEpeerreviewmaketitle

\section{Introduction}\label{introduction:section}
\input{introduction}

\section{Background}\label{sec:back}

\input{categoryback}

\subsection{Monads on $\Sets$ and Equational Theories}
\label{sec:set:back}
\input{basicdefs}

\subsection{Monads on $\Met$ and Quantitative Equational Theories}
\label{sec:met:back}
\input{quantitative_algebras}

%
%

\section{$\Sets$ monad $\moncb$ and its presentation}
\label{sec:cplusone}
\input{PresCp1}

\section{$\Sets$ monad $\cset^{\downarrow}$ and its presentation}
\label{sec:cdownarrow}
\input{cdownarrow}

\section{Results about monads on $\Met$}

%

\label{sec:presentationsinmet}
\input{presentationsinmet}

\section{Examples of Applications}

\input{examples.tex}

\section{Directions for future work}

A technical question left open in this work (see Section \ref{sec:met:cplusone}) is the following: is there a monad structure on the $\Met$ functor $\lcset+\onemet$? 
A source of insight (of algebraic nature) is \cite[Thm. 5.3]{SW2018} which is a unique-existence theorem. This possibly suggests that all $\Sets$ monads on $\cset +\oneset$ are presented by equational theories of convex semilattices including the black--hole axiom. This would imply, by our Theorem  \ref{thm:bhtrivial}, that they cannot be lifted to $\Met$.





Another technical question, naturally emerging from our work in Section \ref{sec:cplusone} on the monad on the $\Sets$ functor $\cset + \oneset$ and its corresponding theory of convex semilattices with bottom and black--hole, is the following: is it possible to obtain another monad structure on the functor $\cset + \oneset$ presented by the theory of convex semilattices with top ($x \oplus \star = \star$) and black--hole? Preliminary work suggests that this is indeed the case by using, following the same procedure described in Section \ref{sec:cplusone}, the monad distributive law $\gamma^\top$ defined as:
\vvm
\[\gamma^\top_X (S) = \begin{cases}
    \left\{ \distr \mid \distr \in S \text{ and } \distr(\point) = 0\right\} &\forall \distr \in S \text{, } \distr(\point) = 0\\\point &\text{otherwise.}
\end{cases}
\vvm
\]
 We remark that, since the equational theory involves the black--hole axiom, when moving to the category $\Met$, the result of Theorem \ref{thm:bhtrivial} applies.
Analogously to our work in Section \ref{sec:cdownarrow}, we can also aim at characterizing the monad presented by the  theory of convex semilattices with top
(without black--hole) and its metric lifting. Such theory based on the top axiom has applications in testing--based equivalences \cite{DH84,BSV19}.

A main line for future research is the development of compositional verification techniques, along the lines of the illustrative examples presented in Section \ref{examples:section}.
This includes the axiomatization of behavioural equivalences and distances in expressive probabilistic programming languages, with features such as recursion and parallel composition.


Other directions for future research are inspired by recent related works in the literature. The machinery of \emph{weak distributive laws} of \cite{DBLP:conf/lics/0003P20} might offer a convenient alternative route for obtaining some of our results by compositionally combining results on the simpler monads $\pset$ and $\dset$, rather than working directly with $\cset$. In this direction, using weak distributive laws, the authors of \cite{DBLP:journals/corr/abs-2012-14778} have studied (possibly empty) convex powersets of left--semimodules over a class of semirings and shown presentation results. In another direction, several works (e.g., \cite{DBLP:conf/lics/0003P20,radu2016, DBLP:journals/corr/abs-2012-14778, DBLP:journals/corr/abs-1802-07366}) have considered more general variants of the monads $\pset$, $\dset$ and $\cset$, e.g., allowing non--finitely generated convex sets, probability distributions with infinite support or even probability measures. Adapting the results of our work to these variants likely requires, at the level of the equational presentations, the introduction of appropriate infinitary operations and rules.

\newpage
\bibliographystyle{IEEEtranS}
\bibliography{biblio}

\newpage
\onecolumn

\appendix

\subsection{Appendix to background}
\input{App_background}

\subsection{Proofs for Section \ref{sec:cplusone}}


\input{App_CSBandBH}


\subsection{Proofs for Section \ref{sec:cdownarrow}}
\input{App_Cdownarrow}

\input{App_presentationsinmet}

\subsection{Proofs for Section \ref{examples:section}}

\input{App_examples}

\end{document}

%% file: macro.tex
\newcommand{\bcl}{\xi}

\newcommand{\pccsn}[1]{\pccs_{#1}}
\newcommand{\prov}{\stackrel{\et}{\sim}}
\newcommand{\termmon}{\mathcal T_{\!/\et}}

\newcommand*{\ntodo}[1]{\hspace{1pt}\newline\todo[color=green,inline]{#1}}

\newcommand{\distr}{\varphi}
\newcommand{\ddistr}{\Phi}
\newcommand{\distrb}{\psi}

\newcommand{\distrc}{\theta}

\newcommand{\R}{\mathrel{\mathcal R}}

\newcommand{\supp}[1]{\stackrel{#1}{\mapsto}}

\newcommand{\terms}[2]{\mathcal T(#1,#2)}

\newcommand{\ics}{\kappa}

\newcommand{\alga}{\mathbb A}

\newcommand{\qet}{\mathtt{QTh}}
\newcommand{\et}{\mathtt{Th}}

\newcommand{\sigcs}{\Sigma_{CS}}

\newcommand{\sigpcs}{\Sigma_{CS}^{\star}}

\newcommand{\point}{\star}
\newcommand{\bpoint}{\ast}

\newcommand{\pplus}[1]{+_{#1}}
\DeclareMathOperator*{\bigplus}{\scalerel*{+}{\textstyle{\sum}}}

\newcommand{\pccs}{Proc}
\newcommand{\ndplusccs}{\mathrel{\overline\oplus}}
\newcommand{\pplusccs}[1]{\mathrel{\overline{+}_{#1}}}

\newcommand{\dirac}[1]{\delta_{#1}}

\newcommand{\haus}{H}
\newcommand{\kant}{K}
\newcommand{\hk}{\haus\kant}

\newcommand{\fpset}{\mathcal P}

\newcommand{\pset}{\mathcal P}
\newcommand{\dset}{\mathcal D}
\newcommand{\cset}{\mathcal C}
\newcommand{\oneset}{\mathbf{1}}
\newcommand{\moncb}{\cset+\oneset}

\newcommand{\moncbp}[1]{\cset(#1)+\oneset}

\newcommand{\moncd}{\cset^{\downarrow}}
\newcommand{\moncdp}[1]{\moncd(#1)}

\newcommand{\lcset}{\hat{\cset}}

\newcommand{\lpset}{\hat{\pset}}
\newcommand{\ldset}{\hat{\dset}}
\newcommand{\lmoncd}{\widehat\moncd}

\newcommand{\ldownarrow}{\widehat{\cset^\downarrow}}
\newcommand{\onemet}{\hat\oneset}
\newcommand{\donemet}{d_{\hat\oneset}}

\newcommand{\EM}{\mathbf E\mathbf M}

\newcommand{\Cat}{\mathbf C}

\newcommand{\Sets}{\mathbf{Set}}
\newcommand{\Met}{\mathbf{1Met}}

\newcommand{\wms}{\textnormal{\texttt{WMS}}}

\newcommand{\kompactd}[2]{\mathtt{Comp}(#1,#2)}

\newcommand{\conv}{cc}

\newcommand{\support}{supp}

\newcommand{\mon}{M}

\newcommand{\nil}{\mathbf{nil}}
\newcommand{\trans}{\rightarrowtail}

\newcommand{\onehalf}{\frac 1 2}

\newcommand{\ema}{Eilenberg-Moore\ }

\newcommand{\ub}{\mathtt{UB}} 

\newcommand{\acat}{\mathbf A}
\newcommand{\qacat}{\mathbf {QA}}

\newcommand{\cplus}{\oplus} 
\newcommand{\bigcplus}{\bigoplus} 
\newcommand{\bigpplus}{\bigplus} 

\newcommand{\etsl}{\et_{SL}}
\newcommand{\etca}{\et_{CA}}
\newcommand{\etcs}{\et_{CS}}

\newcommand{\etpcs}{\et_{CS}^{\star}}
\newcommand{\etcsbb}{\et_{CS}^{\bot,BH}}

\newcommand{\etcsb}{\et_{CS}^\bot}

\newcommand{\qetcs}{\qet_{CS}}

\newcommand{\qetpcs}{\qet_{CS}^{\star}}
\newcommand{\qetcsbb}{\qet_{CS}^{\bot, BH}}
\newcommand{\qetcsb}{\qet_{CS}^{\bot}}
\newcommand{\qbot}{\bot_Q}
\newcommand{\qbh}{BH_Q}

\newcommand{\vvcut}[1]{}

\newcommand{\finaltodo}[1]{}

\newcommand{\inl}{\mathsf{inl}}
\newcommand{\inr}{\mathsf{inr}}
\newcommand{\id}{\mathrm{id}}
\newcommand{\one}{\mathds{1}}
\newcommand{\cc}[1]{cc\left\{ #1 \right\}}
\newcommand{\mPne}{\mathcal{P}_{\text{ne}}}

%% file: introduction.tex


In the theory of programming languages the categorical concept of \emph{monad} is used to handle computational effects \cite{Moggi-89,Moggi-91}. As main examples, the \emph{non--empty finite powerset monad} ($\pset:\Sets\rightarrow\Sets$) and the \emph{finitely supported probability distribution monad} ($\dset:\Sets\rightarrow\Sets$) are used to handle nondeterministic and probabilistic behaviours, respectively. The \emph{non--empty convex sets of probability distributions monad} ($\cset:\Sets\rightarrow\Sets$) has been identified in several works (see, e.g., \cite{DBLP:journals/entcs/TixKP09a,DBLP:conf/fossacs/Goubault-Larrecq08a,DBLP:conf/csl/Goubault-Larrecq07,DBLP:journals/entcs/Mislove06,Mislove00,Jacobs08,DBLP:conf/fossacs/Mio14,BSV19,DBLP:conf/lics/0003P20,MV20}) 
as a convenient way to handle the combination of nondeterminism and probability. Liftings of these monads to the category of (1--bounded) metric spaces have been investigated using the technical machinery of Hausdorff and Kantorovich metric liftings:  ($\lpset:\Met\rightarrow\Met$),   ($\ldset:\Met\rightarrow\Met$) and more recently ($\lcset:\Met\rightarrow\Met$)  \cite{breugel2005,BaldanBKK18,MV20}. The category $\Met$  is a natural setting when it is desirable to switch from the concept of \emph{program equivalence} to that of \emph{program distance}.

Monads are tightly connected with equational theories. Mathematically, this connection emerges from the categorical 
notion of Eilenberg-Moore (EM) algebras. 
For every monad $M$ there is an associated category $\EM(M)$ of Eilenberg-Moore algebras for $M$ and, in many interesting cases, this can be presented by (i.e., proved isomorphic to) a well--known category of algebras (in the standard sense of universal algebra, i.e., models of an equational theory and their homomorphisms). For instance, $\EM(\pset)$ is isomorphic to the category of semilattices and semilattice--homomorphisms.  This is the mathematical fact underlying the ubiquity of semilattices in mathematical treatments of nondeterminism and is the basis of several advanced techniques for reasoning about nondeterministic programs (e.g., \emph{bisimulation up--to techniques}   \cite{PS11,DBLP:journals/acta/BonchiPPR17,BonchiKP18}.) 
Other important examples include the presentations of the monads $\dset$ and $\cset$ by the equational theories of 
convex (a.k.a. barycentric) algebras  \cite{swirszcz:1974,Doberkat0608,jacobs:2010} and 
convex semilattices \cite{BSV19,BSV20ar}, respectively.
 Recently, presentation results have been obtained also for the $\Met$ variants of these monads, $\lpset$, $\ldset$ and $\lcset$, 
 using the framework of quantitative algebras and quantitative equational theories of 
 \cite{radu2016,DBLP:conf/lics/MardarePP17,DBLP:conf/lics/BacciMPP18,BacciBLM18,DBLP:journals/entcs/Bacci0LM18}.
These three 
$\Met$ 
monads 
are presented by the quantitative equational theories of quantitative semilattices \cite{radu2016}, quantitative convex algebras (referred to as  barycentric algebras in \cite[\S 8]{radu2016}) and quantitative convex semilattices \cite{MV20}, respectively. 

These presentation results provide equational methods for reasoning about equivalences and distances of programs whose semantics is modelled as a transition system (i.e., a coalgebra) of type $\textnormal{States} \rightarrow F(\textnormal{States})
$, for $F\in\{\pset,\dset,\cset,\lpset,\ldset,\lcset  \}$.\footnote{Or, if labels $L$ are considered, systems of type $\textnormal{States} \rightarrow \big(F(\textnormal{States})\big)^L$.} 
 However these functors may not be appropriate for all modelling purposes. Indeed, for all six functors above, the final $F$--coalgebra has the singleton set as carrier, which means that all states of an $F$--coalgebra are behaviourally equivalent. Usually, what is needed is some kind of \emph{behavioural observation} such as a termination state. This is generally achieved by using the functor $F+\oneset$ (where $+$ and $\oneset$ are the coproduct and the terminal object, respectively): a state can either transition to $F(\textnormal{States})$ or terminate by reaching $\oneset$. 
Even if the functor $F$ carries a monad structure and 
the functor $F +\mathbf{1}$ is similar to $F$,  separate work is needed to answer questions such as: is there a monad having $F +\mathbf{1}$ as underlying functor? How is this related to  the monad $F$? What is its presentation? For some specific cases the answers are well--known. 
For instance, the functor $\pset +\mathbf{1}$ (possibly empty finite powerset) carries a monad structure which is presented by semilattices with bottom (i.e., semilattices with a designated element $\point$, representing termination, 
such that $x\oplus \star = x$).  

\textbf{Contributions.}
\begin{enumerate}
\item We describe in Section \ref{sec:cplusone} a $\Sets$ monad whose underlying functor is $\cset +\oneset$ (possibly empty convex sets of probability distributions) and prove that it is presented by the theory of convex semilattices extended with the 
bottom axiom $x\oplus \star = x$ and the black--hole axiom $x +_p \point = \point$ (see \cite{MOW03,SW2018}). 
Transition systems of type $\cset+\oneset$ are well--known in the literature as (simple) convex Segala systems \cite{Seg95:thesis,BSV04:tcs,Sokolova11} 
and  are widely used to model the semantics of nondeterministic and probabilistic programs. Hence, this result provides equational reasoning methods for an important class of systems.
\item The black--hole axiom annihilates probabilistic termination, thus it is not appropriate in all modelling situations.
So, we investigate in Section \ref{sec:cdownarrow} a monad $\cset^\downarrow$ presented by the weaker theory of convex semilattices with bottom (but without black--hole). This equational theory has already found applications in the study of trace semantics of nondeterministic and probabilistic programs in \cite{BSV19}. 
\item In an attempt to find a $\Met$ monad structure $M$ on the Hausdorff--Kantorovich metric lifting of the $\Sets$ functor $\cset+\oneset$, we prove 
in Section \ref{sec:met:cplusone} some negative results. First, no such $M$ exists having as multiplication the same operation of the $\Sets$ monad $\cset + \oneset$. Secondly, $M$ cannot be presented by the quantitative equational theory of convex semilattices with bottom and black--hole, since this theory is trivial.
\item In Section \ref{sec:met:cdownarrow}, we identify the $\Met$ monad $\lmoncd$ which is the Hausdorff--Kantorovich metric lifting of the $\Sets$ monad $\cset^\downarrow$ of point (2). We exhibit a presentation of this monad via the quantitative equational theory of convex semilattices with bottom.
\end{enumerate}
We conclude with some examples of applications of our results to program equivalences and distances in Section \ref{examples:section}. 
Full proofs of the results presented in this paper and additional background material are available in the Appendix.

%% file: categoryback.tex


We present some definitions and results regarding monads. We assume the reader is familiar with basic concepts of category theory (see, e.g., \cite{Awodey}). 
Facts easily derivable from known results in the literature are systematically marked as ``Proposition'' throughout the paper.

\begin{definition}[Monad]\label{monad:main_definition}
Given a category $\Cat$, a monad on $\Cat$ is a triple $(\mon, \eta, \mu)$ composed of a functor $\mon\colon\Cat \rightarrow \Cat$ together with two
natural transformations: a unit $\eta\colon id_{\Cat}
\Rightarrow \mon$, where $id_{\Cat}$ is the identity functor on $\Cat$, and a multiplication $\mu \colon \mon^{2} \Rightarrow
\mon$, satisfying  
$\mu \circ \eta\mon = \mu \circ \mon\eta = id_{\Cat} $ and 
$  \mu\circ \mon\mu = \mu \circ\mu\mon$.
\end{definition}

If $\Cat$ has coproducts, $A_1,A_2,B\in\Cat$, $f_1:A_1\rightarrow B$ and $f_2:A_2\rightarrow B$,  we denote with $[f_1,f_2]:A_1+A_2\rightarrow B$ the unique morphism such that $f_{1}=[f_{1},f_{2}] \circ \inl$ and $f_{2}=[f_{1},f_{2}] \circ \inr$, where $\inl\!:\! A_1\! \rightarrow\! A_1 + A_2$ and  $\inr\!:\! A_2 \!\rightarrow\! A_1 + A_2$ are the canonical injections. We denote with $\oneset_\Cat$ the terminal object of $\Cat$, if it exists.



\begin{proposition}\label{prop:terminationcat}
Let $\Cat$ be a category having coproducts and a terminal object. The $\Cat$ monad $+\oneset$ is defined as the triple $(\cdot + \mathbf{1}_{\Cat}, \eta^{+ \mathbf{1}_{\Cat}},\mu^{+ \mathbf{1}_{\Cat}})$ whose functor $(\cdot + \mathbf{1}_{\Cat})$ is defined on objects as $A\mapsto  A+\mathbf{1}_{\Cat}$ and on arrows as $f\mapsto [\inl\circ f,\inr]$, with unit $\eta^{+ \mathbf{1}_{\Cat}}=\inl$ and with multiplication $\mu^{+ \mathbf{1}_{\Cat}}= [[\inl,\inr],\inr]$.
\end{proposition}

Monads can be combined together using the notion of monad distributive law.

\begin{definition}[Monad distributive law]\label{def:monadlaw}
    Let $(M, \eta, \mu)$ and $(\widehat{M}, \widehat{\eta}, \widehat{\mu})$ be two monads on $\Cat$. A natural transformation $\lambda: M \widehat{M}\Rightarrow \widehat{M}M$ is called a \emph{monad distributive law of $M$ over  $\widehat{M}$} if it satisfies the equations
    $\lambda \circ M\widehat\eta=\widehat\eta M$, $\lambda \circ \eta \widehat M=\widehat M \eta $, $\lambda \circ \mu \widehat M= \widehat M \mu \circ \lambda M \circ M \lambda$ and $\lambda \circ M \widehat \mu= \widehat \mu M \circ \widehat M \lambda \circ \lambda \widehat M$.
\end{definition}

\begin{proposition}\label{prop-distlawcomposition}
If $\lambda: M \widehat{M} \Rightarrow \widehat{M}M$ is a monad distributive law, then $(\overline{M},\overline{\eta}, \overline{\mu})$ is a monad with $\overline M=\widehat{M}M$, $\overline{\eta} = \widehat{\eta} \diamond \eta$ and $\overline{\mu} = (\widehat{\mu} \diamond \mu) \circ \widehat{M}\lambda M$.\footnote{\label{footnotediamond} For any pair of natural transformations $f:F_{1}\Rightarrow F_{2}$ and $g:G_{1}\Rightarrow G_{2}$, we let $f \diamond g =  G_{2} f\circ g F_{1}=g F_{2} \circ G_{1} f$ (\cite{Riehl}, Lemma 1.4.7).}
\end{proposition}


\begin{corollary}\label{eq:iota}
Let $\Cat$ have  coproducts and a terminal object and $M:\Cat\rightarrow \Cat$ be a monad. Then there is a $\Cat$ monad structure $(\mon(+\oneset), \eta^{\mon(+\oneset)}, \mu^{\mon(+\oneset)})$ on the functor $\mon(+\oneset)$, given by Proposition \ref{prop-distlawcomposition} using the monad distributive law $\iota\colon \mon+\oneset \Rightarrow \mon(\cdot+\oneset)$ defined as  $\iota_X =  {[\mon \inl, \eta^{\mon}_{X+\oneset}\circ \inr]}$.
\end{corollary}

A monad $\mon$ has an associated category of $M$--algebras. 
\begin{definition}[$M$--algebras]\label{def:algebra-of-a-monad}
Let $(\mon,\eta,\mu)$ be a monad on $\Cat$. An algebra for $\mon$ (or $\mon$--algebra) is a pair $(A,\alpha)$ where $A\in\Cat$ is an object and $\alpha:\mon (A)\rightarrow A$ is a morphism such that 
(1) $ \alpha \circ  \eta_A = id_A$ and 
(2) $\alpha\circ \mon \alpha= \alpha \circ \mu_A $ hold.
Given two $\mon$--algebras $(A,\alpha)$ and $(A^\prime,\alpha^\prime)$, an \emph{$\mon$--algebra morphism} is an arrow $f:A\rightarrow A^\prime$ in $\Cat$ such that
$
 f\circ \alpha = \alpha^\prime \circ \mon(f)  
$.
The category of $M$--algebras and their morphisms, denoted $\EM(M)$, is called the \ema category for $\mon$. 
\end{definition}

\begin{definition}[Monad map]\label{def:monadmap}
    Let $(\mon , \eta, \mu)$ and $(\widehat{\mon }, \widehat{\eta}, \widehat{\mu})$ be two monads. A natural transformation $\sigma: \mon  \Rightarrow \widehat{\mon }$ is called a \emph{monad map} if it satisfies the laws $\widehat \eta=\sigma \circ \eta$ and $\sigma \circ \mu=\widehat \mu \circ (\sigma \diamond \sigma)$ (see footnote \ref{footnotediamond}).
\end{definition}
\begin{proposition}[Theorem 6.3\cite{TTT}]\label{prop-MonmapFunctor}
    If $\sigma: \mon \Rightarrow \widehat{\mon}$ is a monad map, then $U^{\sigma} = (A, \alpha) \mapsto (A, \alpha \circ \sigma_A)$ is a functor $\EM(\widehat{\mon}) \rightarrow \EM(\mon)$.
\end{proposition}

%% file: basicdefs.tex

In this section, we restrict our attention to monads on the category $\Sets$, which has coproducts  (disjoint unions) and a terminal object (the singleton set $\{\point\}$). Hence, by Proposition \ref{prop:terminationcat} the $+\oneset$ monad, which we refer to as \emph{termination monad} (also known as \emph{maybe}, or \emph{lift}, or \emph{partiality monad} \cite{Moggi-91}), is well defined in $\Sets$. When no confusion arises, we omit explicit mentioning of the injections, and write for example $(f + \oneset)(x)=f(x)$ for $x\in X$ and $(f + \oneset) (\star)= \star$.

We now introduce the $\Sets$ monad $\cset$ of non--empty finitely generated convex sets of finitely supported probability distributions. This requires a number of definitions and notations regarding sets and probability distributions.




A probability distribution (respectively, subdistribution)
on a set $X$ is a function $\distr: X\rightarrow [0,1]$ such that $\sum_{x\in X} \distr(x) = 1$ (respectively,  $\sum_{x\in X} \distr(x)\leq 1$.) The \emph{support} of $\distr$ is defined as  $\support(\distr) = \{ x\in X \mid \distr(x) \neq 0\}$. In this paper, we only consider probability distributions with finite support, so we just refer to them as {distributions}. The Dirac distribution $\dirac x$ is defined as $\dirac x (x^\prime)=1$ if $x^\prime = x$ and $\dirac x (x^\prime)=0$ otherwise. 
We often denote a distribution having $\support(\distr) = \{x_{1},\dots, x_{n}\}$  by the expression $\sum_{i=1}^{n} p_{i} x_{i}$, with $p_i = \distr(x_i)$. We denote with $\dset(X)$ the set of finitely supported probability distributions on $X$.  This becomes a $\Sets$ functor by defining, for any $f:X\rightarrow Y$ in $\Sets$, the arrow $\dset{(f)}: \dset(X)\rightarrow \dset(Y)$ as the \emph{pushforward distribution}, i.e., $\dset{(f)}(\distr)\! =\!   \big(y\mapsto \sum_{x \in f^{-1}(y)} \distr(x) \big)$.


Given a set $S\subseteq\dset (X)$ of distributions, we denote with $\conv(S)$ the \emph{convex closure} of $S$, i.e., the set of distributions $\distr$ that are convex combinations  $\sum_{i=1}^{n} p_{i} \cdot \distr_{i}$ 
 of distributions $\distr_i\in S$.  Clearly $S\subseteq \conv(S)$.
We say that a convex set $S \subseteq\dset (X)$ is \emph{finitely generated} if there exists a finite set $S^\prime\subseteq \dset (X)$ 
such that $S=\conv(S^\prime)$. The finite set $S^\prime$ is referred to as a \emph{base} of $S$. 
Given a finitely generated convex set $S\subseteq \dset (X)$, there exists a minimal (with respect to the inclusion order)
base which we denote as $\ub(S)$ and refer to as the \emph{unique base} of $S$ (see, e.g., \cite{BSV20ar}). 
We denote with $\cset(X)$ the set defined as: $$\{  S\subseteq\dset(X) \mid\textnormal{$S\neq \emptyset$  and $S$ is convex and finitely generated}\}.$$
This can be turned into a $\Sets$ functor  by defining for every $f:X\rightarrow Y$ the arrow $\cset{(f)}: \cset(X)\rightarrow \cset(Y)$ as $\cset{(f)}(S)= \{ \dset{(f)}(\distr) \mid \distr\in S \} $. We are now ready to define a monad on the $\Sets$ functor $\cset$ (see \cite{BSV19}).


\begin{definition}[Monad $\cset$]\label{def:set:cset}
The \emph{non-empty finitely generated convex sets of distributions} $\Sets$ monad is the triple $(\cset, \eta^{\cset}, \mu^{\cset})$ consisting of the functor $\cset$, unit $\eta^\cset_X (x) = \{\dirac x\}$ and  multiplication 
defined, for any $S\in \cset\cset(X)$, as
\[
{\mu^\cset_X(S)= \bigcup_{\distr \in S}  \wms(\distr)}
\]
where, for any $\distr\in \dset\cset (X)$ of the form $\sum_{i=1}^{n} p_{i} S_{i}$, with $S_i\in\cset (X)$, the \emph{weighted Minkowski sum} operation 
$\wms: \dset\cset (X)\rightarrow \cset (X)$ is defined as follows:
\[
\wms(\distr) = \{ \sum_{i=1}^{n} p_{i} \cdot \distr_{i} \mid \textnormal{for each $1\leq i \leq n$, $\distr_i\in S_i$}\}.
\]
\end{definition}
As a consequence of Proposition \ref{prop-distlawcomposition} and Corollary \ref{eq:iota}, there is also a $\Sets$ monad $(\cset(\cdot+\oneset), \eta^{\cset(+\oneset)}, \mu^{\cset(+\oneset)})$ on the composition of $\cset$ and $+\oneset$.

\begin{proposition}[Monad $\cset(+\oneset)$]\label{def:set:subcset}
There is a $\Sets$ monad $(\cset(\cdot+\oneset), \eta^{\cset(+\oneset)}, \mu^{\cset(+\oneset)})$.
\end{proposition}

\begin{remark}\label{rem:subdistr}
There is a bijective correspondence between distributions $\distr$ on $X + \oneset$ and subdistributions $\distr^\prime$ on $X$ (i.e., $\distr(x)= \distr^\prime(x)$ for $x\in X$ and $\distr(\star)=1- \sum_{x\in X} \distr^\prime(x)$). We will use this identification and often refer to the monad $\cset(\cdot+\oneset)$ as the \emph{non--empty finitely generated convex sets of subdistributions} monad \cite{BSV19}.
\end{remark}

In \cite{BSV19}, presentation theorems for the monads $\cset$ and $\cset(+\oneset)$ are given in terms of the equational theories of convex semilattices and pointed convex semilattices, respectively. 
We assume the reader is familiar with the basic notions of universal algebra such as: signature, algebras for a signature, homomorphisms, etc. The textbook \cite{univalgebrabook} is a standard reference.

\begin{definition}[Convex Semilattices]\label{def:convexsemilattices:set}
The theory $\etcs$ of convex semilattices has signature $\sigcs=(\{\oplus\}\cup \{+_p\}_{p\in(0,1)})$ and the following axioms:

{
\begin{tabular}{l l}
$(A)$ $ \ $ 
&
$ x \oplus ( y \oplus z) = (x \oplus y) \oplus z \ \ \ \ $  \\
$(C)$ 
& 
$ x \oplus y = y \oplus x$\\
$(I) $
& 
$ x\oplus x = x$\\
$(A_{p})$ $ \ $ 
& 
$ (x+_qy)+_pz = x+_{pq}(y+_{\frac{p(1-q)}{1-pq}}z)$\\
$(C_{p}) $
& 
$ x+_py  =  y+_{1-p}x$\\
$(I_{p}) $
& 
$ x+_px  =  x$\\
$(D)$ $ \ $& $ x +_p (y \oplus z) = (x+_p y) \oplus (x +_p z)$
\end{tabular}
}

A convex semilattice is a $\sigcs$-algebra satisfying the equations in $\etcs$.
We let $\acat(\etcs)$ denote the category of convex semilattices and their homomorphisms.
\end{definition}

\begin{definition}[Pointed convex semilattices]\label{def:pointedcs:set}
The theory $\etpcs$ of pointed convex semilattices has signature  $\sigpcs=(\{\oplus\}\cup \{+_p\}_{p\in(0,1)}\cup\{\star\})$ and the same axioms of the theory of convex semilattices. We denote with $\acat(\etpcs)$ the category of pointed convex semilattices and their homomorphisms.
\end{definition}

The presentation theorems for $\cset$ and $\cset(+\oneset$) in \cite{BSV19} can now be formally stated as the following isomorphisms of categories.
\begin{proposition}\label{prop-knownpresentations}
\begin{enumerate}
\item The theory $\etcs$ of convex semilattices is a presentation of the monad $\cset$, i.e., $\EM(\cset)\cong \acat(\etcs)$.
\item The theory $\etpcs$ of pointed convex semilattices is a presentation of the monad $\cset(+\oneset)$, i.e., $\EM(\cset(+\oneset))\cong \acat(\etpcs)$.
\end{enumerate}
\end{proposition}
The isomorphism in Proposition \ref{prop-knownpresentations}.2 is given by a pair of functors
\vvm
\begin{equation}\label{eq:setiso}
\begin{split}
&P: \EM(\cset(+\oneset))\rightarrow \acat(\etpcs) \\   
&P^{-1}: \acat(\etpcs)\rightarrow \EM(\cset(+\oneset))
\end{split}
\vvm
\end{equation}
On objects $(A,\alpha) \in \EM(\cset(+\oneset))$, we have $P(A,\alpha)=(A, \oplus^{\alpha}, \{\pplus p^{\alpha}\}_{p\in (0,1)}, \star^{\alpha} )$ where, for all $a_{1},a_{2}\in A$:
\begin{align*}
&a_{1} \oplus^{\alpha} a_{2}= \alpha(\conv\{\dirac{a_{1}},\dirac{a_{2}}\})\qquad \star^{\alpha}=\alpha(\{\dirac \star\}) \\
&\qquad a_{1} \pplus p^{\alpha} a_{2}= \alpha(\{p \,{a_{1}} + (1-p) \,{a_{2}}\}).
\end{align*}
Given a pointed convex semilattice $\alga = (A, \oplus^{\alga}, +^{\alga}_p, \point^{\alga})\in \etpcs$, $P^{-1}(\alga)$ is the $\cset(\cdot+\mathbf{1})$--algebra $(A, \alpha)$ where $\alpha : \cset(A+\mathbf{1}) \rightarrow A$ is defined by 
\[\alpha(S)=\textstyle{\bigcplus^{\alga}_{\distr\in \ub(S) } (\bigpplus^{\alga}_{b \in \support(\distr) } \distr(b)\,b)}\]
where $\bigcplus_{i\in I} x_{i}$ and $\bigpplus_{i\in I} p_{i} \,x$ are respectively notations for the binary operations $\cplus^{\alga}$ and $\pplus p^{\alga}$ extended to operations of arity $I$, for $I$ finite (see, e.g., \cite{stone:1949,BSS17}), $\ub(S)$ is the unique base of $S$, and $b$ ranges over $A\cup\{\star\}$.

Using the above presentation result and the well-known fact that, for any monad $(M, \eta, \mu)$, the free $M$--algebra generated by $X$ is 
$(M(X), \mu_{X})$,
we can identify (up to isomorphism) the free pointed convex semilattices.


\begin{proposition}\label{prop:freepcs}
The free pointed convex semilattice on $X$ is  (up to isomorphism): $$(\cset(X+\oneset), \oplus^{\cset(+\oneset)}, \{\pplus p^{\cset(+\oneset)}\}_{p\in (0,1)}, \star^{\cset(+\oneset)})$$ 
where, for all $S_1,S_2\in\cset(X+\oneset)$: 
\begin{align*}
&S_1\oplus^{\cset(+\oneset)} S_2 = \conv(S_1\cup S_2) \qquad \star^{\cset(+\oneset)}=\{\dirac \star\} \\
&\qquad  S_1\pplus p^{\cset(+\oneset)} S_2 = \wms( p \,S_1  + (1-p) S_2).
\end{align*}
\end{proposition}

%% file: quantitative_algebras.tex

In this section, we focus on monads on the category $\Met$ of $1$--bounded metric spaces and non--expansive maps. 

\begin{definition}[Category $\Met$]
{A 1--bounded} metric space is a pair $(X,d)$ with $X$ a set and 
 $d:X\times X\rightarrow [0,1]$ such that $d(x,y) = 0$ {if and only if $x=y$}, $d(x,y)=d(y,x)$, and $d(x,y) \leq d(x,z) + d(z,y)$, for all $x,y,z\in X$. A function {$f:X \rightarrow Y$} between two 1--bounded metric spaces {$(X,d_X)$ and $(Y,d_Y)$} is \emph{non--expansive} if {$d_Y( f(x_1), f(x_2) ) \leq d_X(x_1,x_2)$ for all $x_1,x_2\in X$}. We denote with $\Met$ the category  of 1--bounded metric spaces and non--expansive maps.
\end{definition}

Since we only work with $1$--bounded metric spaces, we often refer to them simply as metric spaces.
The category $\Met$ has a terminal object  (the space $\onemet=(\{\point\}, \donemet)$, with $\donemet (\point,\point)= 0)$), products, defined as $(X_1,d_{1}) \times (X_2,d_{2})=(X_1 \times X_2,  d_1\times d_2)$ with $d_1\times d_2 ((x_1,x_2), (y_1,y_2))=\sup_{i=1,2} d_i(x_i,y_i)$,  and coproducts,\footnote{This is the property that makes $\Met$ preferable to the category of all (possibly not $1$--bounded) metric spaces, since the latter does not have coproducts.} defined as $(X_1,d_{1}) + (X_2,d_{2})=(X_1 + X_2,  d_1+d_2)$ where $X + Y$ denotes disjoint union and $(d_1+d_2)(y,w) = d_{1} (y,w)$ if both $y,w\in X_1$, $(d_1+d_2)(y,w) = d_{2} (y,w)$ if both $y,w\in X_2$ and  $(d_{1}+d_2)(y,w) =1$ otherwise. 

 We now introduce the $\Met$ monad $\lcset$, which is the Hausdorff--Kantorovich metric lifting of the $\Sets$ monad $\cset$ and which has been introduced in  \cite[\S 4]{MV20}. 

\begin{definition}[Kantorovich Lifting]\label{def:kantorovich:lifting}
Let $(X,d)$ be a $1$--bounded metric space. The Kantorovich lifting of $d$ is a $1$--bounded metric $\kant(d)$ on $\dset(X)$, the collection of finitely supported distributions on $X$, 
assigning to any pair $\distr_1, \distr_2\in \dset(X)$ the distance $\kant(d) (\distr_1,\distr_2)$ defined as:
\vvcut{$$
\kant(d) (\distr_1,\distr_2)=
\sup_{\stackrel{f:X\to [0,1]}{\textnormal{non--expansive}}} \Big\{  |  \big(\sum_{x\in X} f(x)\distr_1(x) \big) - \big(\sum_{x\in X} f(x)\distr_2(x) \big) | \Big\}
$$
where $| \_ |$ denotes the absolute value. Equivalently\footnote{The equivalence is a nontrivial result and follows from the Kantorovich-Rubinstein duality theorem.}, the metric $\kant(d)$ can be defined as:}
\vvm
\[
\inf_{\omega \in Coup(\distr_1,\distr_2)}\Big( \sum_{(x_1,x_2)\in X \times X }\omega(x_1,x_2) \cdot d(x_1,x_2) \Big)
\vvm
\]
where $Coup(\distr_1,\distr_2)$ is defined as the collection of couplings of $\distr_1$ and $\distr_2$, i.e., $
 Coup(\distr_1,\distr_2) = \{ \omega\in\dset (X\times X) \mid  \dset(\pi_1)(\omega) = \distr_1 \textnormal{ and }  \dset(\pi_2) (\omega) = \distr_2 \}
$
where $\pi_1:X_1\times X_2 \rightarrow X_1$ and $\pi_2:X_1\times X_2 \rightarrow X_2$ are the projection functions. 
\end{definition}

\begin{definition}[Hausdorff Lifting]
Let $(X,d)$ be a $1$--bounded metric space. The Hausdorff lifting of $d$ is a $1$--bounded metric $\haus(d)$ on $\kompactd X d$,  the collection of non--empty compact subsets of $X$ (with respect to the standard metric topology induced by $d$),
assigning to any pair $X_1,X_2 \in \kompactd X d$ the distance $\haus(d)\big(X_1,X_2)$ defined as:
\vvm
$$
\max\big\{  \sup_{x_1\in X_1}\inf_{x_2\in X_2}d(x_1,x_2)  \ \     , \ \ \sup_{x_2\in X_2}\inf_{x_1\in X_1}d(x_1,x_2)    \big\}.
$$
\end{definition}
Hence, for every metric space $(X,d)\in\Met$, the collection of non--empty compact sets of finitely supported probability distributions can be endowed with the Hausdorff--Kantorovich lifted metric $\haus(\kant(d))$, which we write $\hk (d)$. Since all elements of $\cset(X)$ are compact, we obtain that $(\cset(X), \hk(d))$ is a $1$--bounded metric space. This leads to the definition of the $\Met$ monad $\lcset$ of \cite[\S 4]{MV20}.

\begin{definition}[Monad $\lcset$]
The monad $(\lcset, \eta^{\lcset}, \mu^{\lcset})$ on $\Met$ is defined as follows. The functor $\lcset$ is defined as mapping objects $(X,d)$ to $\big(\cset(X), \hk(d) \big)$ and morphisms $f:X\rightarrow Y$ to $\lcset(f) = \cset(f)$ (i.e., as the $\Sets$ functor $\cset$). The unit  $\eta^{\lcset}$ and the multiplication $\mu^{\lcset}$ are defined as for the $\Sets$ monad $\cset$ (Definition \ref{def:set:cset}). This is well--defined since both  $\eta^{\lcset}$ and $\mu^{\lcset}$ are non--expansive. 
\end{definition}

As a consequence of Proposition \ref{prop-distlawcomposition} and Corollary \ref{eq:iota}, there is also a $\Met$ monad $(\lcset(\cdot+\onemet), \eta^{\lcset(+\onemet)}, \mu^{\lcset(+\onemet)})$ on the composition of $\lcset$ and $+\onemet$. 

\begin{proposition}[Monad $\lcset(+\onemet)$]\label{def:set:subcmet}
There is a $\Met$ monad $(\lcset(+\onemet), \eta^{\lcset(+\onemet)}, \mu^{\lcset(+\onemet)})$.
\end{proposition}

Following Remark \ref{rem:subdistr}, we refer to $\lcset$ (respectively $\lcset(+\onemet)$) as the $\Met$ monad of non--empty finitely generated convex sets of distributions (respectively subdistributions) with the Hausdorff--Kantorovich metric. 

A main result of \cite{MV20} is a presentation result\footnote{This result can be easily adapted to obtain a presentation of $\lcset(+\onemet)$ too.} for $\lcset$, based on the recently introduced notions of quantitative algebras and quantitative equational theories introduced in \cite{radu2016} (see also \cite{DBLP:conf/lics/MardarePP17,DBLP:conf/lics/BacciMPP18,BacciBLM18,DBLP:journals/entcs/Bacci0LM18}). 
This framework is, roughly speaking, adapting many usual concepts of universal algebra to deal with \emph{quantitative algebras}, which are  structures $ (A, \{ f^A\}_{f\in \Sigma}, d)$ where $(A, \{ f^A\}_{f\in \Sigma})$ is a set with interpretations for the function symbols of a given signature $\Sigma$ and $d$ is a $1$--bounded metric such that $(A,d)\in \Met$ and, for every $f\in \Sigma$ of arity $n$, the map $f^A$ is non--expansive with respect to the metric $d$ and the product metric $d^n$. The familiar concept of equation of universal algebra is replaced by that of \emph{quantitative inference}  $\{x_i  =_{\epsilon_i} y_i\}_{i\in I} \vdash s =_\epsilon t$ where $\epsilon,\epsilon_i\in[0,1]$ and $s,t$ are terms over $\Sigma$. A quantitative algebra  $ (A, \{ f^A\}_{f\in \Sigma}, d)$ satisfies a quantitative inference if, for all interpretations $\iota(x)\in A$ of the variables $x$ as elements of  $A$, the following holds:
\vvm
\begin{equation}\label{quant_inferences}
\textnormal{ if $d\big( \iota(x_i)  , \iota( y_i) \big) \leq \epsilon_i$ for all $i\in I$, then  $\big(\iota(s),\iota(t)\big)\leq \epsilon.$}
\vvm
\end{equation}
The apparatus of equational logic is replaced by a similar apparatus (see \cite[\S 3]{DBLP:conf/lics/BacciMPP18}) for deriving quantitative inferences from a set of generating quantitative inferences (axioms). Soundness and completeness theorems then provide the link between the semantics of quantitative inferences (Property \ref{quant_inferences} above) and derivability in the deductive apparatus.


\begin{definition}[Quantitative Theory of Convex Semilattices]\label{def:convexsemilattices:met}
The quantitative equational theory $\qetcs$ of convex semilattices has signature $\sigcs= (\{ \oplus \}\cup \{ +_p \}_{p\in (0,1)})$ and is defined as the set of quantitative inferences derivable by the following axioms, stated for arbitrary $p,q\in (0,1)$ and 
{$ \epsilon_1,\epsilon_2\in [0,1]$}:

{
\begin{tabular}{@{\hskip -0.05cm}l @{\hskip 2pt}l}
$(A) $ &
$\emptyset \vdash x \oplus ( y \oplus z) =_0 (x \oplus y) \oplus z \ $  \\
$(C)$ & 
$\emptyset \vdash x \oplus y =_0 y \oplus x$\\
$(I)$ & 
$\emptyset \vdash x\oplus x =_0 x$\\
$(A_{p})$ & 
$\emptyset \vdash (x+_qy)+_pz =_0 x+_{pq}(y+_{\frac{p(1-q)}{1-pq}}z)$\\
$(C_{p})$ & 
$\emptyset \vdash x+_py  =_0  y+_{1-p}x$\\
$(I_{p})$ & 
$\emptyset \vdash x+_px  =_0  x$\\
$(D)$ & $\emptyset \vdash x +_p (y \oplus z) =_{0} (x+_p y) \oplus (x +_p z)$\\
$(H)$ & 
$\big\{x_1=_{\epsilon_1} y_1,     x_2=_{\epsilon_2} y_2\big\}\!\vdash x_1 \oplus x_2 =_{\max(\epsilon_1,\epsilon_2)} y_{1}\oplus y_{2}$
\\
(K) & 
$\big\{x_1\!=_{\epsilon_1}\! y_{1},     x_2\!=_{\epsilon_2}\! y_{2}\big\}\!\vdash x_1 \!+_p\! x_2 =_{p\cdot \epsilon_1 + (1-p)\cdot \epsilon_2} y_{1}  \!+_p\! y_{2}$
\end{tabular}
}

\end{definition}
A quantitative algebra $ (A, \{\oplus^A\} \cup \{ +^{A}_p\}_{p \in (0,1)}, d)$ over the signature $\sigcs$ is called a quantitative convex semilattice if it satisfies all quantitative inferences of $\qetcs$ in the sense of Property \ref{quant_inferences}. Similarly, the quantitative theory of pointed convex semilattices $\qetpcs$ is defined over the signature $\sigpcs= (\{ \oplus \}\cup \{ +_p \}_{p\in (0,1)}\cup \{\point\})$ by the same quantitative inferences of Definition \ref{def:convexsemilattices:met}. We denote with $\qacat(\qetcs)$ (respectively $\qacat(\qetpcs)$) the category having as objects quantitative convex semilattices  (respectively quantitative pointed convex semilattices) and as arrows homomorphisms that are non--expansive. 

The presentation theorems of the $\Met$ monads $\lcset$ and $\lcset(+\onemet$) of \cite{MV20} can now be formally stated as the following isomorphisms of categories.
\begin{proposition}\label{thm:main}
\begin{enumerate}
\item The quantitative theory $\qetcs$ of convex semilattices is a presentation of the monad $\lcset$, i.e., $\EM(\lcset)\cong \qacat(\qetcs)$.
\item The quantitative theory $\qetpcs$ of pointed convex semilattices is a presentation of the monad $\lcset(+\onemet)$, i.e., $\EM(\lcset(+\onemet))\cong \qacat(\qetpcs)$.
\end{enumerate}
\end{proposition}
The isomorphism in Proposition \ref{thm:main}.2 is given by a pair of functors
\vvm
\begin{equation}\label{eq:metiso}
\begin{split}
&\hat P: \EM(\lcset(+\onemet))\rightarrow \qacat(\qetpcs) \\
&\hat P^{-1}: \qacat(\qetpcs)\rightarrow \EM(\lcset(+\onemet))
\end{split}
\vvm
\end{equation}
whose definition is similar to the corresponding $\Sets$ isomorphisms $P, P^{-1}$ from \eqref{eq:setiso}, namely 
$\hat P((A,d),\alpha)= (A, \oplus^{\alpha}, \{\pplus p^{\alpha}\}_{p\in (0,1)}, \star^{\alpha} , d)$.
As in the $\Sets$ case, we can use the presentation to identify the free quantitative pointed convex semilattice.
\begin{proposition}\label{prop:freepcsmet}
The free quantitative pointed convex semilattice on $(X,d)$ is  (up to isomorphism):  
$$(\cset(X+\oneset), \oplus^{\cset(+\oneset)}, \{\pplus p^{\cset(+\oneset)}\}_{p\in (0,1)}, \star^{\cset(+\oneset)}, \hk(d + \donemet))$$ 
with operations interpreted as in Proposition \ref{prop:freepcs}.
\end{proposition}

%% file: PresCp1.tex
The functor $\moncb$ maps a set $X$ to the set of non--empty finitely generated convex sets of distributions on $X$ plus an additional element which we denote as $\star\!\in\! \mathbf{1}$. Equivalently, by seeing this additional element as representing the empty set of distributions on $X$, $\moncbp X$ is the set of (possibly empty) finitely generated convex sets of distributions on $X$. This is  the functor of \emph{convex Segala systems} \cite{Seg95:thesis}, which have been widely studied in the literature (see, e.g., \cite{BSV04:tcs,Sokolova11} for an overview) as models of nondeterministic and probabilistic programs.

In this section, we investigate a monad 
whose underlying functor is $\moncb$.\footnote{We let $\moncb$ denote $(+\mathbf{1}) \cset$, not to be confused with $ \cset(+\mathbf{1})$.} Following common practice, the monad $(\moncb, \eta^{\moncb}, \mu^{\moncb})$ will often be simply denoted by $\moncb$, as its underlying functor. Our main result regarding the monad $\moncb$ is a presentation theorem based on the equational theory of convex semilattices with bottom and black--hole. 
The black--hole axiom has been investigated in the literature, e.g., in the context of convex algebras (\cite{SW2018}, \cite{BSS17}) and in axiomatisations of a nondeterministic and probabilistic process algebras (\cite{MOW03}).

\begin{definition}[Theory $\etcsbb$]\label{def:csbhb}
Let $\etpcs$ be the equational theory of pointed convex semilattices. We let $\bot$ (bottom) and \emph{BH} (black--hole) denote the sets of equations
$
\bot = \{x \oplus \star = x\}$ and $BH = \{  x +_p \star = \star \mid p\in (0,1)\}
$, respectively. The equational theory of convex semilattices with bottom and black--hole, denoted by $\etcsbb$, is the theory generated by the set of equations $\etpcs\cup \bot\cup BH$.
\end{definition}

Our monad  on the functor $\moncb$ is defined using Proposition \ref{prop-distlawcomposition}, i.e., we exhibit a monad distributive law of type $\gamma: \cset (+\mathbf{1}) \Rightarrow (+\mathbf{1}) \cset  $ and this gives a monad structure $(\moncb, \eta^{\moncb}, \mu^{\moncb})$ on the composite functor $(+\mathbf{1})\cset$, i.e., on $\moncb$.

\begin{definition}\label{def:gamma}
For every set $X$, the map $\gamma_X: \cset (X+\oneset) \rightarrow \cset(X)+\oneset  $ 
is defined as follows, for any $S\in \cset(X+\oneset)$
\vvm
\[\gamma_X (S) = \begin{cases}
    \left\{ \distr \mid \distr \in S \text{ and } \distr(\point) = 0\right\} &\exists \distr \in S \text{ s.t. } \distr(\point) = 0\\\point &\text{otherwise}
\end{cases}
\vvm
\]
\end{definition}

The above definition can be understood as follows. By viewing  elements  $S\in \cset (X+\oneset)$ as non--empty finitely generated convex sets of subdistributions on $X$ (see Remark \ref{rem:subdistr}) the map $\gamma_{X}$  maps $S$ to its subset (which is convex) consisting of full probability distributions  (i.e., assigning probability $0$ to $\star\in\oneset$). It is possible, however, that such subset is empty, and therefore not in $\cset(X)$. In this case (second clause in the definition)  $S$ is mapped
to the element $\star\!\in\!\oneset$. This two--cases analysis can be further simplified if we see the element $\star\in\oneset$ as representing the empty set $\emptyset$ so that we can simply write:
$$
S \stackrel{\gamma_X}{\longmapsto}  \big\{  \distr\in S  \mid   \distr \textnormal{ is a full distribution, i.e., $\distr(\point) = 0$}  \big\}
$$


\begin{lemma}\label{lem:gammadistrlaw}
    The family of maps $\gamma_X\!:\! \cset(X+\mathbf{1})\! \rightarrow\! \moncbp X$, for $X\in \Sets$, is a monad  distributive law of the monad $\cset$ over the monad $+\oneset$.
    \end{lemma}
\begin{proof}[Sketch]
To verify the commuting diagrams of Definition \ref{def:monadlaw}, we use the fact that $\gamma_X$ commutes with the operations of union and weighted Minkowski sum, which are used in the definition of the multiplication of the monad $\cset$.
\end{proof}

As a corollary, we obtain from Proposition \ref{prop-distlawcomposition} a monad structure  $(\moncb, \eta^{\moncb}, \mu^{\moncb})$ where $\eta^{\moncb} = \eta^\cset$ and  $\mu^{\moncb}_X= \mu^{+\mathbf{1}}_{\cset X}\circ\big((\mu^{\cset}_X+\mathbf{1})+\mathbf{1}\big)\circ \big(\gamma_{\cset X}+\mathbf{1}\big)$. 
Explicitly, $\eta^{\moncb}_X(x) = \{\dirac x \}
$ and, 
for
$S\in \cset(\cset(X)+\mathbf{1})+\mathbf{1}$, if we let $\bpoint$ denote the element of the outer $\oneset$ and $\star$ denote the element of the inner $\oneset$: 
\begin{equation}\label{eq:mubb}
   \mu^{\moncb}_X(S)
      =
%
   \begin{cases}
	\mu^{\cset}_{X}(\gamma_{\cset(X)}(S))
	& \begin{minipage}{0.5\textwidth}
	$S\neq \bpoint \text{ and } \gamma_{\cset(X)}(S) \neq \point$
	\end{minipage}\\
        \point &      		\text{otherwise}
          \end{cases}
%
\end{equation}

We can now state the main result of this section.
\begin{theorem}\label{main:theorem:moncb}
The monad $\moncb$ is presented by the equational theory of convex semilattices with bottom and black--hole, i.e.,
$\EM(\moncb )\cong\acat( \etcsbb)$.
\end{theorem}

The rest of this section outlines the proof of Theorem \ref{main:theorem:moncb} above.
Our first technical step is to prove that the distributive law $\gamma: \cset(+\oneset)\Rightarrow \moncb$ of Lemma \ref{lem:gammadistrlaw}  is also a monad map (see Definition \ref{def:monadmap}) between the monads $\cset(+\oneset)$ and $\moncb$.

\begin{lemma}\label{lem:gammamonadmap}
$\gamma: \cset(+\oneset)\Rightarrow \moncb$ is a monad map.
\end{lemma}
Using Proposition \ref{prop-MonmapFunctor}, the monad map $\gamma$ can be turned into a functor $U^\gamma: \EM(\moncb) \rightarrow \EM(\cset(+\oneset))$ between the Eilenberg-Moore categories of the two monads.
\begin{lemma}\label{embedding_lemma_cbotbh}
 There is a functor $U^{\gamma}: \EM(\moncb) \rightarrow \EM(\cset(+\oneset))$ defined on objects by $(A, \alpha) \mapsto (A, \alpha \circ \gamma_A)$ and acting as identity
 on morphisms. This functor is an embedding, i.e., it is fully faithful and injective on objects.
\end{lemma}

As a consequence of the above lemma, the category  $\EM(\moncb)$
is a full subcategory of  $\EM(\cset(+\oneset))$. Similarly, there is an embedding $\iota: \acat(\etcsbb) \rightarrow \acat(\etpcs)$ of the category of convex semilattices with bottom and black--hole (see Definition \ref{def:csbhb}) into the category of all pointed convex semilattices. This is simply the functor that ``forgets'' that elements of $\acat(\etcsbb)$ satisfy the additional axioms $\bot$ and $BH$. Hence, $\acat(\etcsbb)$ is a full subcategory of   $\acat(\etpcs)$.

Recall that by the presentation of the monad $\cset(+\oneset)$ (Proposition \ref{prop-knownpresentations}), we have the isomorphism
$P: \EM(\cset(+\oneset))\cong \acat(\etpcs) :P^{-1}$ of Equation \eqref{eq:setiso}.
Hence, we have the following diagram:
\[
            \begin{tikzcd}
            {\EM(\cset(+\oneset))} \arrow[d, "P"'] & {\EM(\moncb)} \arrow[l, hook', "U^{\gamma}"'] \\
            {\acat(\etpcs)} \arrow[u, "P^{-1}"', shift right=1.5ex]  & {\acat( \etcsbb)} \arrow[l, "\iota"', hook']
            \end{tikzcd}
\]

In order to prove  $\EM(\moncb ) \cong \acat( \etcsbb)$, which is the statement of Theorem \ref{main:theorem:moncb}, we show that the functors $P$ and $P^{-1}$, when restricted to the subcategories $\EM(\moncb )$ and  $\acat(\etcsbb)$, respectively, are isomorphisms of type:
\begin{align*}
&P: \EM(\moncb )\rightarrow \acat(\etcsbb) \\ 
&P^{-1}: \acat(\etcsbb)\rightarrow \EM(\moncb).
\end{align*}
This amounts to proving the following technical result.
\begin{lemma}\label{moncb-final-lemma}
\begin{enumerate}
    \item Given $(A, \alpha)\in \EM(\moncb)$, which is embedded via $U^\gamma$ to $(A, \alpha\circ \gamma_A) \in \EM(\cset(+\oneset))$, the pointed convex semilattice $P((A, \alpha\circ \gamma_A))$ satisfies the $\bot$ and $BH$ equations, and therefore it belongs to the subcategory $\acat(\etcsbb)$.
    \item Given any $\mathbb{A}\in \acat(\etcsbb)$, which is embedded via $\iota$ to  $\mathbb{A}\in\acat(\etpcs)$, the Eilenberg-Moore algebra $P^{-1}(\mathbb{A})\in \EM(\cset(+\oneset)) $ is in the image of $U^\gamma$, and therefore it belongs to the subcategory  $\EM(\moncb)$.
    \end{enumerate}
    \end{lemma}
\begin{proof}[Sketch]    
For item (1), we need to show that for any $(A, \alpha)\in \EM(\moncb)$ the pointed convex semilattice $P\circ U^{\gamma}((A,\alpha))=P((A,\alpha \circ \gamma_A))=(A, \oplus^{\alpha \circ \gamma_A}, \{\pplus p^{\alpha \circ \gamma_A}\}_{p\in (0,1)}, \star^{\alpha \circ \gamma_A} )$ satisfies the bottom and black-hole axioms. The following steps:
\begin{align*}
&a \oplus^{\alpha \circ \gamma_A} \point^{\alpha \circ \gamma_A} \\
&=(\alpha \circ \gamma_A)(\cc{\dirac{a}, \dirac{\alpha \circ \gamma_A(\{\dirac{\point}\})}}) & \text{def. of $P$}\\
&= \alpha \circ \gamma_A (\cc{\dirac{\alpha \circ \gamma_A(\{\dirac{a}\})}, \dirac{\alpha \circ \gamma_A(\{\dirac{\point}\})}}) &\text{Def. \ref{def:algebra-of-a-monad} eq. (1)}\\
&= \alpha \circ \gamma_A \circ \mu^{\cset(+\oneset)}_{A}(\cc{\dirac{\{\dirac{a}\}}, \dirac{\{\dirac{\point}\}}}) &\text{Def. \ref{def:algebra-of-a-monad} eq. (2)}\\
&= \alpha \circ \gamma_A (\cc{\dirac{a}, \dirac{\point}}) &\text{def. of $\mu^{\cset(+\oneset)}$}\\
&=\alpha(\{\dirac{a}\}) &\text{def. of $\gamma$}\\
&=a &\text{Def. \ref{def:algebra-of-a-monad} eq. (1)}
\end{align*}
prove that the bottom axiom holds. The proof for the black--hole axiom is similar.

For item (2), let $(A,\alpha)= P^{-1}(\mathbb{A})$ with $\mathbb{A}= (A, \oplus^{\alga}, +_p^{\alga}, \point^{\alga})\in \acat(\etcsbb)$. We prove that $(A,\alpha)$ is in the image of $U^\gamma$ by showing that
$\alpha = \alpha' \circ \gamma_A$, where $\alpha':\cset(A)+\oneset\to A$ is defined as follows, for $S\in\cset(A)+\oneset$
\[\alpha' (S) = \begin{cases}
     \alpha(S) &S\in \cset(A) \\
     \alpha(\{\dirac \point\})&S=\point
\end{cases}
\]
For every $S\in \cset(A+\mathbf{1})$ which contains no full distribution (i.e., $\point \in \support(\distr)$ for all $\distr \in S$, see Remark \ref{rem:subdistr}) it holds:
$$\alpha(S) = \oplus^{\alga}_{\distr \in \ub(S)} \alpha(\{\distr\}) = \star^{\alga} =\alpha(\{\dirac\point\})= \alpha'(\gamma_{A}(S))$$
where the second equality follows by the black--hole axiom and idempotency.
Moreover, when $S$ contains at least one full distribution (i.e., $\point \not\in \support(\distr)$ for some $\distr \in S$) we have
\begin{align*}
    \alpha(S) &= \bigoplus_{\distr \in \ub(S)}^{\alga} \alpha(\{\distr\}) &\text{def. of $P^{-1}$}\\
    &= \bigoplus_{\{\distr \in \ub(S) \mid \point \not\in \support(\distr)\}}^{\alga} \alpha(\{\distr\}) &\mbox{by $\bot$ and $BH$}\\
    &= \alpha(\cc{\distr \in \ub(S) \mid \point \not\in \support(\distr)})&\text{def. of $P^{-1}$}\\
    &= \alpha'(\gamma_A(S)) &\text{def. of $\gamma$}.
\end{align*}
This shows $\alpha = \alpha' \circ \gamma_A$. As a final step, it is sufficient to prove that  $(A,\alpha')\in \EM(\cset+\oneset)$ and this concludes the proof.
\end{proof} 

%% file: cdownarrow.tex

In the previous section, we explored a monad structure on the $\Sets$ functor $\cset + \mathbf{1}$ which is natural and important, in the context of program semantics, since it models (simple) convex Segala systems \cite{Seg95:thesis,BSV04:tcs,Sokolova11}. The presentation we obtained for $\cset + \mathbf{1}$ is in terms of convex semilattices with bottom and black--hole. The latter axiom imposes a strong equation $x \pplus p \point=\point$  which is not necessarily adequate in all modelling situations (think, for example, about the equation $\nil +_p P = \nil$ in a probabilistic process algebra). For this reason, 
in this section we consider the equational theory of convex semilattices with bottom (i.e., without the black--hole axiom), 
which has already appeared as relevant in the study of probabilistic program equivalences (see, e.g., \cite{BSV19} for applications in trace semantics) and investigate the monad which is presented by this theory.


\begin{definition}[Theory  $\etpcs$] \label{def:etcsb}
Let $\etpcs$ denote the theory of pointed convex semilattices  (Definition \ref{def:pointedcs:set}) 
and 
let $\bot$ be the equation set $ \{x \oplus \star = x\}$.
We denote with $\etcsb$ 
the theory generated by the set of equations $\etpcs \cup \bot$ and refer to it as the \emph{theory of convex semilattices with bottom}.
\end{definition}

Before formally introducing the monad presented by the theory  $\etcsb$, we develop some useful intuitions.
First, the theory $\etcsb$ is the quotient of $\etpcs$ obtained by adding the $\bot$ axiom. Recall from Proposition \ref{prop-knownpresentations} that $\etpcs$ gives a presentation of the monad $\cset (+ \mathbf{1} )$, which maps a set $X$ to the collection  $ \cset(X + \textbf{1})$ of non-empty finitely generated convex sets of distributions on $X\cup\{\star\}$ (i.e., of subdistributions on $X$). Hence the $\bot$ axiom can be understood as the restriction of $ \cset(X + \textbf{1})$ to sets of distributions containing the $\dirac{\star}$ distribution (equivalently, the subdistribution with mass $0$). Furthermore:

\begin{lemma}\label{lem-intuitioneqncbot}
The following equality is derivable in $\etcsb$, for any $p,q\in(0,1)$:
$
x +_p y = (x +_p y) \oplus  ((x+_q \star) +_p y) \oplus (\star +_p y).
$
\end{lemma}
\begin{proof}\let\qed\relax
Using the identity $z\oplus w = z\oplus w \oplus (z+_q w)$, which is valid in all convex semilattices, we first derive the following, by instantiating $z $ and $w $ with $x+_p y$ and $\star +_p y$, respectively:
\[(x+_p y) \oplus (\star +_p y) = (x+_p y) \oplus (\star +_p y) \oplus ((x+_p y) \pplus q (\star +_p y)).\]
From this we derive the equation
\begin{equation}
\label{eq:downclosureequation2}
(x+_p y) \oplus (\star +_p y) = (x+_p y) \oplus (\star +_p y) \oplus ((x +_q \star) +_p y)
\end{equation}
since $(x+_p y) \pplus q (\star +_p y) = (x +_q \star) +_p y$, which holds even in the theory of convex algebras as can be seen (using the fact that convex algebras present the distribution monad $\dset$) by looking at the two terms as probability distributions on $x,y,\star$, and checking that they are the same probability distribution.
We thereby conclude as follows:
\begin{align*}
x +_p y &\stackrel{\bot}{=} (x \oplus \star) +_p y \\
&\stackrel{D}{=}  (x+_p y) \oplus (\star +_p y) \\
&\stackrel{\eqref{eq:downclosureequation2}}{=} (x+_p y) \oplus (\star +_p y) \oplus ((x +_q \star) +_p y) & \ \ \ \ \ \square
\end{align*}
\end{proof}
Lemma \ref{lem-intuitioneqncbot} can be understood as stating that, under the theory $\etcsb$, if a convex set contains a subdistribution $\distr$ with $\distr(x)=p$, then it also contains any subdistribution $\distrb$ defined as $\distr$ except that $\distrb(x)< p$ (equivalently, $\distrb(x) = qp$  or $\distrb(x) = 0$).
This leads to the following definition. 

\begin{definition}[$\bot$--closed convex set]
Let $X$ be a set and let $S\in \cset(X + \oneset)$. We say that $S$ is \emph{$\bot$--closed} if
$\left\{ \distrb \in \dset(X+\mathbf{1}) \mid \forall x \in X, \distrb(x) \leq \distr(x)\right\} \subseteq S$ for all $\distr \in S$.
We denote with $\cset^{\downarrow}(X)\subseteq  \cset(X + 1)$ the set of non-empty finitely generated $\bot$--closed convex sets of subdistributions on $X$.
\end{definition}

We now give a useful alternative characterisation of $\bot$--closed sets, by defining a homomorphism of pointed convex semilattices $\bcl_X:\cset(X+\mathbf{1}) \to \cset(X+\mathbf{1})$. Recall that $\cset(X+\mathbf{1})$ is (the carrier of) the free pointed convex semilattice on $X$ (Proposition \ref{prop:freepcs}).
Hence, for any pointed convex semilattice $A$ and for any $f:X\to A$, there is a unique pointed convex semilattice homomorphism extending $f$.
\begin{definition}\label{def:K}
Let $X$ be a set and $f:X \rightarrow \cset(X+\mathbf{1})$ be defined as $f(x) = \conv \{\dirac{x}, \dirac\point\} = \{p\, x+ (1-p) \,\point \mid p \in [0,1]\}$. We denote with $\bcl_X:\cset(X+\mathbf{1}) \to \cset(X+\mathbf{1})$ the unique pointed convex semilattice homomorphism extending $f$.
\end{definition}

%
Theorem \ref{thm:Kclosure} below states that $\bcl_X$ maps a set $S\in \cset(X+\mathbf{1})$ to its $\bot$--closure. 


\begin{theorem}\label{thm:Kclosure}
Let $X$ be a set. 
Then $\bcl_X$ is the $\bot$--closure operator, i.e., for any $S\in \cset(X + \oneset)$, $\bcl_X(S)$ is the smallest $\bot$--closed set containing $S$.
\end{theorem}
It follows from Theorem \ref{thm:Kclosure} that, by restricting its codomain to the image, $\bcl_X$ defines a surjective function of type $\cset(X+\oneset)\to \moncdp X$, which we  also denote by $\bcl_X$.
We can now define a functor and a monad on  $\moncdp X$, with unit given by $ \bcl \circ \eta^{\cset(+\oneset)}$ and multiplication given by $\mu^{\cset(+\oneset)}$ restricted to $\bot$--closed sets.


\begin{definition}[Monad $\cset^{\downarrow}$]\label{def:set:downarrow}
The \emph{non-empty finitely generated $\bot$--closed convex sets of subdistributions} 
monad $(\cset^{\downarrow}, \eta^{\cset^{\downarrow}}, \mu^{\cset^{\downarrow}})$ on $\Sets$ is defined as follows. The functor $\cset^{\downarrow}$ maps a set $X$ to $\cset^{\downarrow}(X)$ 
and maps a morphism $f:X\rightarrow Y$ to $\cset^{\downarrow}(f): \cset^{\downarrow}(X)\rightarrow \cset^{\downarrow}(Y)$ defined as the restriction of  $\cset( f + \mathbf{1})$ to $\bot$--closed sets.
The unit 
is defined as 
$\eta^{\cset^\downarrow}(x)=\bcl_X(\{\dirac{x}\}) = \{p \,x+ (1-p) \,\point \mid p \in [0,1]\}$. 
The multiplication $ \mu^{\cset^\downarrow}$
is defined as
the restriction of $\mu^{\cset(+\oneset)}$ to $\bot$--closed sets. 
\end{definition}


%
%

\begin{theorem}\label{thm:downarrowmonad}
    The triple $(\cset^{\downarrow}, \eta^{\cset^\downarrow}, \mu^{\cset^\downarrow})$ is a monad.
\end{theorem}

We can now state the main result of this section.

\begin{theorem}\label{main:theorem:moncd}
The monad $\moncd$ is presented by the equational theory of convex semilattices with bottom, i.e.,
 $\EM(\moncd) \cong \acat( \etcsb)$.
\end{theorem}


The structure of the proof of Theorem \ref{main:theorem:moncd} is similar to that of Theorem \ref{main:theorem:moncb}. In the rest of this section, we outline the main steps. First, we show that the operation $\bcl_X$ of $\bot$--closure can be seen as a monad map from $\cset(+\mathbf{1})$ to $\moncd$.
\begin{lemma}\label{lem:kmonadmap}
    The family of functions $\bcl_X: \cset(X+\mathbf{1}) \rightarrow \moncdp X$, for $X\in \Sets$, is a monad map between the monads $\cset(+\oneset)$ and $\moncd$.
\end{lemma}

Using Proposition \ref{prop-MonmapFunctor}, the monad map $\bcl$ gives a functor $U^\bcl: \EM(\moncd) \rightarrow \EM(\cset(+\oneset))$ between the Eilenberg-Moore categories. 

\begin{lemma}\label{embedding_lemma_cbot}
   There is a functor $U^{\bcl}: \EM(\moncd) \rightarrow \EM(\cset(+\mathbf{1}))$ defined on objects by $(A, \alpha) \mapsto (A, \alpha \circ \bcl_A)$ and acting as identity on morphisms.  This functor is an embedding, i.e., it is fully faithful and injective on objects.
\end{lemma}

As a consequence of the above lemma, the category  $\EM(\moncd)$ 
is a full subcategory of  $\EM(\cset(+\oneset))$. Analogously, we have an embedding of categories $\iota: \acat(\etcsb) \rightarrow \acat(\etpcs)$ of the category of convex semilattices with bottom (see Definition \ref{def:etcsb}) into the category of pointed convex semilattices, defined as the functor which forgets that the $\bot$--axiom is satisfied. Hence $\acat(\etcsb)$ is a full subcategory of   $\acat(\etpcs)$.

Let $P: \EM(\cset(+\mathbf{1})) \cong \acat(\etpcs): P^{-1}$ be the isomorphisms of categories from Proposition \ref{prop-knownpresentations}. 
We prove $\EM(\moncd) \cong \acat( \etcsb)$ (i.e., Theorem \ref{main:theorem:moncd}) by showing that the functors $P$ and $P^{-1}$, when respectively restricted to the subcategories $\EM(\moncd )$ and  $\acat(\etcsb)$, are isomorphisms of type:
$$P: \EM(\moncd )\rightarrow \acat(\etcsb) \ \ \ \ \ \ \ P^{-1}: \acat(\etcsb)\rightarrow \EM(\moncd).
$$
This amounts to proving the following result.
\begin{lemma}\label{moncd-final-lemma}
\begin{enumerate}
    \item Given any $(A, \alpha)\in \EM(\moncd)$, which is embedded via $U^\bcl$ to $(A, \alpha\circ \bcl_A) \in \EM(\cset(+\oneset))$, the pointed convex semilattice $P((A, \alpha\circ \bcl_A))$ satisfies the $\bot$ equation, and therefore it belongs to $\acat(\etcsb)$.
    \item Given any $\mathbb{A}\in \acat(\etcsb)$, which is embedded via $\iota$ to  $\mathbb{A}\in\acat(\etpcs)$, the Eilenberg--Moore algebra $P^{-1}(\mathbb{A})\in \EM(\cset(+\oneset)) $ is in the image of $U^\bcl$, and therefore it belongs to the subcategory  $\EM(\moncd)$.
    \end{enumerate}
    \end{lemma}
\begin{proof}[Sketch]   
Item (1) follows similarly to Lemma \ref{moncb-final-lemma}.1.
The proof of item (2) requires to show that the Eilenberg--Moore algebra $(A,\alpha)=P^{-1}(\mathbb{A})$ satisfies $\alpha=\alpha'\circ \bcl_A$, 
with $\alpha':\moncd(A)\to A$ defined as the restriction of $\alpha$ to $\bot$--closed sets.
Let $S = \conv (\bigcup_{0\leq i \leq n} \{\distr_{i}\}) \in \cset(A+\mathbf{1})$, with $\bigcup_{0\leq i \leq n} \{\distr_{i}\}$ the unique base for $S$. In order to prove that $\alpha'(\bcl_A(S)) = \alpha(S)$, we use two results. 
First, we identify a finite base for the set $\bcl_A(S)$ as follows:
\begin{equation}
\label{eq:botfinitebase1}
\bcl_A(S)= \conv\Big(\bigcup_{0\leq i \leq n} \big(\bigcup_{F\subseteq \support(\distr_{i})\backslash \{\star\}}\{\distr_{i|_{F}}\}\big)\Big)
\end{equation}
    where for any $\distr$ and $F\subseteq \support(\distr)\backslash \{\star\}$ we define
    $$\distr_{|_{F}}=(\sum_{a\in F} \distr(a) a) + (1-(\sum_{a\in F} \distr(a))) \point.$$
Then we show that the following equation, with variables $x_{i}$ universally quantified and with $\sum_{0\leq i\leq n} p_{i}=1$, is derivable in the theory $\etcsb$:
\begin{equation}
\label{eq:botfinitebase2}
\bigpplus_{0\leq i\leq n} p_{i} x_{i}= \bigoplus_{F\subseteq \{1,...,n\}} \Big((\bigpplus_{i\in F} p_{i} x_{i}) + (1-(\sum_{i\in F} p_{i}))\point \Big)
\end{equation}
Now, we can apply the definition of $P^{-1}$ and the characterization in Equation \eqref{eq:botfinitebase1}  to derive that $\alpha'(\bcl_A(S))$ is equal to the interpretation in $\alga$ of the term:
$$\bigoplus_{0\leq i \leq n}\Big(\bigoplus_{F \subseteq \support(\distr_{i})\backslash\{\star\}} \big((\bigpplus_{a\in F} \distr_{i}(a) a) + (1-(\sum_{a\in F} \distr_{i}(a)))\point\big)\Big)$$
and then using the Equation \eqref{eq:botfinitebase2} of $\etcsb$
we derive that this is in turn equal to the interpretation in $\alga$ of the term
\[\bigoplus_{0\leq i \leq n} \big(\bigpplus_{b\in \support (\distr_{i})} \distr_{i}(b) b\big)\]
with $b$ ranging over $A\cup\{\point\}$. By the definition of $P^{-1}$, the interpretation in $\alga$ of this term is equal to $\alpha(S)$.
\end{proof}

%% file: presentationsinmet.tex

In Section \ref{sec:cplusone}, we investigated a $\Sets$ monad whose underlying functor is $\moncb$ and obtained its presentation in terms of convex semilattices with bottom and black--hole. In Section \ref{sec:cdownarrow}, we investigated the $\Sets$ monad $\moncd$ and proved that it is presented by the theory of convex semilattices with bottom. In this section, we investigate similar questions but in the category $\Met$ of 1--bounded metric spaces.
First, take the functor $\lcset+\onemet$, which is the $\Met$ lifting of  the functor $\cset+\oneset$ obtained with the Hausdorff--Kantorovich lifting. 
%
Is there a $\Met$ monad whose underlying functor is $\lcset+\onemet$?
We do not answer the question in full generality, 
yet we provide some negative 
results by showing that
any such monad:
\begin{enumerate}
\item cannot have a multiplication defined as the one
of the $\Sets$ monad $\cset + \oneset$, and
\item cannot be presented by the quantitative theory of convex semilattices with bottom and black--hole, since this theory is trivial.
\end{enumerate}

Secondly, the question is to find the $\Met$ monad presented by the quantitative theory of convex semilattices with bottom. 
In this case, we are successful and we show that this monad is exactly the lifting of the $\Sets$ monad $\moncd$ to $\Met$ via the Hausdorff-Kantorovich distance.
%
%


\subsection{Negative results on monad structures on $\lcset+\onemet$}
\label{sec:met:cplusone}

Recall from Section \ref{sec:cplusone} that the $\Sets$ monad $\moncb$ is obtained from the distributive law $\gamma:\cset(+\oneset) \Rightarrow \moncb$. We now show that $\gamma$ fails to be non--expansive when its domain and codomain are equipped with the Hausdorff--Kantorovich lifted metrics. This implies that most of the machinery developed in Section \ref{sec:cplusone} to obtain the $\Sets$ monad $\moncb$ and its presentation is not applicable in $\Met$.

\begin{lemma}\label{lem:gammane}
There is a metric space $(X,d)$
such that $\hat \gamma_{(X,d)}: ({\cset(X+\oneset)}, \hk(d + \donemet)) \to (\cset (X)+\oneset, \hk(d) +  \donemet) $, defined as the $\Sets$ function $\gamma_{X}$ from Definition \ref{def:gamma},  is not non-expansive. 
\end{lemma}
\begin{proof}
Let $X$ be non--empty and take $(X,d)$ with $d$ the discrete metric. Consider $S_1,S_2\in \cset(X+\oneset)$ defined as $S_{1}=\{ \onehalf\,x +\onehalf\, \star\}$ and $S_{2}= \{\dirac x\}$. Then 
$$(\hk(d)+ \donemet)(\gamma_{X}(S_{1}), \gamma_{X}(S_{2}))= (\hk(d)+ \donemet)(\star, \{\dirac x\})=1$$
which is strictly greater than 
$\frac{1}{2}= \hk(d+ \donemet)(S_{1}, S_{2}).$
\end{proof}

As the multiplication $\mu^{\moncb}$ of the $\Sets$ monad $\moncb$ is defined using $\gamma$ (see Equation \eqref{eq:mubb} after Lemma \ref{lem:gammadistrlaw}), this counterexample can be adapted to show that also the multiplication $\mu^{\moncb}$ is not non--expansive. 
Hence, no $\Met$ monad $M$ whose underlying functor is $\lcset + \onemet$ can have a multiplication $\mu^M$  which, once the metric is forgotten, coincides with $\mu^{\moncb}$.
Furthermore, no such monad $M$ can be presented by the quantitative theory of convex semilattices with bottom and black--hole ($\qetcsbb$). This is because $\qetcsbb$ is trivial in the sense that the quantitative inference $\emptyset \vdash x=_{0} y$, expressing that all elements are at distance $0$, is derivable from the axioms or, equivalently, any quantitative algebra in $\qacat(\qetcsbb)$ has the singleton metric space $\onemet$ as carrier.


\begin{definition}\label{def:qetcsbb}
The quantitative theory $\qetcsbb$ of \emph{quantitative convex semilattices with bottom and black--hole} has the signature $\sigpcs$ of pointed convex semilattices and is generated by the quantitative inferences $\qetpcs \cup \qbot \cup \qbh$, with  $\qbot=\{\emptyset \vdash x\oplus \star =_{0} x\}$ and $\qbh=\{\emptyset\vdash x \pplus p \star =_{0} \star \mid p\in (0,1)\}$.
\end{definition}

\begin{theorem}\label{thm:bhtrivial}
Any quantitative equational theory $\qet$ containing $\qetpcs$ and $\qbh$ is trivial, i.e., the quantitative inference $\emptyset\vdash x=_{0} y$ is derivable in $\qet$.
\end{theorem}
\begin{proof}
The deductive system of quantitative equational logic is the one of \cite[\S 3]{DBLP:conf/lics/BacciMPP18}. Let $p\in (0,1)$. From the axioms $\vdash x=_{0} x $ and $\vdash x=_{1} \star$, we derive by the (K) rule $\vdash x\pplus p \star =_{1-p} x\pplus p x$. By $(I_{p})$, we have $\vdash x\pplus p x =_{0} x$, so by triangular inequality we derive $\vdash x\pplus p \star =_{(1-p)} x$. Now, by the $\qbh$ axiom ($\vdash x\pplus p \star =_{0} \star$), symmetry and triangular inequality, we have $\vdash \star=_{(1-p)} x$. Since $p \in (0,1)$ was arbitrary, we have equivalently derived that $\vdash \star=_{p} x$ belongs to $\qet$. For any $y$, we analogously obtain $\vdash \star =_{p} y$. Then, by symmetry and triangular inequality we derive $\vdash x=_{p} y$ for all $p\in (0,1)$, and by (Max) we have $\vdash x=_{\epsilon} y$ for all $\epsilon >0$. We conclude by applying (Arch) $\{x=_{\epsilon} y\}_{\epsilon>0} \vdash x=_{0} y$.
\end{proof}



\subsection{Lifting of $\moncd$ to metric spaces and its presentation}
\label{sec:met:cdownarrow}

The $\Sets$ monad $\moncd$ (Definition \ref{def:set:downarrow}) is obtained using $\bcl \circ \eta^{\cset(+\oneset)}$ as unit, where  $\bcl_X:\cset(X+\oneset)\rightarrow\cset^\downarrow(X)$ is the operation of $\bot$--closure (see Definition \ref{def:K}), and as multiplication the restriction of $\mu^{\cset(+\oneset)}$ to $\bot$--closed sets.
We can give a similar definition in the category $\Met$ using the unit and multiplication of the $\Met$ monad $\lcset(+\onemet)$ and the natural transformation $\bcl$, provided the latter exists in $\Met$, i.e., it is non--expansive. 


\begin{lemma}\label{lem:met:kne}
For any metric space $(X,d)$, 
the function 
$\hat \bcl_{(X,d)}: ({\cset(X+\oneset)}, \hk(d + \donemet)) \to (\cset (X+\oneset), \hk(d) + \donemet)$, 
defined as the $\Sets$ function $\bcl_{X}$ from Definition \ref{def:K},  is non-expansive.
\end{lemma}
\begin{proof}
By Definition  \ref{def:K}, $\bcl_{X}$ is the unique pointed semilattice homomorphism extending $f:X\to \cset(X+\oneset)$, with $f(x)=\conv(\{\dirac x, \dirac \point\})$. The function 
$$\hat f: (X,d) \to (\cset(X+\oneset), \hk(d + \donemet)),$$
defined as $f$ on $X$, is easily seen to be an isometry, and thus non-expansive. Hence, $\hat f$ is a morphism in $\Met$.
Now, recall that $((\cset(X+\oneset), \hk(d+\donemet))$ is the free quantitative pointed convex semilattice on $(X,d)$ and since the unique extension of $\hat{f}$ is also a pointed convex semilattice homomorphism, its action on sets must coincide with $\bcl_X$. Hence, $\hat \bcl_{(X,d)}$ is the unique quantitative pointed convex semilattice homomorphism extending $\hat f$. Therefore, $\hat \bcl_{(X,d)}$ is a morphism in $\Met$, which means that it is non-expansive.
\end{proof}


Based on Lemma \ref{lem:met:kne} and on $\lcset(+\onemet)$ being a monad in $\Met$, we can obtain the following result in a way similar to Theorem \ref{thm:downarrowmonad}.


\begin{definition}[Monad $\lmoncd$ in $\Met$]\label{def:met:downarrow}
The monad $(\lmoncd, \eta^{{{\ldownarrow}}}, \mu^{{{\ldownarrow}})})$ in $\Met$ is defined as follows. The functor $\lmoncd$ maps a metric space $(X,d)$ to $(\moncd(X),\hk(d+\donemet))$. The action of the functor on arrows, the unit and the multiplication are defined as those of the $\Sets$ monad $\moncd$.
\end{definition}

\begin{theorem}\label{thm:downarrowmonadmet}
    The triple $(\lmoncd, \eta^{\ldownarrow}, \mu^{\ldownarrow})$ is a  $\Met$ monad.
\end{theorem}

We now introduce the quantitative equational theory of quantitative convex semilattices with bottom and state the main result of this section.

\begin{definition}\label{def:qetcsb}
The quantitative equational theory $\qetcsb$ of quantitative convex semilattices with bottom is the quantitative equational theory generated by the set of quantitative inferences $\qetpcs \cup \qbot$, with
$\qbot=\{\vdash x\oplus \star =_{0} x\}.$
\end{definition}

\begin{theorem}\label{thm:maincdownarrowmet}
The monad $\lmoncd$ is presented by the quantitative equational theory of quantitative convex semilattices with bottom, i.e.,
$\EM(\lmoncd)\cong\qacat( \qetcsb)$.
\end{theorem}
The proof of Theorem \ref{thm:maincdownarrowmet} is similar to that of Theorem \ref{main:theorem:moncd}. First, we identify $\EM(\lmoncd)$ and $\qacat( \qetcsb)$ as full subcategories of $\EM(\lcset(+\onemet))$ and $\qacat( \qetpcs)$, respectively. Then, we obtain the isomorphism $\EM(\lmoncd)\!\cong\!\qacat( \qetcsb)$ by restricting the isomorphism $\EM(\lcset(+\onemet))\!\cong\!\qacat( \qetpcs)$ of Proposition \ref{thm:main}.2.

%% file: examples.tex
\label{examples:section}


The results presented in this paper are summarised in Table \ref{tab:summary}.
\begin{table}
\centering
\begin{tabular}{| c | c |  }
\hline
$\Sets$ Monad & Eq. Theory\\[2pt]
\hline
$\cset(X+\oneset)$ & $\etpcs$\\[2pt]
\hline
$\cset(X)+\oneset$ & $\etcsbb$ \\[2pt]
\hline
$\cset^\downarrow(X)$ & $\etcsb$ \\[2pt]
\hline
\end{tabular}
\;
\begin{tabular}{| c | c |  }
\hline
$\Met$ Monad & Quant. Eq. Theory\\[2pt]
\hline
$\lcset(X+\onemet)$ &  $\qetpcs$ \\[2pt]
\hline
trivial & $\qetcsbb$ \\[2pt]
\hline
$\lmoncd(X)$ & $\qetcsb$ \\[2pt]
\hline
\end{tabular}
\caption{Summary of results}
\label{tab:summary}
\end{table}
%
The equational theories $\etpcs$, $\etcsbb$ and $\etcsb$, or closely related variants, have appeared in several works on  mathematical formalisations of semantics of programming languages combining probability and nondeterminism, with applications including: SOS process algebras \cite{DArgenioGL14,GeblerLT16,Bartels02}, axiomatisations of bisimulation  \cite{MOW03,BS01}, trace semantics  \cite{BSV19} and up--to techniques \cite{BSV19}. The functor $\cset+\oneset$, defining the well--known class of \emph{convex Segala systems} \cite{Seg95:thesis,Sokolova11,BSV04:tcs,DBLP:conf/fossacs/Mio14}, has also been considered in many works. The value of our contribution is to have established the mathematical foundation for unifying, modifying and extending several of these works. The goal of this section is to illustrate the general usefulness of our results by means of some examples. 



\begin{figure}
\centering 
\fbox{
\begin{minipage}{7.5cm}
\vspace{0.3cm}
$$
\AxiomC{}	
\LeftLabel{action}
\UnaryInfC{$a.P \trans P$}
\DisplayProof 
\qquad
\AxiomC{}	
\LeftLabel{$\nil$}
\UnaryInfC{$ \nil \trans \star$}
\DisplayProof 
$$
\vspace{0cm}
$$
\AxiomC{ $P_{1} \trans t_{1}$}	
\AxiomC{ $P_{2} \trans t_{2}$}
  \LeftLabel{$\ndplusccs$}
  \BinaryInfC{ $P_{1} \ndplusccs P_{2} \trans  t_{1} \oplus t_{2}$ }
\DisplayProof
$$
\vspace{0cm}
$$
\AxiomC{ $P_{1} \trans t_{1}$}	
\AxiomC{ $P_{2} \trans t_{2}$}
  \LeftLabel{$\pplusccs p$}
  \BinaryInfC{ $P_{1} \pplusccs p P_{2} \trans  {t_{1}} \pplus p {t_{2}}$ }
\DisplayProof
$$ 
\vspace{0cm}
\end{minipage}
}
\caption{Operational semantics.}
\label{fig:pccs}
\end{figure}
We start by introducing a minimalistic process algebra with both nondeterministic and probabilistic choice. Process terms are defined by the grammar:
\vvm
 $$P ::= \nil\ | \ a.P\  |\  P_{1} \ndplusccs P_{2}\ |\ P_{1} \pplusccs p P_{2} 
 \vvm
 $$
 for $p\in (0,1)$.
 We let $\pccs$ denote the set of all process terms. Intuitively, $\nil$ is the terminating process, $ a.P$ does an $a$--action and then behaves as $P$, and $\ndplusccs$ and  $\pplusccs p {}{}$ are (convex) nondeterministic and probabilistic choice operators, respectively. We assume that $a.(\_)$ has binding priority over the other language operators and, for $n\geq 0$, we define $a^{n}. P$ inductively as $a^{0}. P= P$ and $a^{n+1}. P= a. (a^{n}.P)$.
For the sake of simplicity, we just consider a single action label $a$. Variants with multiple labels can be easily given. The transition function is defined as the map: $$\tau : {\pccs } \rightarrow \terms  \pccs {\sigpcs}$$
inductively defined in Figure \ref{fig:pccs}, assigning to each process a term in $\terms  \pccs {\sigpcs}$, i.e., a term in the signature of pointed convex semilattices built using process terms as generators. If $\tau(P)= t$ we say that  $t$ is the \emph{continuation} of $P$ and write, with infix notation, that $P \trans t$.

\begin{figure}
\centering
\fbox{
\begin{minipage}{7.5cm}
\centering
\begin{tikzpicture}[thick]


\matrix[matrix of nodes, row sep= 0.5cm, column sep=.2cm,ampersand replacement=\&, every node/.style={scale=0.8}]
{ 	
			\&\node (x) {$P_{1}$};		\\
			\& \node (x1) {$a^{2}.\nil$};	 \\
			\& \node (x3) {$a.\nil$};	 \\
			\& \node (z) {$\nil$};	 \\
	};
\draw[-latex] (x) to node[right] {$ $} (x1);
\draw[-latex] (x1) to node[left] {$ $} (x3);
\draw[-latex] (x3) to node[left] {$ $} (z);


\begin{scope}[xshift=3.5cm]
\matrix[matrix of nodes, row sep= 0.5cm, column sep=-0.6cm,ampersand replacement=\&, every node/.style={scale=0.8}]
{ 	
	\&					\&\node (x) {$P_{2}$};		\\
	\&					\&\node (d) {}; 				\& \\[-0.2cm]
	\&\node (x1) {$\quad\qquad a^{2}.\nil\qquad \quad$};			\& 						\& \node (x2) {$a^{2} .\nil \ndplusccs a. (a. \nil \ndplusccs \nil)$};	 \\[-0.2cm]
  	\&					\&\node (x3) {$a.\nil$};		\& 						\&	\node (x4) {\hspace{-0.6cm}$a.\nil\ndplusccs \nil$};			 	 \\[-0.2cm]
	\&					\&						\& 	\node (z) {$\nil$};		\&				 	 \\[-0.2cm]
	};
\draw[-latex] (x) to node[right] {$ $} (d);
\draw[-latex] (x1) to node[left] {$ $} (x3);
\draw[-latex] (x2) to node[left] {$ $} (x3);
\draw[-latex] (x2) to node[right] {$ $} (x4);
\draw[-latex] (x3) to node[left] {$ $} (z);
\draw[-latex] (x4) to node[right] {$ $} (z);
\node[below left=0.03cm and -1.15cm of x2] (x2c) {};
\draw[thick,red] (x2c) arc (-120:-60:1cm);

\draw[dotted,->] (d) to node[above] {$\frac 1 2$}  (x1);
\draw[dotted,->] (d) to node[above] {$\frac 1 2$}  (x2);


\end{scope}

\end{tikzpicture}
\end{minipage}}
\caption{$\cset+1$ transition semantics of $P_{1},P_{2}$.}
\label{fig:ts}
\end{figure}
%

%
%
Our presentation results guarantee that the monads $\cset(+\oneset), \cset +\oneset$ and $\cset^\downarrow$ are respectively  isomorphic to the term monads given by the theories $\etpcs$, $\etcsbb$ and $\etcsb$. This implies that, once we fix one such theory $\et$, we have an isomorphism between elements of $\terms \pccs \sigpcs_{\!/\et}$, i.e., process continuations quotiented by $\et$, and elements of $F(\pccs)$, for $F$ respectively denoting the functor $\cset(+\oneset), \cset +\oneset $ and $\cset^\downarrow$.
Hence, the interesting point about the definition of $\tau$ is that, depending on which equational theory is applied to continuations, we obtain different transition semantics (i.e., $\Sets$ coalgebras): 
\vvm
$$
\tau_{F}: \pccs\rightarrow F(\pccs)  
\vvm
$$
for $F$ a functor in $\{ \cset(+\oneset), \cset +\oneset, \cset^\downarrow \}$.

For instance, by choosing the theory $\etcsbb$ of convex semilattices with bottom and black--hole we obtain $F=  \cset +\oneset$, which is  the well--studied functor of \emph{convex Segala systems} \cite{Seg95:thesis},
and our process algebra can be considered as the core of the calculus of \cite{MOW03} (see also \cite{BS01}). Following the discussion in \cite{MOW03}, the black--hole axiom\footnote{In \cite{MOW03}, the black--hole axiom appears as axiom (D2).} can be understood as follows from the point of view of program semantics: in a process $P \pplusccs p  Q$ the nondeterministic choices of $P$ and $Q$ are resolved first, and then combined probabilistically; if one of them offers no choices at all (i.e., if $P=\nil$ or $Q=\nil$), then $P\pplusccs p Q$ is also inactive (i.e., $P\pplusccs p Q = \nil$).
Adopting this semantics, the continuation of  $\nil$ is the element $\point\in\oneset$ (which can be identified with the emptyset $\emptyset$, see the remark after Definition \ref{def:gamma}) and the semantics of the process terms 
\begin{align*}
&P_1 = \,a^{3}.\nil \ndplusccs \nil\\
&P_2 = \,a^{3}.\nil \pplusccs \onehalf a. (a^{2} .\nil \ndplusccs a. (a. \nil \ndplusccs \nil))
\end{align*}
can be depicted as in Figure \ref{fig:ts}.  As customary, we omit the dotted probabilistic arrow when the probability is $1$ and only depict some
 of the reached distributions in the convex set (the red arc indicates their convex closure).
We can now reason about process behaviours using standard definitions. For instance, the following is one way (see, e.g., \cite[\S 2]{Sokolova11} for a detailed exposition) of defining behavioural equivalence coalgebraically.   

\begin{definition}\label{def:bis}
Let $F$ be a $\Sets$ endofunctor,  $c:X \to F(X)$ a coalgebra for $F$, $R\, \subseteq X\times X$ an equivalence relation, $X/_{R}$ the collection of $R$-equivalence classes and $q_{R}: X\rightarrow X/_{R}$ the quotient map. We say that $R$ is a behavioural equivalence if for all $(x,y)\in R$ it holds that
$\big(F(q_{R}) \circ c\big)(x) = \big(F(q_{R}) \circ c\big)(y).$
Elements $x,y\in X$ are \emph{behaviourally equivalent} (written $x\simeq_F y$) if there is a behavioural equivalence $R$ such that $(x,y)\in R$.  
\end{definition}
If, for instance, we let $F=\cset + \oneset$, the above definition coincides with the standard convex bisimulation equivalence of Segala (see \cite{Seg95:thesis,BSV04:tcs,Sokolova11}).

Our main results for $\Sets$ monads presentations via equational theories, summarised in Table \ref{tab:summary}, allow for the definition of a simple sound and complete proof system for behavioural equivalence of $F$--coalgebras described by process terms, for $F \in \{ \cset(+\oneset), \cset +\oneset, \cset^\downarrow \}$. For a chosen equational theory $\et\in  \{\etpcs, \etcsbb, \etcsb \}$, the proof system allows for the derivation of judgments of the form $P \stackrel{\et}{\sim} Q$ using the simple deductive rule:
\begin{equation}\label{deduction_rule}
\infer{P \stackrel{\et}{\sim} Q}{ (\,\et \ \uplus \ \stackrel{\et}{\sim}\,) \ \, \vdash_{E} \  \tau(P) = \tau(Q)}
\end{equation}
In other words, in order to derive the equality between the process terms $P$ and $Q$, it is sufficient to show that their continuations $\tau(P)$ and $\tau(Q)$ are provably equal (in the apparatus of equational logic) from the axioms of the chosen equational theory $\et$ extended with the set of axioms $\{P_i = Q_j \mid P_i \stackrel{\et}{\sim} Q_j \}$, representing pairs of process terms provably equal in the proof system which may appear as generators in the terms $\tau(P)$ and $\tau(Q)$. Formally, we have the following inductive definition.

\begin{definition}
The relation $ {\stackrel{\et}{\sim}}\subseteq \pccs \times \pccs$ of \emph{derivability in the proof system} is defined as the smallest equivalence relation $R\subseteq \pccs \times \pccs$ such that
$$
R = \{ (P,Q) \mid  (\et \ \uplus \ R) \ \vdash_{E} \  \tau(P) = \tau(Q). \}
$$
where $\vdash_{E}$ is derivability in equational logic from a set of axioms and $(\et \ \uplus \ R)$ is the disjoint union of $\et$ and of the set of axioms $\{ P=Q \mid P \,R\, Q\}$.
\end{definition}
Before providing some simple illustrative examples of usage of the proof system, we state the following result.
\begin{theorem}[Soundness and completeness]
\label{thm:soundcomp}
For  $\et\in  \{\etpcs, \etcsbb, \etcsb \}$, let  $F \in \{ \cset(+\oneset), \cset +\oneset, \cset^\downarrow \}$ be the corresponding functor (see Table \ref{tab:summary}).  The following holds for all $P,Q\in Proc$:
\begin{center}
$P\simeq_F Q$ if and only if $P \stackrel{\et}{\sim}Q$.
\end{center}
\end{theorem}
\begin{proof}[Sketch]
Both directions are proved using our presentation results, in the form of an isomorphism between the monads $\cset(+\oneset), \cset +\oneset$ and $\cset^\downarrow$ and the term monads given by the theories $\etpcs$, $\etcsbb$ and $\etcsb$, respectively.
For soundness (right to left implication), we show that $ \stackrel{\et}{\sim}$ is a behavioural equivalence (Definition \ref{def:bis}). For completeness  (left to right implication), we firstly inductively stratify process terms by the complexity of their continuations as follows: $\pccsn 0  = \emptyset$ and 
$\pccsn {n+1} =  \{P\in \pccs \mid \tau(P) \in \terms {\pccsn n} {\sigpcs}\} \cup \pccsn n$ 
 and observe that $\pccs=\bigcup_{n\geq 0} \pccsn n$. Secondly, we prove by induction on $n$ that completeness holds for all terms in $\pccsn n$.
%
\end{proof}

Let us now consider a concrete example by fixing the equational theory $\et=\etcsbb$ and the functor $F=\cset + \oneset$ of Segala systems. Our goal is to prove, using the proof system, that the process terms  $P_1$ and $P_2$ are behaviourally equivalent. 

In order to establish that $P_1 \stackrel{\et}{\sim} P_2$ we need to derive, in equational logic, that  $\et \ \uplus \  \stackrel{\et}{\sim} \ \vdash_{E} \  \tau(P) = \tau(Q)$ (see Deductive Rule \eqref{deduction_rule} above). This can be done as follows:
\begin{align*}
\tau(P_{1})
&  =  a^{2}.\nil\oplus \star &\text{definition of $\tau$}\\
&  =  a^{2}.\nil & \text{$\bot$ axiom}\\
&  =  a^{2}.\nil  +_\frac{1}{2}  a^{2}.\nil  & \text{$I_{p}$ axiom}\\
&  = a^{2}.\nil  +_\frac{1}{2}  P_{3}   & a^{2}.\nil   \stackrel{\et}{\sim} P_{3}\\
&  = \tau(P_{2}) &\text{definition of $\tau$}
\end{align*}
where $P_3$ is the process term $a^{2} .\nil \ndplusccs a. (a. \nil \ndplusccs \nil)$. Note that the above derivation includes the equational step ($a^{2}.\nil   \stackrel{\et}{\sim} P_{3}$). Hence to conclude the proof we need to prove the subgoal $a^{2}.\nil   \stackrel{\et}{\sim} P_{3}$. This can be derived 
as follows:
\begin{align*}
\tau(a^{2}.\nil)
&  = a.\nil & \text{definition of $\tau$}\\
&  =a.\nil \oplus a.\nil   & \text{$I$ axiom}\\
&  =a.\nil  \oplus ( a.\nil \ndplusccs \nil)   & a.\nil   \stackrel{\et}{\sim}a.\nil \ndplusccs \nil \\
&  =\tau(P_{3}) &\text{definition of $\tau$}.
\end{align*}
We can then conclude the proof by proving the final subgoal   $a.\nil   \stackrel{\et}{\sim}a.\nil \ndplusccs \nil $ as follows:
\begin{align*}
\tau(a.\nil)
&  =\nil &\text{definition of $\tau$}\\
&  =\nil \oplus \star & \text{$\bot$ axiom}\\
&  =\tau(a.\nil \ndplusccs \nil ) &\text{definition of $\tau$.}
\end{align*}
Note that in the above proof of $P_1 \stackrel{\et}{\sim} P_2$ the bottom axiom ($\bot$) is used but the black--hole axiom is not. Indeed, the two process terms $P_{1}$ and $P_{2}$ are also equated in the weaker theory $\etcsb$ of convex semilattices with bottom ($F=\cset^\downarrow$), but not in the  even weaker theory $\etpcs$ of pointed convex semilattices ($F=\cset(X+\oneset)$), which does not include the bottom axiom. Similarly, the two process terms $\nil$ and $(\nil \pplusccs  \onehalf P)$ are equated, for all $P$,  when $\etcsbb$ is used ($F=\cset(X)+\oneset$, i.e., Segala system) but not in $\etcsb$ or $\etpcs$, as proving them equivalent requires the black--hole axiom. 
Which choice of functor $F \in \{ \cset(+\oneset), \cset +\oneset, \cset^\downarrow \}$ is best suited in a specific modelling situation is of course beyond the scope of this paper. But once the choice is made, appropriate equational theories are automatically provided by our results.

Hence, our equational reasoning apparatus can be used in program 
equivalence proofs, for various types of transition systems ($F \in \{ \cset(+\oneset), \cset +\oneset, \cset^\downarrow \}$).  While the process algebra considered in this section is deliberately simple, more programming features (e.g., parallel composition, recursion, etc.) as well as more deductive principles (e.g., bisimulation up--to methods \cite{DBLP:journals/acta/BonchiPPR17}) can be considered and built on top of this core equational framework.

Our results also provide quantitative equational theories for metric reasoning when the chosen functor is  $F \in \{ \cset(+\oneset), \cset^\downarrow \}$. Following standard ideas \cite{DBLP:journals/lmcs/BreugelSW08,baldan14,DGJP03} it is possible to endow the set $\pccs$ of process terms with the \emph{bisimulation metric} using the Hausdorff--Kantorovich liftings.

\begin{definition}\label{def:bis:metric}
Let $F \in \{ \cset(+\oneset), \cset^\downarrow \}$ be a functor on $\Sets$ and let $\tau:X \to F(X)$ be a coalgebra for $F$. 
A $1$--bounded metric $d:X\times X\rightarrow[0,1]$ is a \emph{bisimulation metric} if for all $x,y\!\in\! X$ it holds that 
$
\Delta(d)(\tau_{F}(x), \tau_{F}(y)\big) \leq d(x,y)
$, for $\Delta(d)=\hk(d+\donemet)$.
The bisimilarity metric on $X$, denoted by $d_\simeq$, is the point--wise infimum of all bisimulation metrics.
\end{definition}
Once $Proc$ is endowed with $d_\simeq$, the transition function 
$\tau_{\hat F}: (\pccs,d_\simeq) \rightarrow \hat{F}(\pccs,d_\simeq)$, for  $\hat{F} \in \{ \lcset(+\onemet), \ldownarrow \}$
 is non--expansive and so we obtain a $\Met$ coalgebra on process terms. For a simple example, consider the quantitative theory $\qetpcs$ of pointed convex semilattices (i.e., $\hat{F}= \lcset(+\onemet)$) and the two process terms 
\[Q_1 = \nil \pplusccs{\frac{1}{2}} a.\nil\qquad Q_2 = \nil \pplusccs{\frac{1}{4}} a.\nil.\]
 These two processes are not equivalent (Definition \ref{def:bis}) but it is possible to prove that $d_\simeq(Q_1,Q_2)\leq \frac{1}{4}$, using the metric $d$ defined as $d(Q_{1},Q_{2})=d(Q_{2},Q_{1})=\frac 1  4$ and as the discrete metric on all other pairs.
We have to show that $\Delta(d)\big(\tau(Q_1), \tau(Q_2)\big) \leq \frac 1 4$.
The deductive apparatus of the quantitative equational theory $\qetpcs$ makes this easy:
\vvm

$ \Delta(d)\big(\tau(Q_1), \tau(Q_2)\big) $\\
{\centering
\begin{tabular}{l l r}
$=$ & $\Delta(d)\big( \star +_{\frac{1}{2}} a.\nil      ,    \star +_{\frac{1}{4}} a.\nil  \big)$ & definition of $\tau$\\

$=$ & $\Delta(d)\big( (\star \pplus \onehalf \star )+_{\frac{1}{2}} a.\nil      ,    \star +_{\frac{1}{4}} a.\nil  \big)$ & $I_{p}$ axiom\\

$=$ & $\Delta(d)\big( \star \pplus {\frac 1 4} (\star +_{\frac{1}{3}} a.\nil)      ,    \star +_{\frac{1}{4}} a.\nil  \big)$ & $A_{p}$ axiom\\

$\leq$ & ${\frac 1 4}\cdot 0 + \frac 3 4 \cdot  \frac 1 3 = \frac 1 4$ 
\end{tabular}
}

where the last inequality is deduced with the (K) inference rule (see Definition \ref{def:convexsemilattices:met}), with premises $\Delta(d)( \star,\star)\!=\! 0$ and
\vvm

$\Delta(d)( \star +_{\frac{1}{3}} a.\nil     ,  a.\nil  \big)$ \\
\begin{tabular}{l l r}
${=}$ &$\Delta(d)\big( \star +_{\frac{1}{3}} a.\nil     ,    a.\nil +_{\frac{1}{3}} a.\nil  \big)$ & \ \ \ \ \ \ \ \ \ $I_p$ { axiom}\\
 $\leq$ &$\frac{1}{3}$
\end{tabular}
\vvm

where, again, the inequality is derived by the (K) inference rule, using as premises $\Delta(d)( \star,a.\nil ) \leq 1$ and $\Delta(d)( a.\nil,a.\nil )= 0$.

%% file: App_background.tex

In this section, we recall and expand some definitions given in Section \ref{sec:back} in order to help the reader understand the more technical proofs in the remainder of the appendix. In particular, we draw the commutative diagrams corresponding to some definitions; we use single arrows to represent morphisms and double arrows to represent natural transformations.

\textbf{Monad}

Given a category $\Cat$, a monad on $\Cat$ is a triple $(\mon, \eta, \mu)$ composed of a functor $\mon\colon\Cat \rightarrow \Cat$ together with two
natural transformations: a unit $\eta\colon id
\Rightarrow \mon$, where $id$ is the identity functor on $\Cat$, and a multiplication $\mu \colon \mon^{2} \Rightarrow
\mon$, satisfying the two laws \eqref{diag-unitmonad} 
$\mu \circ \eta\mon = \mu \circ \mon\eta = id_{\Cat} $ and 
\eqref{diag-multmonad} 
$  \mu\circ \mon\mu = \mu \circ\mu\mon$.

    \begin{minipage}{0.48\textwidth}
        \begin{equation}\label{diag-unitmonad}
            \begin{tikzcd}
                \mon  \arrow[rd, Rightarrow, "\one_\mon "'] \arrow[r, Rightarrow, "\mon \eta"] & \mon^2 \arrow[d, Rightarrow, "\mu"] & \mon  \arrow[ld, Rightarrow, "\one_\mon"] \arrow[l, Rightarrow, "\eta \mon "'] \\ & \mon  &
            \end{tikzcd}
        \end{equation}
    \end{minipage}
    \begin{minipage}{0.48\textwidth}
        \begin{equation}\label{diag-multmonad}
            \begin{tikzcd}
                \mon^3 \arrow[d, Rightarrow, "\mu \mon "'] \arrow[r, Rightarrow, "\mon \mu"] & \mon^2 \arrow[d, Rightarrow, "\mu"] \\
		\mon^2 \arrow[r, Rightarrow, "\mu"'] & \mon 
            \end{tikzcd}
        \end{equation}
    \end{minipage}

\textbf{Monad Distributive Law}    
    
     Let $(M, \eta, \mu)$ and $(\widehat{M}, \widehat{\eta}, \widehat{\mu})$ be two monads on $\Cat$. A natural transformation $\lambda: M \widehat{M}\Rightarrow \widehat{M}M$ is called a \emph{monad distributive law of $M$ over  $\widehat{M}$} if it it satisfies the laws
    $\lambda \circ M\widehat\eta=\widehat\eta M$, $\lambda \circ \eta \widehat M=\widehat M \eta $, $\lambda \circ \mu \widehat M= \widehat M \mu \circ \lambda M \circ M \lambda$ and $\lambda \circ M \widehat \mu= \widehat \mu M \circ \widehat M \lambda \circ \lambda \widehat M$, i.e., it 
makes \eqref{diag-mondistlaw1} and \eqref{diag-mondistlaw2} commute.\\
        \begin{minipage}{0.43\textwidth}
    \begin{equation}\label{diag-mondistlaw1}
        \begin{tikzcd}
            M \arrow[rd, "\widehat{\eta}M"', Rightarrow] \arrow[r, "M\widehat{\eta}", Rightarrow] & M\widehat{M} \arrow[d, "\lambda", Rightarrow] & \widehat{M} \arrow[l, "\eta \widehat{M}"', Rightarrow] \arrow[ld, "\widehat{M}\eta", Rightarrow] \\ & \widehat{M}M &
        \end{tikzcd}
    \end{equation}
    \end{minipage}
        \begin{minipage}{0.55\textwidth}
    \begin{equation}\label{diag-mondistlaw2}
        \begin{tikzcd}
            MM\widehat{M} \arrow[d, "M\lambda"', Rightarrow] \arrow[r, "\mu \widehat{M}", Rightarrow] 
            & M\widehat{M} \arrow[dd, "\lambda", Rightarrow] 
            & M\widehat{M}\widehat{M} \arrow[d, "\lambda \widehat{M}", Rightarrow] \arrow[l, "M\widehat{\mu}"', Rightarrow] \\
            M\widehat{M}M \arrow[d, "\lambda M"', Rightarrow]   
                    &    & \widehat{M}M\widehat{M} \arrow[d, "\widehat{M}\lambda", Rightarrow]\\
            \widehat{M}MM \arrow[r, "\widehat{M}\mu"', Rightarrow] 
            & \widehat{M}M 
            & \widehat{M}\widehat{M}M \arrow[l, "\widehat{\mu}M", Rightarrow] 
        \end{tikzcd}
    \end{equation}
    \end{minipage}

\textbf{$\mon$--algebra}

Let $(\mon:\Cat \rightarrow \Cat,\eta,\mu)$ be a monad. An algebra for $\mon$ is a pair $(A,\alpha)$ where $A\in\Cat$ is an object and $\alpha:\mon (A)\rightarrow A$ is a morphism such that 
\eqref{diag-algunit} 
$ \alpha \circ  \eta_A = \id_A$ and 
\eqref{diag-algmult} 
$\alpha \circ \mon \alpha= \alpha \circ \mu_A $ hold.\\
\begin{minipage}{0.32\textwidth}
        \begin{equation}\label{diag-algunit}
            \begin{tikzcd}
                A \arrow[rd, "\id_A"'] \arrow[r, "\eta_A"] & \mon A \arrow[d, "\alpha"] \\ & A
            \end{tikzcd}
        \end{equation}
    \end{minipage}
    \begin{minipage}{0.36\textwidth}
        \begin{equation}\label{diag-algmult}
            \begin{tikzcd}
                \mon ^2A \arrow[d, "\mon (\alpha)"'] \arrow[r, "\mu_A"] & \mon A \arrow[d, "\alpha"] \\
                \mon A \arrow[r, "\alpha"']  & A  
                \end{tikzcd}
        \end{equation}
    \end{minipage}
    \begin{minipage}{0.32\textwidth}
        \begin{equation}\label{diag-algmor}
            \begin{tikzcd}
                \mon A \arrow[d, "\alpha"'] \arrow[r, "\mon(f)"] & \mon A' \arrow[d, "\alpha'"] \\
                \mon A \arrow[r, "\alpha"']  & A'  
                \end{tikzcd}
        \end{equation}
    \end{minipage}\\
Given two $\mon$--algebras $(A,\alpha)$ and $(A^\prime,\alpha^\prime)$, a \emph{$\mon$--algebra morphism} is an arrow $f:A\rightarrow A^\prime$ in $\Cat$ such that \eqref{diag-algmor}
$
 f\circ \alpha = \alpha^\prime \circ \mon(f)  
$.

\textbf{Monad Map}

    Let $(\mon , \eta, \mu)$ and $(\widehat{\mon }, \widehat{\eta}, \widehat{\mu})$ be two monads, a natural transformation $\sigma: \mon  \Rightarrow \widehat{\mon }$ is called a \emph{monad map} if it satisfies the equations \eqref{diag-monmap1} $\widehat \eta=\sigma \circ \eta$ and \eqref{diag-monmap2} $\sigma \circ \mu=\widehat \mu \circ (\sigma \diamond \sigma)$.\\
    \begin{minipage}{0.48\textwidth}
        \begin{equation}\label{diag-monmap1}
            \begin{tikzcd}
                \id_{\mathbf{C}} \arrow[rd, "\widehat{\eta}"', Rightarrow] \arrow[r, "\eta", Rightarrow] & \mon  \arrow[d, "\sigma", Rightarrow] \\ & \widehat {\mon }
            \end{tikzcd}
        \end{equation}
    \end{minipage}
    \begin{minipage}{0.48\textwidth}
        \begin{equation}\label{diag-monmap2}
            \begin{tikzcd}
                \mon^2 \arrow[d, "\mu"', Rightarrow] \arrow[r, "\sigma \diamond \sigma", Rightarrow] & \widehat{\mon }^2 \arrow[d, "\widehat{\mu}", Rightarrow] \\
                \mon  \arrow[r, "\sigma"', Rightarrow] & \widehat{\mon }
            \end{tikzcd}
        \end{equation}
    \end{minipage}

\textbf{Monad $\cdot+\oneset$}

The \emph{termination} monad on $\Sets$ is the triple $(\cdot + \oneset, \eta^{+ \oneset},\mu^{+ \oneset})$ defined as in Proposition $\ref{prop:terminationcat}$.
For objects $X$ in $\Sets$, the functor $\cdot+\oneset$ maps $X$ to the coproduct $X+\oneset$, i.e., the disjoint union of the sets $X$ and $\oneset=\{\star\}$.
For arrows $f\!:\!X\rightarrow\! Y$ in $\Sets$, the functor $\cdot+\oneset$ maps $f$ to $f+\oneset \colon X + \oneset \to Y+\oneset$, defined as $f+\oneset = [\inl\circ f, \inr]$.
The unit $\eta^{+\oneset}_X:X\rightarrow X+\oneset$ is $\eta^{+\oneset}_X(x) = \inl(x)$ and
the multiplication $\mu^{+\oneset}_X:  (X+\oneset)+\oneset \rightarrow X+\oneset$ is defined as 
$\mu_X^{+\mathbf{1}} = [[\inl,\inr], \inr]$.
 If clear from the context, we may omit explicit mentioning of the injections, and write for example $(f+\oneset)(x) = x$ for $x \in X$ and $(f+\oneset)(\star) = \star$. 
 For the unit $\eta^{+\oneset}_X:X\rightarrow X+\oneset$ we write $\eta^{+\oneset}_X(x) = x$.
 For the multiplication $\mu^{+\oneset}_X:  (X+\oneset)+\oneset \rightarrow X+\oneset$ we let $\bpoint$ denote the element of the outer $\oneset$ and $\point$ denote the element of the inner $\oneset$, and we write $\mu_X^{+\mathbf{1}}(x) = x$ for $x\in X$,  $\mu_X^{+\mathbf{1}}(\point) = \point$ and $\mu_X^{+\mathbf{1}}(\bpoint) = \point$.

\textbf{Monad $\cset(+\oneset)$}

The \emph{finitely generated non-empty convex powerset of subdistributions} monad $(\cset(\cdot+\oneset), \eta^{\cset(+\oneset)}, \mu^{\cset(+\oneset)})$ in $\Sets$ is defined as follows. Given an object $X$ in $\Sets$, $\cset(X)$ is the collection of non-empty finitely generated convex sets of probability subdistributions on $X$, i.e., $\cset(X+\oneset) = \{ \conv(S) \mid S\in \fpset\dset{(X+\oneset)} \}$.
Given an arrow $f:X\rightarrow Y$ in $\Sets$, the arrow $\cset{(f)}: \cset(X+\oneset)\rightarrow \cset(Y+\oneset)$ is defined as $\cset(+\oneset){(f)}(S)= \{ \dset{(f+\oneset)}(\distr) \mid \distr\in S \} $.
The unit $\eta_X^{\cset(\cdot+\oneset)}:X\rightarrow \cset(X+\oneset)$ is defined as $\eta^{\cset(+\oneset)}_X (x) = \{\dirac x\}$. 
The multiplication $\mu^{\cset(+\oneset)}_X:  \cset({\cset}(X+\oneset)+\oneset)) \rightarrow {\cset}(X+\oneset)$ is defined using Proposition \ref{prop-distlawcomposition} and the distributive law in Corollary \ref{eq:iota}, we have
that for any $S\in \cset({\cset}(X+\oneset)+\oneset))$
\begin{equation}\label{eqn:defofmuCp1}
    \mu^{\cset(+\oneset)}_X(S)= \mu^{\cset}_{X+\oneset}(\bigcup_{\ddistr \in S}  \{\ddistr^{\point}\})= \bigcup_{\ddistr \in S}  \wms(\ddistr^{\point})
\end{equation}
where, if we let $\bpoint$ denote the element of the outer $\oneset$ and $\point$ denote the element of the inner $\oneset$, for any $\ddistr \in \dset(\cset(X+\oneset)+\oneset)$
we define  $\ddistr^{\point}\in \dset(\cset(X+\oneset))$ as
\[\ddistr^{\point} =\dset(\iota_{X+\oneset})\circ \cset\cset(\mu_{X}^{+\oneset}) (\ddistr)= \Big(\sum_{\substack{U\in \cset(X+\oneset)\\U\neq \{\dirac \point\}}} \ddistr(U) U \Big)+\big(\ddistr(\{\dirac \point\})+\ddistr(\bpoint)\big)\{\dirac\point\}.\]

\textbf{The isomorphism 
$\ics : \cset(X+\mathbf{1}) \to {\terms  X {\sigpcs}}_{/\etcs} $}

As explained in
Proposition \ref{prop:freepcs},
the presentation of the monad $\cset(+\oneset)$ in terms of the theory $\etpcs$ of pointed convex semilattices (Proposition \ref{prop-knownpresentations}) implies that the free pointed convex semilattice generated by $X$ is isomorphic to the pointed convex semilattice $(\cset(X+\mathbf{1}) ,\oplus,+_p, \{\dirac\point\})$ where for all $S_1,S_2\in\cset (X+\mathbf{1})$, $S_1\oplus S_2 = \conv(S_1\cup S_2)$ (convex union), $S_1 +_p S_2 = \wms ( p S_1 + (1-p)S_2)$ (weighted Minkowski sum), and the distinguished element is $\{\dirac\point\}$ . In other words, the set ${\terms  X {\sigpcs}}_{/\etpcs}$ of pointed convex semilattice terms modulo the equational theory $\etpcs$ can be identified with the set $\cset (X+\mathbf{1})$ of non-empty, finitely generated convex sets of finitely supported probability subdistributions on $X$.
The isomorphism is a simple variant of the isomorphism described in \cite{BSV20ar} for the theory of convex semilattices (without a point). It is given by
$\ics : \cset(X+\mathbf{1}) \to {\terms  X {\sigpcs}}_{/\etpcs} $
defined as 
$\ics(S) = [\bigcplus_{\distr\in \ub(S) } (\bigpplus_{x \in \support(\distr) } \distr(x)\,x)]_{/\etpcs}$, where $\bigcplus_{i\in I} x_{i}$ and $\bigpplus_{i\in I} p_{i} \,x$ are respectively notations for the binary operations $\cplus$ and $\pplus p$ extended to operations of arity $I$, for $I$ finite (see, e.g., \cite{stone:1949,BSS17}), and where $\ub(S)$ is the unique base of $S$ defined as follows.
Given a finitely generated convex set $S\subseteq \dset (X)$, there exists one minimal (with respect to the inclusion order) finite set $\ub(S)\subseteq \dset {(X)}$  such that $S=\conv(\ub(S))$. The finite set $\ub(S)$ is referred to as the \emph{unique base} of $S$ (see, e.g., \cite{BSV20ar}). The distributions in $\ub(S)$ are convex--linear independent, i.e.,  if $\ub(S) = \{\distr_1,\dots, \distr_n\}$, then for all $i$, $\distr_i \notin \conv(\{ \distr_j \, |\, j\neq i \})$.
{We remark that the equation $x \oplus y = x \oplus y \oplus (x \pplus p y)$, which explicitly expresses closure under taking convex combinations, is derivable from the theory of convex semilattices (see, e.g., \cite[Lemma 14]{BSV20ar}), and that this derivation critically uses the distributivity axiom (D).}

%

%% file: App_CSBandBH.tex

We first prove two useful lemmas, showing that $\gamma$ commutes with the operations of (possibly infinite) union and of weighted Minkowski sum, respectively. Since we are dealing with (generally not convex) unions, we consider the function $\gamma$ as defined (see Definition \ref{def:gamma}) on arbitrary subsets of $S\subseteq\dset(X+\oneset)$, rather than convex subsets $S\in\cset(X+\oneset)$. So the generalised $\gamma$ maps an arbitrary set $S\subseteq\dset (X+\oneset)$ to its subset of full probability distributions (i.e., such that $\distr(\point)=0$) if such set is nonempty, and to $\point\in\oneset$ otherwise.

\begin{lemma}\label{lem:gammaunion}
For any family of non--empty sets $\{S_i \subseteq \dset(X+\oneset)\}_{i \in I}$, we have 
    \[\gamma_X\left( \bigcup_{i \in I} S_i \right) = \begin{cases}
        \bigcup_{\substack{i \in I\\ \gamma_X(S_i) \neq \point}} \gamma_X(S_i) &\exists i \in I, \gamma_X(S_i) \neq\point\\\point &\text{o/w}
    \end{cases},\]
\end{lemma}
\begin{proof}
    It is clear that if $\gamma_X(S_i) = \point$ for all $i \in I$, then all distributions $\distr$ in $\bigcup_{i \in I} S_i$ are not full (i.e., $\distr(\point)>0$), thus $\gamma_X(\cup_i S_i) = \point$. Now, suppose $\exists i \in I, \gamma_X(S_i) \neq \point$, or equivalently, there is at least one full distribution in $\cup_{i \in I}S_i$. Then, $\gamma_X(\cup_i S_i)$ is, by definition, the union of all full distributions in each $S_i$. Finally, since there are no full distributions in $S_i$ if and only if $\gamma_X(S_i) = \point$, we obtain
    \[\gamma_X\left( \bigcup_{i \in I} S_i \right) = \bigcup_{\substack{i \in I\\ \gamma_X(S_i) \neq \point}} \gamma_X(S_i)\]
\end{proof}
\begin{lemma}\label{lem:gammawms}
  For any $\ddistr \in \dset\cset(X+\mathbf{1})$, we have
    \[\gamma_X(\wms(\ddistr)) = \begin{cases}
        \wms(\dset(\gamma_X)(\ddistr)) & \forall U \in \support(\ddistr), \gamma_X(U) \neq \point\\
        \point & \text{o/w}
    \end{cases}.\]
\end{lemma}
\begin{proof}
Note that if there exists $V \in \support(\ddistr)$ with $\gamma_X (V) = \point$ (i.e., if $V$ does not contain full distributions) then all distributions $\sum_{U \in \support(\ddistr)} \ddistr(U)\cdot d_U \in \wms(\ddistr)$ are not full, i.e., $\gamma_X(\wms(\ddistr)) = \point$. This proves the second condition of the lemma. Now, for the first one, assume that $\forall U \in \support(\ddistr), \gamma_X(U) \neq \point$. We then have the following derivation.
    \begin{align*}
        \gamma_X(\wms(\ddistr)) &= \gamma_X\left\{ \sum_{U \in \support(\ddistr)} \ddistr(U)\cdot d_U: d_U \in U \right\}\\
        &= \left\{ \sum_{U \in \support(\ddistr)} \ddistr(U)\cdot d_U: d_U \in U \textnormal{ is full, i.e., } d_U(\point) =0 \right\}\\
        &= \left\{ \sum_{U \in \support(\ddistr)} \ddistr(U)\cdot d_U: d_U \in \gamma_X(U) \right\}\\
        &= \wms(\dset(\gamma_X)(\ddistr))
    \end{align*}
where the second equality holds because $\sum_{U \in \support(\ddistr)} \ddistr(U)\cdot d_U$ is full if and only if all $d_U\in U$ are full.
\end{proof}

We are now ready to prove the results of Section \ref{sec:cplusone}.


\textbf{Proof of Lemma \ref{lem:gammadistrlaw}.}

We show that the family 
\[\gamma_X (\left\{ \distr_i \mid i \in I\right\}) = \begin{cases}
    \left\{ \distr_i \mid i \in I, \distr_i(\point) = 0\right\} &\exists i, \distr_i(\point) = 0\\\point &\text{o/w}
\end{cases}\]
is a natural transformation. First, $\gamma_X$ is well-typed because when $S \in \cset(X+\mathbf{1})$, there is a finite set $\ub(S)$ of distributions satisfying $\conv(\ub(S)) = S$. Then, if $S$ contains at least one full distribution, one can verify that \[\gamma_X(S) = \cc{\distr \in \ub(S) \mid \distr(\point)=0},\]thus $\gamma_X(S)$ is a convex and finitely generated subset of $\dset(X)$.
Second, $\gamma$ is natural by the following derivation, for any $f: X \rightarrow Y$ and $S = \left\{ \distr_i \mid i \in I\right\} \in \cset(X+\mathbf{1})$:
\begin{align*}
    (\cset(f)+\mathbf{1})(\gamma_X(S)) &= \begin{cases}
        \cset(f)\left\{ \distr_i \mid \distr_i(\point) = 0\right\} &\exists i, \distr_i(\point) = 0\\\point &\text{o/w}
    \end{cases}\\
    &= \begin{cases}
        \cset(f+\mathbf{1})\left\{ \distr_i \mid \distr_i(\point) = 0\right\} &\exists i, \distr_i(\point) = 0\\\point &\text{o/w}
    \end{cases}\\
    &= \gamma_Y(\cset(f+\mathbf{1})(S))
\end{align*}
The last equality holds because $f(x) \neq \point$ for any $x \in X$, since $\point \in \oneset$ is assumed to not belong to $Y$.

We need to show that $\gamma: \cset(\cdot+\mathbf{1}) \Rightarrow \cset+\mathbf{1}$ is a monad distributive law, i.e., that it satisfies the two commuting diagrams of \eqref{diag-mondistlaw1} and \eqref{diag-mondistlaw2}.

First we show that \eqref{diag-botconvmondistlaw1} commutes.
\begin{equation}\label{diag-botconvmondistlaw1}
    \begin{tikzcd}
        \mathcal{C}X \arrow[rd, "\inl"'] \arrow[r, "\mathcal{C}(\inl)"] & \mathcal{C}(X + \mathbf{1}) \arrow[d, "\gamma_X"] & X + \mathbf{1} \arrow[ld, "\eta^{\mathcal{C}}_X +\mathbf{1}"] \arrow[l, "\eta^{\mathcal{C}}_{X+\mathbf{1}}"'] \\
        & \mathcal{C}X+\mathbf{1} &
    \end{tikzcd}
\end{equation}
For the L.H.S. (Left Hand Side), note that $\mathcal{C}(\inl)$ maps distributions $\distr\in\cset(X)$ to the corresponding ``full distribution'' $\distr\in\cset(X+\oneset)$, i.e., such that $\distr(\point)=0$. We have 
\[\gamma_X(\cset(\inl)(\left\{ \distr_i \mid i \in I\right\})) = \gamma_X(\left\{ \distr_i \mid i \in I\right\}) = \left\{ \distr_i \mid i \in I\right\}\]
where the last equality holds as $\left\{ \distr_i \mid i \in I\right\}$ only contains full distributions.
For the R.H.S., take an element $\omega \in X+\mathbf{1}$. If $\omega=x \in X$, it is first sent by $\eta^{\mathcal{C}}_{X+\mathbf{1}}$ to $\{\dirac{x}\}$ and then it is sent to $\{\dirac{x}\}$ by $\gamma_{X}$, and indeed we have  $(\eta^{\cset}_{X}+\mathbf{1})(x) =\{\dirac{x}\}$. 
If $\omega=\point \in \oneset$, then it is first sent to $\{\dirac\star\}$ and then to $\star$, and we have $(\eta^\cset_{X}+\mathbf{1})(\star) =\star$.


Finally, we show that \eqref{diag-botconvmondistlaw2} commutes:
\begin{equation}\label{diag-botconvmondistlaw2}
    \begin{tikzcd}
        \mathcal{C}\mathcal{C}(X+\mathbf{1}) \arrow[rr, "\mu^{\mathcal{C}}_{X+\mathbf{1}}"] \arrow[d, "\mathcal{C}(\gamma_X)"'] 
        &  & \mathcal{C}(X+ \mathbf{1}) \arrow[dd, "\gamma_X"] 
        &  & \mathcal{C}((X+\mathbf{1})+\mathbf{1}) \arrow[d, "\gamma_{X+\mathbf{1}}"] \arrow[ll, "{\mathcal{C}(\mu^{+\mathbf{1}}_X)}"'] \\
        \mathcal{C}(\mathcal{C}X+\mathbf{1}) \arrow[d, "\gamma_{\mathcal{C}X}"'] & & & & \mathcal{C}(X+\mathbf{1})+\mathbf{1} \arrow[d, "\gamma_X +\mathbf{1}"] \\
        \mathcal{C}\mathcal{C}X+\mathbf{1} \arrow[rr, "\mu^{\mathcal{C}}_X +\mathbf{1}"'] &  & \mathcal{C}X+\mathbf{1} &  & (\mathcal{C}X+\mathbf{1})+\mathbf{1} \arrow[ll, "{\mu^{+\mathbf{1}}_{\mathcal{C}X}}"]       
    \end{tikzcd}
\end{equation}
We first prove the L.H.S. of \eqref{diag-botconvmondistlaw2}. Starting with a convex set $S = \{\ddistr_i\}_{i \in I} \subseteq \cset\cset(X+\mathbf{1})$, the right-then-down path yields, by definition of $\mu^\cset$ (Definition \ref{def:set:cset}) and applying Lemma \ref{lem:gammaunion}:
\begin{align*}
    \gamma_X(\mu^{\cset}_{X+\mathbf{1}}(S)) &= \gamma_X\left(\bigcup_{i}\wms(\ddistr_i) \right)\\
    &= \begin{cases}
        \bigcup_{\substack{\{i\mid \gamma_X(\wms(\ddistr_i)) \neq \point\}}}\gamma_X(\wms(\ddistr_i))& \exists i, \gamma_X(\wms(\ddistr_i)) \neq \point\\
        \point &\text{o/w}
    \end{cases}\\
\end{align*}    
Then, by Lemma \ref{lem:gammawms}, we derive that $\gamma_X(\wms(\ddistr_i)) \neq \point$ if and only if $\forall U \in \support(\ddistr_i), \gamma_X(U) \neq \point$ and we can rewrite the function above as: 
\begin{align*}    
    &= \begin{cases}
        \bigcup_{\{i\mid \forall U \in \support(\ddistr_i), \gamma_X(U) \neq \point\}}\wms(\dset(\gamma_X)(\ddistr_i))& \exists i, \forall U \in \support(\ddistr_i), \gamma_X(U) \neq \point\\
        \point &\text{o/w}
    \end{cases}
\end{align*}
Taking the down-then-right path, we have the following derivation.
\begin{align*}
    S &\supp{\cset(\gamma_X)} \left\{ \dset(\gamma_X)(\ddistr_i)\right\}\\
    &\supp{\gamma_{\cset X}} \begin{cases}
        \left\{ \dset(\gamma_X)(\ddistr_i) \mid \star\not\in \support(\dset(\gamma_X)(\ddistr_i))\right\}
        & \exists i,  \star\not\in \support(\dset(\gamma_X)(\ddistr_i))\\
        \point &\text{o/w}
    \end{cases}\\
    &= \begin{cases}
        \left\{ \dset(\gamma_X)(\ddistr_i) \mid \forall U \in \support(\ddistr_i), \gamma_X(U) \neq \point\right\} & \exists i, \forall U \in \support(\ddistr_i), \gamma_X(U) \neq \point\\
        \point &\text{o/w}
    \end{cases}\\
    &\supp{\mu^{\cset}_X+\mathbf{1}} \begin{cases}
        \bigcup_{\substack{\{i\mid \forall U \in \support(\ddistr_i), \gamma_X(U) \neq \point\}}}\wms(\dset(\gamma_X)(\ddistr_i))& \exists i, \forall U \in \support(\ddistr_i), \gamma_X(U) \neq \point\\
        \point &\text{o/w}
    \end{cases}
\end{align*}
where the equality in the derivation holds as 
$$\point \notin\support(\dset(\gamma_{X})(\ddistr_i))\text{ if and only if } \forall U \in \support(\ddistr_{i}), \gamma_{X}(U)\neq\point.$$
Hence, the L.H.S commutes.

For the R.H.S. of \eqref{diag-botconvmondistlaw2}, let $S = \{\distr_i\}_{i \in I} \in \cset ((X+\mathbf{1})+\mathbf{1})$. In the sequel, $\point$ will denote the element of the innermost $\oneset$  and $\bpoint$ the element of the outermost $\oneset$. Taking the top arrow, the morphism $\mathcal{C}(\mu^{+\mathbf{1}}_X)$ identifies both stars together by sending $S$ to $\{\bar\distr_i\}_{i \in I}$ where, again omitting injections, we let \[\bar\distr_i= (\distr_i(\point)+\distr_i(\bpoint))\point + \sum_{x \in X}\distr_i(x)x.\] 
Applying $\gamma_X$ to $\{\bar\distr_i\}_{i \in I}$ then leads to 
\[
\begin{cases}
    \point &\forall i, \point \in \support(\bar\distr_i)\\
    \left\{ \bar\distr_i \mid \point \notin \support(\bar\distr_i)\right\}& \text{o/w}
\end{cases}.\]
Note that for every $i$ it holds
\[\point \in \support(\bar\distr_i) \Leftrightarrow (\bpoint \in \support(\distr_i)\text{ or }\point \in \support(\distr_i))\]
and that, if there exists a $\bar\distr_i$  such that $\point \not\in \support(\bar\distr_i)$, then 
$\left\{ \bar\distr_i \mid \point \notin \support(\bar\distr_i)\right\}= \left\{ \distr_i \mid \point, \bpoint\notin \support(\distr_i)\right\}$.
We then derive that the left-then-down path gives
\[
{ \gamma_X \circ \mathcal{C}(\mu^{+\mathbf{1}}_X)} (S)=
\begin{cases}
    \point &\forall i, (\bpoint \in \support(\distr_i)\text{ or }\point \in \support(\distr_i))\\
    \left\{ \distr_i \mid \point, \bpoint\notin \support(\distr_i)\right\}& \text{o/w}
\end{cases}.\]

Taking the down-then-left path, we have the following chain
\begin{align*}
    S &\supp{\gamma_{X +\mathbf{1}}}\begin{cases}
        \bpoint & \forall i, \bpoint \in \support(\distr_i)\\
        \left\{ \distr_i \mid \bpoint \notin \support(\distr_i)\right\}&\text{o/w}
    \end{cases}\\
    &\supp{\gamma_{X}+\mathbf{1}}\begin{cases}
        \bpoint & \forall i, \bpoint \in \support(\distr_i)\\
        \point & \exists i, \bpoint \not\in \support(\distr_i)\text{ and }\forall i (\bpoint\notin \support(\distr_i)\Rightarrow \point \in \support(\distr_i))\\
        \left\{ \distr_i \mid \point,\bpoint \notin \support(\distr_i)\right\}&\text{o/w}
    \end{cases}\\
    &\supp{\mu_{\cset X}^{+\mathbf{1}}}\begin{cases}
        \point &(A)\\
        \left\{ \distr_i \mid \point,\bpoint \notin \support(\distr_i)\right\}& \text{o/w}
    \end{cases}
\end{align*}
where (A) is the condition
$$(\forall i, \bpoint \in \support(\distr_i)) \text{ or }
(\exists i, \bpoint \not\in \support(\distr_i)\text{ and }\forall i (\bpoint\notin \support(\distr_i)\Rightarrow \point \in \support(\distr_i))).$$
Condition (A) is equivalent to
$$\forall i (\point \in \support(\distr_i)\text{ or }\point \in \support(\distr_i))$$
and we thereby conclude that the R.H.S. commutes.


We conclude that $\gamma: \cset(\cdot+\mathbf{1}) \Rightarrow \cset+\mathbf{1}$ is a distributive law. \qed


\textbf{Proof of Lemma \ref{lem:gammamonadmap}}

We need to show that $\gamma$ is a monad map from $\cset(\cdot+\mathbf{1})$ to $\cset+\mathbf{1}$. First, the unit diagram in \eqref{diag-botmonmapunitbeta} commutes because both units send $x\in X$ to $\{\dirac{x}\}$ and $\gamma_X\{\dirac{x}\} = \{\dirac{x}\}$ by definition.
\begin{equation}\label{diag-botmonmapunitbeta}
    \begin{tikzcd}
        X \arrow[r, "\eta^{\cset(+\mathbf{1})}_X"] \arrow[rd, "\eta^{\cset+\mathbf{1}}_X"'] & \cset(X+\mathbf{1}) \arrow[d, "\gamma_X"] \\
        & \cset X+\mathbf{1}
        \end{tikzcd}
\end{equation}
Then, it is left to show that the following diagram commutes.
\begin{equation}\label{diag-botmonmapmultbetasimple}
 \begin{tikzcd}
 \cset ( \cset (X+\mathbf{1})+\mathbf{1}) 
 \arrow[rr, "\gamma\diamond \gamma", Rightarrow, bend left] 
 \arrow[r, "\gamma_{ \cset (X+\mathbf{1})}"'] 
 \arrow[d, "\mu^{\cset(+\oneset)}_{X}"'] 
 & \cset ( \cset (X+\mathbf{1}))+\mathbf{1} \arrow[r, " \cset (\gamma_{X})+\mathbf{1}"'] 
 & \cset ( \cset X+\mathbf{1})+\mathbf{1} 
 \arrow[d, "\mu_{X}^{\moncb}"]
 \\
 \cset (X+\mathbf{1}) \arrow[rr, "\gamma_X"'] & & \cset X+\mathbf{1} 
 \end{tikzcd}
\end{equation}
Let $\point$ denotes the element of the innermost $\oneset$ and $\bpoint$ the element of the outermost $\oneset$ and consider an arbitrary $S = \{\ddistr_i\}_{i\in I} \in \cset(\cset(X+\mathbf{1})+\mathbf{1})$. We have the following derivation for the top path.

\begin{align*}
    S &\supp{\gamma_{\cset(X+\mathbf{1})}} \begin{cases}
        \left\{ \ddistr_i \mid\bpoint\not\in \support(\ddistr_{i})\right\} & \exists i, \bpoint\not\in \support(\ddistr_{i})\\
        \bpoint & \text{o/w}
    \end{cases}\\
    &\supp{\cset(\gamma_X)+\mathbf{1}} \begin{cases}
        \left\{ \dset(\gamma_{X})(\ddistr_i) \mid \bpoint\not\in \support(\ddistr_{i}) \right\} 	&         \exists i, \bpoint\not\in \support(\ddistr_{i})  \\
        \bpoint & \text{o/w}\  \textnormal{ (i.e.,   $\forall i, \bpoint\in \support(\ddistr_{i})$ )  } \end{cases}\\
\end{align*}
Then, by applying $\mu^{\moncb}_{X}$ as defined in Equation \ref{eq:mubb} (Section \ref{sec:cplusone}),  we obtain that the right-then-down path yields: 

\begin{align*}
 \begin{cases}
	\mu^{\cset}_{X}(\gamma_{\cset(X)}\big( \left\{ \dset(\gamma_{X})(\ddistr_i) \mid \bpoint\not\in \support(\ddistr_{i}) \right\} )\big)
	& \begin{minipage}{0.5\textwidth}
	$   \exists i, \bpoint\not\in \support(\ddistr_{i}) \text{ and } \\
	\gamma_{\cset(X)}\big( \left\{ \dset(\gamma_{X})(\ddistr_i) \mid \bpoint\not\in \support(\ddistr_{i}) \right\} \big) \!\neq\! \point$
	\end{minipage}\\
        \point &      		\text{o/w} \\
\end{cases}\\
\end{align*}

Note that the condition 
$$
\gamma_{\cset(X)}\big( \left\{ \dset(\gamma_{X})(\ddistr_i) \mid \bpoint\not\in \support(\ddistr_{i}) \right\} \big) \!\neq\! \point
$$
 holds if there exists some $i$ such that $\dset(\gamma_{X})(\ddistr_i)$ is a full distribution, which is the case if and only if $\forall U \in \support(\ddistr_{i}),\gamma_{X}(U)\neq\point$. Hence we can rewrite as follows:

\begin{align*}
\begin{cases}
	\mu^{\cset}_{X}(\gamma_{\cset(X)}\big( \left\{ \dset(\gamma_{X})(\ddistr_i) \mid \bpoint\not\in \support(\ddistr_{i}) \right\} )\big)
	&        \begin{minipage}{0.5\textwidth}
        $\exists i.\Big( \bpoint\not\in \support(\ddistr_{i}) $ \\
         \text{ and } $\forall U \in \support(\ddistr_{i}),\gamma_{X}(U)\neq\point\Big)$
        \end{minipage}\\
        \point &      		\text{o/w} \\
          \end{cases}\\
\end{align*}

Finally, by applying the definition of $\mu^{\cset}_{X}$ and of $\gamma_{\cset(X)}$ we have the equality:

\begin{center}
$\mu^{\cset}_{X}(\gamma_{\cset(X)}\big( \left\{ \dset(\gamma_{X})(\ddistr_i) \mid \bpoint\not\in \support(\ddistr_{i}) \right\} )\big) $
\\
$=$\\
$  \bigcup \left\{ \wms (\dset(\gamma_{X})(\ddistr_i)) \mid \bpoint\not\in \support(\ddistr_{i}) \text{ and }\forall U \in \support(\ddistr_{i}),\gamma_{X}(U)\neq\point \right\} 	
$
\end{center}
which allows us to rewrite as follows:
\begin{align*}
    \begin{cases}
        \bigcup \left\{ \wms (\dset(\gamma_{X})(\ddistr_i)) \mid \bpoint\not\in \support(\ddistr_{i}) \text{ and }\forall U \in \support(\ddistr_{i}),\gamma_{X}(U)\neq\point \right\} 	
        &  
       \begin{minipage}{0.5\textwidth}
        $\exists i.\Big( \bpoint\not\in \support(\ddistr_{i}) $ \\
         \text{ and } $\forall U \in \support(\ddistr_{i}),\gamma_{X}(U)\neq\point\Big)$
        \end{minipage}
	\\
        \point & \text{o/w}    \end{cases}
\end{align*}

Let us now consider the down-then-right path of diagram \eqref{diag-botmonmapmultbetasimple}. By applying $\mu^{\cset(+\oneset)}_{X}$ as given in \eqref{eqn:defofmuCp1}, we obtain:
$$\mu^{\cset(+\oneset)}_{X}(S)=
        \bigcup \{ \wms (\ddistr_i^{\point}) \mid i \in I\}. $$
Then $\gamma_{X}$ gives, by applying Lemma \ref{lem:gammaunion} and Lemma \ref{lem:gammawms}: 
\begin{align*}
    \gamma_{X}(\mu^{\cset+\oneset}_{X}(S))
    &=\gamma_{X}(\bigcup \{ \wms (\ddistr_i^{\point}) \mid i \in I\})\\
            &= \begin{cases}
        \bigcup_{\{i\mid \gamma_X(\wms(\ddistr_i^{\point})) \neq \point\}}\gamma_X(\wms(\ddistr_i^{\star}))
        & \exists i, \gamma_X(\wms(\ddistr_i^{\point}))\neq \point\\
        \point &\text{o/w}
    \end{cases}\\
        &= \begin{cases}
        \bigcup_{\{i\mid \forall U \in \support(\ddistr_i^{\star}), \gamma_X(U) \neq \point \}}\wms(\dset(\gamma_X)(\ddistr_i^{\star}))
        & \exists i, \forall U \in \support(\ddistr_i^{\star}), \gamma_X(U) \neq \point\\
        \point &\text{o/w}
    \end{cases}
\end{align*}
We then conclude that diagram \eqref{diag-botmonmapmultbetasimple} commutes, as the following equality holds:
\[\{\ddistr_{i}^{\point}\mid \forall U \in \support(\ddistr_i^{\star}), \gamma_X(U) \neq \point\} 
= \{\ddistr_{i}\mid \bpoint\not\in \support(\ddistr_{i}) \text{ and }\forall U \in \support(\ddistr_{i}),\gamma_{X}(U)\neq\point\} \]
For the left-to-right inclusion ($\subseteq$), suppose that $\ddistr_{i}^{\point}$ is such that $\forall U \in \support(\ddistr_i^{\star}), \gamma_X(U) \neq \point$.
Note that $\ddistr_i^{\star}= \ddistr_i$ because $\bpoint \notin \support(\ddistr_{i})$. Indeed, it is not possible that $\bpoint \in \support(\ddistr_{i})$ as this leads to a contradiction because it implies that $\{\dirac \point\}\in \support{(\ddistr_{i}^{\point})}$, with $\gamma_X(\{\dirac \point\}) = \point$. Hence, it follows from the hypothesis that $\forall U \in \support(\ddistr_i), \gamma_X(U) \neq \point$.

For the right-to-left inclusion $(\supseteq)$, if $\bpoint \not\in \support(\ddistr_{i})$ then $\ddistr_i^{\star}= \ddistr_i$, and thus $\forall U \in \support(\ddistr_i), \gamma_X(U) \neq \point$ implies $\forall U \in \support(\ddistr_i^{\point}), \gamma_X(U) \neq \point$. \qed

\textbf{Proof of Lemma \ref{embedding_lemma_cbotbh}}.

    The fact that $U^{\gamma}$ is a functor is the content of Proposition \ref{prop-MonmapFunctor}. The fact that it is fully faithfull is obvious since $U^{\gamma}$ acts like the identity on morphisms. The fact that it is injective on objects follows from surjectivity of $\gamma_A$ for any set $A$. Indeed, if $(A, \alpha)$ and $(A, \alpha')$ are such that $\alpha \circ \gamma_A = \alpha' \circ \gamma_A$, then $\alpha = \alpha'$. \qed

\textbf{Proof of Lemma \ref{moncb-final-lemma}.}

For the first point (1), we need to show that the pointed convex semilattice $P\circ U^{\gamma}((A,\alpha))=P((A,\alpha \circ \gamma_A))=(A, \oplus^{\alpha \circ \gamma_A}, \{\pplus p^{\alpha \circ \gamma_A}\}_{p\in (0,1)}, \star^{\alpha \circ \gamma_A} )$ satisfies the bottom and black-hole axioms. This is proved by the following equational reasoning steps.
\begin{align*}
a \oplus^{\alpha \circ \gamma_A} \point^{\alpha \circ \gamma_A} &=(\alpha \circ \gamma_A)(\cc{\dirac{a}, \dirac{\alpha \circ \gamma_A(\{\dirac{\point}\})}}) & \text{definition of $P$}\\
&= \alpha \circ \gamma_A (\cc{\dirac{\alpha \circ \gamma_A(\{\dirac{a}\})}, \dirac{\alpha \circ \gamma_A(\{\dirac{\point}\})}}) &\text{definition of $\EM(\cset(+\oneset))$}\\
&= \alpha \circ \gamma_A \circ \mu^{\cset(+\oneset)}_{A}(\cc{\dirac{\{\dirac{a}\}}, \dirac{\{\dirac{\point}\}}}) &\text{definition of $\EM(\cset(+\oneset))$}\\
&= \alpha \circ \gamma_A (\cc{\dirac{a}, \dirac{\point}}) &\text{definition of $\mu^{\cset(+\oneset)}$}\\
&=\alpha(\{\dirac{a}\}) &\text{definition of $\gamma$}\\
&=a &\text{definition of $\EM(\cset+\oneset)$}
\end{align*}
\begin{align*}
a \pplus p^{\alpha \circ \gamma_A} \point^{\alpha \circ \gamma_A} &=(\alpha \circ \gamma_A)(\{p\, a + (1-p) {\alpha \circ \gamma_A(\{\dirac{\point}\})}\}) & \text{definition of $P$}\\
&= \alpha \circ \gamma_A (\{ p\, {\alpha \circ \gamma_A(\{\dirac{a}\})} + (1-p) {\alpha \circ \gamma_A(\{\dirac{\point}\})}\}) &\text{definition of $\EM(\cset(+\oneset))$}\\
&= \alpha \circ \gamma_A \circ \mu^{\cset(+\oneset)}_{A}(\{p\,\{\dirac{a}\} + (1-p) \{\dirac{\point}\}\}) &\text{definition of $\EM(\cset(+\oneset))$}\\
&= \alpha \circ \gamma_A (\{p\,{a} + (1-p) {\point}\}) &\text{definition of $\mu^{\cset(+\oneset)}$}\\
&=\alpha \circ \gamma_A (\{\dirac{\point}\}) &\text{definition of $\gamma$}\\
&=\point^{\alpha \circ \gamma_A} &\text{definition of $P$}
\end{align*}
For the second point (2), let $(A, \alpha)$ be the $(\cset(\cdot+\mathbf{1}))$--algebra corresponding via $P^{-1}$ to a pointed convex semilattice $\mathbb{A} = (A, \oplus^{\alga}, +_p^{\alga}, \point^{\alga})$ satisfying bottom and black-hole, i.e.,  $(A, \alpha) = P^{-1}(\mathbb{A})$. This $(\cset(\cdot+\mathbf{1}))$--algebra has two key properties. First, $\alpha(S) = \alpha(\{\dirac\point\})$ whenever $\forall \distr \in S, \point \in \support(\distr)$. Indeed, we have $\alpha(S) = \oplus_{\distr \in \ub(S)} \alpha(\{\distr\})$ and since each term of this sum contains $\point$, the whole sum is equal to the interpretation of $\point$, due to the black--hole axiom. Second, when $S$ contains at least one full distribution, we have $\alpha(S) = \alpha(\cset(\inl^{A+\mathbf{1}})(\gamma_A(S)))$. This follows from the following derivation.
\begin{align*}
    \alpha(S) &= \bigoplus_{\distr \in \ub(S)}^{\alga} \alpha(\{\distr\})\\
    &= \bigoplus_{\distr \in \ub(S) : \distr(\point) = 0}^{\alga} \alpha(\{\distr\}) &&\mbox{terms with $\point$ are ignored (by $\bot$ and $BH$)}\\
    &= \alpha(\cc{\distr \in \ub(S) : \distr(\point) = 0})\\
    &= \alpha(\cset(\inl^{A+\mathbf{1}})(\gamma_A(S))) &&\text{by definition of $\gamma_A$}
\end{align*}
From these properties, we conclude that $\alpha = \alpha' \circ \gamma_A$, where \[\alpha': \cset A+\mathbf{1} \rightarrow A = [\alpha \circ \cset(\inl^{A+\mathbf{1}}), \alpha(\{\dirac\point\})].\]
Therefore, it is left to show that $(A, \alpha')$ is a $(\cset+\mathbf{1})$-algebra to conclude that $(A, \alpha)$ is in the image of $U^{\gamma}$.

We have to show two diagrams commute.
\begin{equation}\label{diag-botCp1algebraunit}
    \begin{tikzcd}
    A \arrow[rd, "\id_{A}"'] \arrow[r, "a \mapsto \{\dirac{a}\}"] & \cset A+\mathbf{1} \arrow[d, "\alpha'"] \\ & A
        \end{tikzcd}
\end{equation}
Diagram \eqref{diag-botCp1algebraunit} commutes since $x\in A$ is sent to $\{\dirac{x}\}$ and $\alpha'\{\dirac{x}\} = \alpha\{\dirac{x}\} = x$.
\begin{equation}\label{diag-botCp1algebraassoc}
\begin{tikzcd}
\cset(\cset A +\mathbf{1})+\mathbf{1} \arrow[rrr, "{\cset(\alpha')+\mathbf{1}}"] \arrow[d, "\gamma_{\cset A}+\mathbf{1}"'] & & & \cset A+\mathbf{1} \arrow[dd, "\alpha'"] \\
\cset\cset A+\mathbf{1}+\mathbf{1} \arrow[d, "{[\mu^{\cset}_A, \inr^{\cset A+\mathbf{1}}, \inr^{\cset A+\mathbf{1}}]}"'] & & & \\
\cset A+\mathbf{1} \arrow[rrr, "\alpha'"'] & & & A 
\end{tikzcd}
\end{equation}

For Diagram \eqref{diag-botCp1algebraassoc}, let $S \in \cset(\cset A+\mathbf{1})+\mathbf{1}$. We need to distinguish three cases.

First,  if $S = \bpoint$ then both path send $S$ to $\alpha\{\dirac\point\}$.

Second, if $\forall \ddistr\in S,\point \in \support(\ddistr)$, then the down-then-right path sends $S$ to $\alpha\{\dirac\point\}$ again and taking the right-then-down path and denoting $S = \{\ddistr_i + p_i\point\}_i$, we have (we omit all inclusions so $\alpha'(U) = \alpha(U)$ when $U \neq \bpoint \in \cset A +\mathbf{1}$):
\begin{align*}
    S &\supp{\cset(\alpha')+\mathbf{1}} \left\{ \mathcal{D}(\alpha)(\ddistr_i + p_i\{\dirac\point
    \}) \right\}\\
    &\supp{\alpha'} \alpha\left\{\mathcal{D}(\alpha)(\ddistr_i + p_i\{\dirac\point\}) \right\}\\
    &= \alpha(\cset(\alpha+\mathbf{1})\left( \left\{\ddistr_i+ p_i\{\dirac\point\} \right\} \right))\\
    &= \alpha\left( \mu^{\cset(+\mathbf{1})}_A\left( \left\{ \ddistr_i + p_i\{ \dirac\point\}\right\} \right) \right)&&\text{$\alpha$ is a $\cset(\cdot+\mathbf{1})$--algebra}\\
    &= \alpha\{\dirac\point\}
\end{align*}
Where the last equality follows because every distribution in the multiplication will contain $\point$ in its support and $\alpha$ satisfies the bottom and black-hole axioms.

Third, if $S$ contains at least one full distribution, the down-then-right path first sends $S = \left\{ \ddistr_i \mid i \in I\right\}$ to $\left\{ \ddistr_i \mid \ddistr_i(\point) = 0\right\}$. Then, this is sent to $$\alpha'(\mu^{\cset(+\mathbf{1})}_A\left\{ \ddistr_i \mid \ddistr_i(\point) = 0\right\}) = \alpha(\mu^{\cset(+\mathbf{1})}_A\left\{ \ddistr_i \mid \ddistr_i(\point) = 0\right\}).$$ On the right-then-down path, we have the following derivation.
\begin{align*}
    \left\{ \ddistr_i + p_i\point\right\} &\supp{\cset(\alpha')+\mathbf{1}} \left\{ \sum_{U \in \support(\ddistr_i)} \ddistr_i(U)\alpha(U) +p_i\alpha\{\dirac\point\} \mid i \in I\right\}\\
    &\supp{\alpha'} \alpha\left\{ \sum_{U \in \support(\ddistr_i)} \ddistr_i(U)\alpha(U) +p_i\alpha\{\dirac\point\} \mid i \in I\right\}\\
    &=\alpha\left( \cset(\alpha+\mathbf{1})\left\{ \sum_{U \in \support(\ddistr_i)} \ddistr_i(U)U+p_i\{\dirac\point\} \mid i \in I\right\}\right)\\
    &= \alpha\left( \mu^{\cset(+\mathbf{1})}_A \left\{ \sum_{U \in \support(\ddistr_i)} \ddistr_i(U)U +p_i\{\dirac\point\} \mid i \in I\right\}\right)\\
    &= \alpha\left( \mu^{\cset(+\mathbf{1})}_A\left\{ \ddistr_i \mid p_i = 0\right\} \right)
\end{align*}
The last equality holds because $\alpha$ satisfies the bottom and black-hole axioms.
\qed

%% file: App_Cdownarrow.tex

We first prove the following useful lemma.

\begin{lemma}\label{lem:kclosure1}
Let $\distr \in \dset(X+\mathbf{1})$. Then $\bcl_X(\{\distr\})$ is the smallest $\bot$--closed set containing $\distr$ and
        $\bcl_X(\{\distr\}) 
        = \left\{ \distrb \in \dset(X+\mathbf{1}) \mid \forall x \in X, \distrb(x) \leq \distr(x)\right\}$.
        \end{lemma}
\begin{proof}
        We first prove that for any $\distr \in \dset(X+\mathbf{1})$,
        $$\bcl_X\{\distr\} 
        = \left\{ \distrb \in \dset(X+\mathbf{1}) \mid \forall x \in X, \distrb(x) \leq \distr(x)\right\}.$$
        By $\bcl_X$ being a homomorphism, we obtain:
        \begin{align*}
            \bcl_X\{\distr\} &= \bcl_X\left( \wms \Big(\sum_{x \in \support(\distr)} \distr(x) \,\{\dirac{x}\} \Big)\right)\\
            &= \wms\Big(\sum_{x \in \support(\distr)} \distr(x) \,\bcl_X(\{\dirac{x}\})\Big)\\
            &= \wms\Big(\sum_{x \in \support(\distr)} \distr(x) \,\{p \,x+(1-p)\,\point\mid p \in [0,1]\}\Big)\\
            &= \left\{ \sum_{x \in \support(\distr)}\distr(x)(p_x \,x+(1-p_x)\,\point) \mid \forall x, p_{x} \in [0,1]\right\}\\
            &= \left\{ \distrb \in \dset(X+\mathbf{1}) \mid \forall x \in X, \distrb(x) \leq \distr(x)\right\}
        \end{align*}
        The last equality follows because the $p_x$s are chosen independently in $[0,1]$, thus any distribution with weight lower than $\distr$ at all $x \in X$ can be obtained by choosing the right $p_x$s.

This set is clearly $\bot$--closed, and any $\bot$--closed containing $\distr$ must contain $\bcl_X\{\distr\}$ by definition. Thus, $\bcl_X\{\distr\}$ is the smallest $\bot$--closed set containing $\distr$.

\end{proof}
%

As a corollary and using the definition of $\bot$--closure we derive the following:
\begin{lemma}\label{lem:kclosure2}
A set $S \in \cset(X+\mathbf{1})$ is $\bot$--closed if and only if for any $\distr \in S$, $\bcl_X(\{\distr\}) \subseteq S$.
\end{lemma}

\textbf{Proof of Theorem \ref{thm:Kclosure}}

    Let $S \in \cset(X+\mathbf{1})$. For each $\distr \in S$, $\bcl_X(\{\distr\}) \subseteq \conv(\bcl_X(S) \cup \bcl_X(\{\distr\})) = \bcl_X(\conv(S \cup \{\distr\})) = \bcl_X(S)$. Hence, Lemma \ref{lem:kclosure2} implies that $\bcl_X(S)$ is $\bot$--closed, i.e., $\bcl_X(S)\in  \cset^{\downarrow}(X)$.
    Moreover, as for all $\distr \in S$ it holds that $\distr \in \bcl_X(\{\distr\})$, by the same reasoning, we derive that $S\subseteq \bcl_X(S)$.
    To prove minimality, let $S'$ be a $\bot$--closed such that $S\subseteq S'$. We need to show that $\bcl_X(S)\subseteq S'$. 
    Let $\{\distr_1, \dots,\distr_n\}$ be a basis for $S$ (i.e., $S = \conv(\bigcup_{i =1}^n \{\distr_i\})$). By Lemma \ref{lem:kclosure2}, as $S'$ is $\bot$--closed and $\{\distr_1, \dots,\distr_n\}\subseteq S\subseteq S'$, we derive that $\bcl_X(\{\distr_{i}\})\subseteq S'$ for any $\distr_{i}$ in the basis.
    Thus, by convexity of $S'$ and by $\bcl_X$ being a homomorphism, we conclude that 
    \begin{equation}\label{eq:kbot}
    \bcl_X(S) =\bcl_X\left( \conv(\bigcup_{i =1}^n \{\distr_i\}) \right) =\conv(\bigcup_{i =1}^n \bcl_X(\{\distr_i\}))\subseteq S'.
    \end{equation}
    Hence, $\bcl_X(S)$ is the smallest $\bot$--closed set containing $S$.
    
We can now prove that $S$ is $\bot$--closed if and only if $S=\bcl_X(S)$. Suppose that $S$ is $\bot$--closed. Then, by \eqref{eq:kbot}, we have $\bcl_X(S)\subseteq S$. Since we already knew that $S\subseteq \bcl_X(S)$, we conclude that $S=\bcl_X(S)$.
For the converse implication, if $S=\bcl_X(S)$ then, as $\bcl_X(S)$ is $\bot$--closed, we have that $S$ is $\bot$--closed.\qed

\textbf{Proof of Theorem \ref{thm:downarrowmonad}}

To show that the triple $(\cset^{\downarrow}, \eta^{\moncd}, \mu^{\moncd})$ is a monad, we first prove some useful commuting property of $\bcl$ with respect to (arbitrary) unions and with respect to the multiplication of the monad $\cset(\cdot+1)$.

\begin{lemma}\label{lem:kunion}
Let $S\in \cset(X+\oneset)$. Then $\bcl_X(S)= \cup_{\distr\in S}\bcl_X(\{\distr\})$.
\end{lemma}
\begin{proof}
Let $S=\conv(\bigcup_{i =1}^n \{\distr_{i}\})\in \cset(X+\oneset)$.
By Theorem \ref{thm:Kclosure}, $\bcl_X(S)$ is $\bot$--closed, and thus by Lemma \ref{lem:kclosure2} we derive
that for all  $\distr\in \bcl_X(S)$, $\bcl_X(\{\distr\})\subseteq \bcl_X(S)$. As $S\subseteq \bcl_X(S)$, we have that $\bigcup_{\distr\in S}\bcl_X(\{\distr\})\subseteq \bcl_X(S)$. For the converse inclusion, using Theorem \ref{thm:Kclosure}, we find that it is enough to show $\bigcup_{\distr\in S}\bcl_X(\{\distr\})$ is $\bot$--closed as it clearly contains $S$.
First, we show it is convex. Let $\distrb_{1}, \distrb_{2}\in \bigcup_{\distr \in S} \bcl_X(\{\distr\})$.  
By Lemma \ref{lem:kclosure1}, there are $\distrc_{1},\distrc_{2} \in S$ such that $\forall x, \distrb_{1}(x)\leq \distrc_{1}(x)$ and $\forall x, \distrb_{2}(x)\leq \distrc_{2}(x)$.
We want to prove that 
$$p\cdot \distrb_{1} + (1-p)\cdot \distrb_{2} \in \bigcup_{\distr \in S} \bcl_X(\{\distr\})$$
which by Lemma \ref{lem:kclosure1} is equivalent to showing that there is a $\distr\in S$ such that $$\forall x, p\cdot \distrb_{1} + (1-p)\cdot \distrb_{2}(x)\leq \distr(x).$$
This holds by taking $\distr=p\cdot \distrc_{1} + (1-p)\cdot \distrc_{2}$, as 
$$\forall x, p\cdot \distrb_{1} + (1-p)\cdot \distrb_{2}(x)\leq p\cdot \distrc_{1} + (1-p)\cdot \distrc_{2}(x)$$
and, since $S$ is convex, we indeed have $p\cdot \distrc_{1} + (1-p)\cdot \distrc_{2}\in S$. Next, for any $\distrb \in \bigcup_{\distr \in S} \bcl_X(\{\distr\})$, there exists $\distr \in S$ such that $\distrb \in \bcl_X\{\distr\}$. We infer (using Lemma \ref{lem:kclosure1}) that $\bcl_X\{\distrb\} \subseteq \bcl_X\{\distr\} \subseteq \bigcup_{\distr \in S} \bcl_X(\{\distr\})$. Therefore, $\bigcup_{\distr \in S} \bcl_X(\{\distr\})$ is $\bot$--closed by Lemma \ref{lem:kclosure2}.

\end{proof}

\begin{lemma}\label{lem:kmu}
\begin{enumerate}
\item Let $S = \{\ddistr_i\}_{i \in I} \in \cset(\cset(X+\oneset)+\oneset)$. Then 
$$\bcl_X(\mu^{\cset(+\oneset)}_{X}(S))
=\bigcup_{i} \wms\Big(\sum_{\{\bcl_X(U) \mid U\in \support(\ddistr_i^{\point})\}}\big(\sum_{U\in \bcl_X^{-1}(\bcl_X(U))}\ddistr_i^{\point}(U)\big) \, \bcl_X(U)\Big)
$$
\item Let $S = \{\ddistr_i\}_{i \in I} \in \cset(\cset^{\downarrow}(X)+\oneset)$. Then
$$\bcl_X(\mu^{\cset(+\oneset)}_{X}(S))=\mu^{\cset(+\oneset)}_{X}(S)$$.
\end{enumerate}

\end{lemma}
\begin{proof}
\begin{enumerate}
\item It follows from Lemma \ref{lem:kunion} that
$$\bcl_X(\mu^{\cset(+\oneset)}_{X}(S))
=\bigcup_{i} \bcl_X(\wms(\ddistr_i^{\point}))$$
and, as $\bcl_X$ is a homomorphism with respect to weighted Minkowski sums (i.e., the interpretation of $+_{p}$ on $C(X+\oneset)$),
$$\bcl_X(\wms(\ddistr_i^{\point})) = \wms(\sum_{\{\bcl_X(U) \mid U\in \support(\ddistr_i^{\point})\}}(\sum_{U\in \bcl_X^{-1}(\bcl_X(U))}\ddistr_i^{\point}(U))\, \bcl_X(U)).$$
    
\item Let $S = \{\ddistr_i\}_{i \in I} \in \cset(\cset^{\downarrow}(X)+\oneset)$. By applying property 1 we have
$$\bcl_X(\mu^{\cset(+\oneset)}_{X}(S))
= \bigcup_{i} \wms(\sum_{\{\bcl_X(U) \mid U\in \support(\ddistr_i^{\point})\}}(\sum_{U\in \bcl_X^{-1}(\bcl_X(U))}\ddistr_i^{\point}(U))\, \bcl_X(U)).$$
For each $\ddistr_{i}\in S$ and for each $U \in \support(\ddistr_i^{\point})$, either $U\in \cset^{\downarrow}(X)$, and is thereby $\bot$--closed by definition, or $U=\{\dirac \star\}$, which is $\bot$--closed. Hence, each $U \in \support(\ddistr_i^{\point})$ is $\bot$--closed, and by Theorem \ref{thm:Kclosure} we derive that $\bcl_X(U) = U$ for any $U \in \support(\ddistr_i^{\point})$.
From this we derive
$$\bigcup_{i} \wms(\sum_{\{\bcl_X(U) \mid U\in \support(\ddistr_i^{\point})\}}(\sum_{U\in \bcl_X^{-1}(\bcl_X(U))}\ddistr_i^{\point}(U))\, \bcl_X(U))
=\bigcup_{i} \wms(\sum_{U \in \support(\ddistr_i^{\point})}\ddistr_i^{\point}(U) \, U)
$$
which is in turn equal to $\bigcup_{i} \wms(\ddistr_i^{\point})$. By definition, this is $\mu^{\cset(+\oneset)}_{X}(S)$.
    \end{enumerate}
\end{proof}

We can now show that the triple $(\cset^{\downarrow}, \eta^{\moncd}, \mu^{\moncd})$ is a monad.

It is immediate to verify that given any $f:X\rightarrow Y$ and $S\in \cset^{\downarrow}(X)$, it holds that $\cset^{\downarrow}(f)(S) \in \cset^{\downarrow}(Y)$, i.e., that $\cset^{\downarrow}(f)$ has indeed the correct type, and that $\cset^{\downarrow}$ is a functor.
    Then, we check the typing of the unit and multiplication. 
    For any set $X$ and $x \in X$, the set $\eta^{\moncd}_X(x) = \bcl_X\{\dirac{x}\}$ is $\bot$--closed by Lemma \ref{lem:kclosure1}. 
    For any $S\in \moncdp{\moncdp{X}}$, by Lemma \ref{lem:kmu}.2 it holds $\bcl_X(\mu^{\cset(+\oneset)}_{X}(S))=\mu^{\cset(+\oneset)}_{X}(S)$, which in turn is equal to $\mu^{\moncd}_{X}(S)$ by definition of  $\mu^{\moncd}$ and as $S$ is $\bot$--closed. Hence, the set $\mu^{\moncd}_{X}(S)$ is $\bot$--closed by Theorem \ref{thm:Kclosure}.
    
    The unit $\eta^{\moncd}$ is natural, as for any function $f: X \rightarrow Y$ we have
    \begin{align*}
        \cset(f+\mathbf{1})(\bcl_X\{\dirac{x}\}) &= \cset(f+\mathbf{1})(\{p\,x+ (1-p)\,\point\mid p \in [0,1]\})\\
        &= \{p\,f(x)+ (1-p)\, \point\mid p \in [0,1]\}\\
        &= \bcl_Y(\{\dirac{f(x)}\}).
    \end{align*}
    The naturality of $\mu^{\moncd}$ follows from the naturality of $\mu^{\cset(+\mathbf{1})}$ as the multiplication maps and the functors are defined similarly.
    
    Second, we show that the unit diagram \eqref{diag-unitmonad} commutes, namely, for any set $X$, $\mu^{\moncd}_X \circ  \cset^{\downarrow}(\eta^{\moncd}_X) = \id_{\cset^{\downarrow}(X)} = \mu^{\moncd}_X \circ \eta^{\moncd}_{\cset^{\downarrow}(X)}$. For the L.H.S., let $S = \bigcup_{i}\{\distr_i\}$, we have
    \begin{align*}
        &\mu^{\moncd}_X(\cset^{\downarrow}(\eta^{\moncd}_X)(S)) \\
        &= \mu^{\moncd}_X(\cset(\eta^{\moncd}_X+\oneset)(S))\\
        &= \mu^{\moncd}\left( \cup_{i}\{\dset(\eta^{\moncd}_X+\oneset)(\distr_i)\} \right)\\
        &= \bigcup_{i} \wms\left((\dset(\eta^{\moncd}_X+\oneset)(\distr_i))^{\point}\right)\\
        &= \bigcup_{i} \wms\left(\Big(\big(\sum_{x \in \support(\distr_i)}\distr_i(x) \{p_{x} \,x+ (1-p_{x})\,\point \mid \forall x \in \support(\distr_i),  p_{x} \in [0,1] \} \big)+ \big(1-\sum_{x \in \support(\distr_i)}\distr_i(x)\big) \bpoint \Big)^{\point}\right)\\
                &= \bigcup_{i} \wms\left(\Big(\sum_{x \in \support(\distr_i)}\distr_i(x) \{p_{x} \,x+ (1-p_{x})\,\point \mid \forall x \in \support(\distr_i),  p_{x} \in [0,1] \} \Big)+ \Big(1-\sum_{x \in \support(\distr_i)}\distr_i(x)\Big) \{\dirac\point\} \right)\\
        &= \bigcup_{i} \left\{ \Big(\sum_{x \in \support(\distr_i)}\distr_i(x) \cdot (p_{x} \,x+ (1-p_{x})\,\point)\Big) + \Big(1-\sum_{x \in \support(\distr_i)}\distr_i(x)\Big) \point \mid \forall x \in \support(\distr_i),p_{x} \in [0,1] \right\}\\
        &= \bigcup_{i} \left\{ \Big(\sum_{x \in \support(\distr_i)}(\distr_i(x) \cdot p_{x}) \,x\Big)+  \Big(1-\sum_{x \in \support(\distr_i)} (\distr_i(x) \cdot p_{x})\Big)\,\point \mid \forall x \in \support(\distr_i), p_{x} \in [0,1] \right\}\\
        &= \bigcup_{i} \left\{ \Big(\sum_{x \in \support(\distr_i)}q_{x} \,x\Big)+  \Big(1-\sum_{x \in \support(\distr_i)} q_{x}\Big)\,\point \mid \forall x \in \support(\distr_i), 0\leq q_{x} \leq \distr_{i}(x) \right\}\\
    \end{align*}
    Since it is clear that each $\distr_i$ is in this set, we infer that $S \subseteq \mu^{\moncd}_X(\cset(\eta^{\moncd}_X)(S))$. For the other inclusion, we have that $S$ is $\bot$-closed, and thus if $\distr_{i} \in S$ then whenever $0\leq q_{x} \leq \distr_{i}(x)$ it holds that $(\sum_{x \in \support(\distr_i)}q_{x} \,x)+  (1-\sum_{x \in \support(\distr_i)} q_{x}) \point \in S$.
%
%
%
For the R.H.S., we have 
    \begin{align*}
    \mu^{\moncd}_X(\eta^{\moncd}_{\cset^{\downarrow}(X)}(S)) 
    &= \mu^{\moncd}_X\{p\,S+ (1-p) \,\bpoint \mid p \in [0,1]\} \\
    &= \bigcup_{p \in [0,1]} \wms\big((p\,S+ (1-p) \,\bpoint)^{\point}\big) \\
    &= \bigcup_{p \in [0,1]} \wms\big(p\,S+ (1-p) \,\{\dirac \point\}\big) \\
    &= \bigcup_{p \in [0,1]} \left\{p\cdot \distr + (1-p)\cdot \dirac{\point}\mid \distr \in S \right\} \\
    &= S
    \end{align*}
    The last equality holds because is $\bot$--closed.

    Finally, we need to show that the associativity diagram \eqref{diag-multmonad} commutes. Again, this holds merely from the fact that $\mu^{\moncd}$ is defined exactly as $\mu^{\cset(+\mathbf{1})}$, and the diagram commutes for  $\mu^{\cset(+\mathbf{1})}$ by monadicity of $\cset(+\mathbf{1})$.

\subsection{Proof Theorem \ref{main:theorem:moncd}}

\textbf{Proof of Lemma \ref{lem:kmonadmap}}

    The fact that each $\bcl_X$ is well-typed was proven in Theorem \ref{thm:Kclosure}. We now show that $\bcl$ is natural, i.e., $\cset(f+\mathbf{1})(\bcl_X(S))= \bcl_Y(\cset(f+\mathbf{1})(S))$. Take $S = \cup_{i}\{\distr_i\} \in \cset(X+\mathbf{1})$ and $f: X \rightarrow Y$. Then by applying first Lemma \ref{lem:kunion} and then Lemma \ref{lem:kclosure1} we get
        \begin{align*}
        \cset(f+\mathbf{1})(\bcl_X(S))
        &= \cset(f+\mathbf{1})\left( \bigcup_{i}\bcl_X(\{\distr_i\})\right)\\
        &= \cset(f+\mathbf{1})\left( \bigcup_{i}\left\{ \distrb \mid \distrb \in \dset(X+\mathbf{1})\text{ and }\forall x \in X, \distrb(x) \leq \distr_{i}(x)\right\}\right)\\
        &= \bigcup_{i}\left\{ \dset(f+\oneset)(\distrb)  \mid \distrb \in \dset(X+\mathbf{1})\text{ and }\forall x \in X, \distrb(x) \leq \distr_{i}(x)\right\}
    \end{align*}
On the other side, by the same properties of $\bcl$ we have:
        \begin{align*}
        \bcl_Y(\cset(f+\mathbf{1})(S))
        &= \bcl_Y\left( \bigcup_{i}\{\dset(f+\mathbf{1})(\distr_i)\}\right)\\    
        &= \bigcup_{i}\bcl_Y\left( \{\dset(f+\mathbf{1})(\distr_i)\}\right)\\      
        &= \bigcup_{i}\left\{ \distrc \mid \distrc \in \dset(Y+\mathbf{1})\text{ and }\forall y \in Y, \distrc(y) \leq (\dset(f+\mathbf{1})(\distr_i))(y)\right\}
    \end{align*}
We prove that for any $\distr_{i}\in S$, the sets
\begin{enumerate}
\item $\left\{ \dset(f+\oneset)(\distrb)  \mid \distrb \in \dset(X+\mathbf{1})\text{ and }\forall x \in X, \distrb(x) \leq \distr_{i}(x)\right\}$
\item $\left\{ \distrc \mid \distrc \in \dset(Y+\mathbf{1})\text{ and }\forall y \in Y, \distrc(y) \leq (\dset(f+\mathbf{1})(\distr_i))(y)\right\}$
\end{enumerate}
coincide.
For $1\subseteq 2$, let $\dset(f+\oneset)(\distrb)$ with $\distrb \in \dset(X+\mathbf{1})\text{ and }\forall x \in X, \distrb(x) \leq \distr_{i}(x)$. Then $\dset(f+\oneset)(\distrb)\in \dset(Y+\mathbf{1})$ and for each $y\in Y$, $\dset(f+\oneset)(\distrb)(y)= \sum_{x\in f^{-1}(y)} \distrb(x) \leq \sum_{x\in f^{-1}(y)} \distr_{i}(x) = (\dset(f+\mathbf{1})(\distr_i))(y)$. 
Hence, $\dset(f+\oneset)(\distrb)$ is in the second set. 
For $2\subseteq 1$, 
let $\distrc \in \dset(Y+\mathbf{1})$ with $\forall y \in Y, \distrc(y) \leq (\dset(f+\mathbf{1})(\distr_i))(y)=\sum_{x\in f^{-1}(y)} \distr_{i}(x)$. Then, as $f^{-1}$ partitions $X$, we can assign to each $x$ a probability value $p_{x} \leq \distr_{i}(x)$ such that for all $y$, $\sum_{x\in f^{-1}(y)} p_{x} = \distrc(y)$. This in turn gives the probability distribution $\distrb\in \dset(X+\mathbf{1})$ defined as $\distrb(x)=p_{x}$ and $\distrb(\star) = 1-\sum_{x} p_{x}$, which indeed satisfies that $\forall x \in X, \distrb(x) \leq \distr_{i}(x)$.
As $\distrc= \dset(f+\oneset)(\distrb)$, we conclude that $\distrc$ is in the second set.

    Next, we show commutativity of the monad map diagrams. First, the unit diagram \eqref{diag-monmap1} commutes, because for any set $X$ and $x \in X$, \[(\bcl_X \circ \eta^{\cset(+\mathbf{1})}_{X})(x)= \bcl_X(\{\dirac{x}\}) = \{p\,x+ (1-p)\, \point \mid p \in [0,1]\} = \eta^{\moncd}_{X}(x).\]
    Second, we need to show \eqref{diag-monmap2} commutes, i.e. $\bcl \circ \mu^{\cset(+\mathbf{1})} = \mu^{\moncd} \circ (\bcl \diamond \bcl)$.
    
    Let $S \in \cset(\cset(X+\mathbf{1})+\mathbf{1})$.
    By Lemma \ref{lem:kmu}.1, applying the L.H.S. yields
\begin{equation}\label{eq:kmm1}
    \bcl_X(\mu^{\cset(+\oneset)}_{X}(S))
=\bigcup_{\ddistr \in S} \wms\Big(\sum_{\{\bcl_X(U) \mid U\in \support(\ddistr^{\point})\}}\big(\sum_{U\in \bcl_X^{-1}(\bcl_X(U))}\ddistr^{\point}(U)\big) \, \bcl_X(U)\Big)
\end{equation}
    For the R.H.S., we first apply $\bcl \diamond \bcl$, which yields:
    \begin{align*}
        S &\supp{\cset(\bcl_X+\mathbf{1})} \bigcup_{\ddistr \in S} \{\Big(\sum_{\{\bcl_X(U) \mid U\in \support(\ddistr), U \neq \bpoint\}}\big(\sum_{U\in \bcl_X^{-1}(\bcl_X(U))}\ddistr(U)\big) \, \bcl_X(U)\Big) + \ddistr(\bpoint ) \bpoint\}\\
        &\supp{\bcl_{\cset(X+\mathbf{1})}} \bcl_X(\bigcup_{\ddistr \in S}  \{\Big(\sum_{\{\bcl_X(U) \mid U\in \support(\ddistr), U \neq \bpoint\}}\big(\sum_{U\in \bcl_X^{-1}(\bcl_X(U))}\ddistr(U)\big) \, \bcl_X(U)\Big) + \ddistr(\bpoint ) \bpoint\})
          \end{align*}
       By Lemma \ref{lem:kunion}, the latter is equal to
          $$\bigcup_{\ddistr \in S}  A_{\ddistr}  \text{ with } A_{\ddistr}=\bcl_{\cset(X+\mathbf{1})}(\{\Big(\sum_{\{\bcl_X(U) \mid U\in \support(\ddistr), U \neq \bpoint\}}\big(\sum_{U\in \bcl_X^{-1}(\bcl_X(U))}\ddistr(U)\big) \, \bcl_X(U)\Big) + \ddistr(\bpoint ) \bpoint\})$$
           Then we apply $\mu^{\moncd}_{X}$ and we get by Lemma \ref{lem:kmu}.2
  \begin{equation}\label{eq:kmm2}        
          \bigcup_{\ddistr \in S}  \bigcup_{\Theta \in A_{\ddistr}} \wms(\Theta^{\point})
  \end{equation}
     We first prove that \eqref{eq:kmm1} is included in \eqref{eq:kmm2}.     
     Note that, by definition of $\ddistr^{\point}$,
     \begin{align*}
     &\sum_{\{\bcl_X(U) \mid U\in \support(\ddistr^{\point})\}}\big(\sum_{U\in \bcl_X^{-1}(\bcl_X(U))}\ddistr^{\point}(U)\big) \, \bcl_X(U) \\
     &= \Big(\sum_{\{\bcl_X(U) \mid U\in \support(\ddistr), U \neq \bpoint, U\neq \{\dirac \point\}\}}\big(\sum_{U\in \bcl_X^{-1}(\bcl_X(U))}\ddistr(U)\big) \, \bcl_X(U) \Big)+ \big(\ddistr (\{\dirac \point\})+\ddistr (\bpoint)\big) \{\dirac \point\}\\
     &=\left(\Big(\sum_{\{\bcl_X(U) \mid U\in \support(\ddistr), U \neq \bpoint\}}\big(\sum_{U\in \bcl_X^{-1}(\bcl_X(U))}\ddistr(U)\big) \, \bcl_X(U)\Big) + \ddistr(\bpoint ) \bpoint\right)^{\point}
     \end{align*}
     Hence, as 
     $$\Big(\sum_{\{\bcl_X(U) \mid U\in \support(\ddistr),U \neq \bpoint\}}\big(\sum_{U\in \bcl_X^{-1}(\bcl_X(U))}\ddistr(U)\big) \, \bcl_X(U)\Big) + \ddistr(\bpoint ) \bpoint \;\in A_{\ddistr}$$
    we derive that there is a $\Theta^{\point}\in A_{\ddistr}$ such that
    $$\sum_{\{\bcl_X(U) \mid U\in \support(\ddistr^{\point})\}}\big(\sum_{U\in \bcl_X^{-1}(\bcl_X(U))}\ddistr^{\point}(U)\big) \, \bcl_X(U)= \Theta^{\point}$$
    Then we conclude that \eqref{eq:kmm1} is included in \eqref{eq:kmm2}. For the converse inclusion,
     let $\theta\in \eqref{eq:kmm2}$. Then $\theta\in\wms(\Theta^{\point})$ for some $\Theta \in A_{\ddistr}$ and for some $\ddistr \in S$. 
    By definition of $A_{\ddistr}$, the distribution $\Theta\in \dset(\cset(X+1)+1)$ is such that, by Lemma \ref{lem:kclosure1},
    for each $U'\in \cset (X+\oneset)$ it holds
    \begin{equation}\label{eq:kmm22}
    \Theta(U')\leq \left(\Big(\sum_{\{\bcl_X(U) \mid U\in \support(\ddistr), U \neq \bpoint\}}\big(\sum_{U\in \bcl_X^{-1}(\bcl_X(U))}\ddistr(U)\big) \, \bcl_X(U)\Big) + \ddistr(\bpoint ) \bpoint\right) (U')
    \end{equation}
    By definition, $\Theta^{\point}$ equals 
    $$\Big(\sum_{\{\bcl_X(U) \mid U\in \support(\ddistr),U\neq \{\dirac\point\}, U \neq \bpoint\}}\Theta(\bcl_X(U)) \, \bcl_X(U)\Big) + \Big(1- (\sum_{\{\bcl_X(U) \mid U\in \support(\ddistr),U\neq \{\dirac\point\}, U \neq \bpoint\}}\Theta(\bcl_X(U)))\Big) \{\dirac\point\}$$
%
By definition of $\ddistr^{\point}$, this is equivalent to
    $$\Theta^{\point}= \Big(\sum_{\{\bcl_X(U) \mid U\in \support(\ddistr^{\point}),U\neq \{\dirac\point\}\}}\Theta(\bcl_X(U)) \, \bcl_X(U)\Big) + \Big(1- (\sum_{\{\bcl_X(U) \mid U\in \support(\ddistr^{\point}),U\neq \{\dirac\point\}\}}\Theta(\bcl_X(U)))\Big) \{\dirac\point\}$$

Then, as $\theta\in\wms(\Theta^{\point})$, there are distributions $\{\distr_{\bcl_X(U)} \mid U\in \support(\ddistr) \text{ and } U\neq \{\dirac\point\}\}$ such that
$$\theta=\Big(\sum_{\{\bcl_X(U) \mid U\in \support(\ddistr^{\point}),U\neq \{\dirac\point\}\}}\Theta(\bcl_X(U)) \cdot \distr_{\bcl_X(U)}\Big) + \Big(1- (\sum_{\{\bcl_X(U) \mid U\in \support(\ddistr^{\point}),U\neq \{\dirac\point\}\}}\Theta(\bcl_X(U)))\Big)\cdot \dirac\point$$
By \eqref{eq:kmm22} and the definition of $\ddistr^{\point}$, it holds that for each $\bcl_X(U)$ such that $U\in \support(\ddistr^{\point}) \text{ and } U\neq \{\dirac\point\}$,
$$
\Theta(\bcl_X(U))\leq \sum_{U\in \bcl_X^{-1}(\bcl_X(U))}\ddistr^{\point}(U)
$$
Hence, for each $x$ it holds that $\theta(x)\leq \psi(x)$, for $\psi$ the distribution
\begin{align*}
&\big(\sum_{\{\bcl_X(U) \mid U\in \support(\ddistr^{\point}),U\neq \{\dirac\point\}\}} (\sum_{U\in \bcl_X^{-1}(\bcl_X(U))}\ddistr^{\point}(U)) \cdot \distr_{\bcl_X(U)}\big) \\
\quad &+ \left(1-\Big(\sum_{x}\Big(\sum_{\{\bcl_X(U) \mid U\in \support(\ddistr),U\neq \{\dirac\point\}\}} \big(\sum_{U\in \bcl_X^{-1}(\bcl_X(U))}\ddistr^{\point}(U)\big) \cdot \distr_{\bcl_X(U)}\Big) (x) \Big) \right)\cdot \dirac\point
\end{align*}
As $\bcl_X(\{\dirac \point\})=\{\dirac \point\}$, we have
$$
\psi \in \wms\Big(\sum_{\{\bcl_X(U) \mid U\in \support(\ddistr^{\point})\}}\big(\sum_{U\in \bcl_X^{-1}(\bcl_X(U))}\ddistr^{\point}(U)\big) \, \bcl_X(U)\Big)
$$
and thus, as $\theta(x)\leq \psi(x)$ for all $x$,
by Lemma \ref{lem:kclosure1} we derive that $\theta$ is in the set
$$\bcl_X\left(\wms\Big(\sum_{\{\bcl_X(U) \mid U\in \support(\ddistr^{\point})\}}\big(\sum_{U\in \bcl_X^{-1}(\bcl_X(U))}\ddistr^{\point}(U)\big) \, \bcl_X(U)\Big)\right)$$
By $\bcl_X$ being a homomorphism and by $\bcl_X(\bcl_X(U))=U$ (Theorem \ref{thm:Kclosure}), we have
\begin{align*}
&\bcl_X\left(\wms\Big(\sum_{\{\bcl_X(U) \mid U\in \support(\ddistr^{\point})\}}\big(\sum_{U\in \bcl_X^{-1}(\bcl_X(U))}\ddistr^{\point}(U)\big) \, \bcl_X(U)\Big)\right)\\
&=\wms\Big(\sum_{\{\bcl_X(U) \mid U\in \support(\ddistr^{\point})\}}\big(\sum_{U\in \bcl_X^{-1}(\bcl_X(U))}\ddistr^{\point}(U)\big) \, \bcl_X(\bcl_X(U))\Big)\\
&=\wms\Big(\sum_{\{\bcl_X(U) \mid U\in \support(\ddistr^{\point})\}}\big(\sum_{U\in \bcl_X^{-1}(\bcl_X(U))}\ddistr^{\point}(U)\big) \, \bcl_X(U)\Big)
\end{align*}
This set is included in \eqref{eq:kmm1}, so we have proved that for each $\theta \in$ \eqref{eq:kmm2}, $\theta \in$ \eqref{eq:kmm1}.

\textbf{Proof of Lemma \ref{moncd-final-lemma}.1}

Let $(A, \alpha)\in \EM(\moncd)$ and let $(A, \alpha\circ \bcl_A) \in \EM(\cset(+\oneset))$ be its embedding via $U^{\bcl}$. 
By definition of the pointed convex semilattice $P((A, \alpha\circ \bcl_A))$ and by definition of $\bcl$, for any $a\in A$ it holds
\begin{align*}
a \oplus^{\alpha \circ \bcl_A} \point^{\alpha \circ \bcl_A} &=(\alpha \circ \gamma_A)(\cc{\dirac{a}, \dirac{\alpha \circ \bcl_A(\{\dirac{\point}\})}}) & \text{definition of $P$}\\
&= \alpha \circ \bcl_A (\cc{\dirac{\alpha \circ \bcl_A(\{\dirac{a}\})}, \dirac{\alpha \circ \bcl_A(\{\dirac{\point}\})}}) &\text{definition of $\EM(\cset(+\oneset))$}\\
&= \alpha \circ \bcl_A \circ \mu^{\cset(+\oneset)}_{A}(\cc{\dirac{\{\dirac{a}\}}, \dirac{\{\dirac{\point}\}}}) &\text{definition of $\EM(\cset(+\oneset))$}\\
&= \alpha \circ \bcl_A (\cc{\dirac{a}, \dirac{\point}}) &\text{definition of $\mu^{\cset(+\oneset)}$}\\
&=\alpha (\cc{\dirac{a}, \dirac{\point}}) &\text{definition of $\bcl$}\\
&=a &\text{definition of $\EM(\moncd)$}
\end{align*}
Hence, the pointed convex semilattice $P((A, \alpha\circ \bcl_A))$ satisfies the $\bot$ equation, and therefore it belongs to $\acat(\etcsb)$.

\textbf{Proof of Lemma \ref{moncd-final-lemma}.2}

Let $\mathbb{A}\in \acat(\etcsb)$, which is embedded via $\iota$ to  $\mathbb{A}\in\acat(\etpcs)$.
We want to show that  $P^{-1}(\mathbb{A})=(A,\alpha)\in \EM(\cset(+\oneset))$ is in the image of $U^{\bcl}$, i.e., that there exists an algebra
$(A, \alpha') \in \EM(\moncd)$ such that $\alpha=\alpha' \circ \bcl_{A}$.
We show that by taking as $\alpha'$ the restriction of $\alpha$ to $\bot$--closed sets, i.e., by letting $\alpha'=\alpha|_{\cset^{\downarrow}(A)}$, we have that $(A, \alpha') \in \EM(\moncd)$ and $\alpha=\alpha' \circ \bcl_{A}$.

We first prove that $(A, \alpha')$ is a $\cset^{\downarrow}$-algebra, i.e., that \eqref{diag-algunit} and \eqref{diag-algmult} commute. 
Note that, by definition of $P^{-1}$ and by $\alga$ satisfying the bottom axiom, for each $a\in A$ the following equation holds.
$$\alpha(\conv\{\dirac{a}, \dirac{\point}\}) = a \oplus^{\alga} \point^{\alga}  \stackrel{(\bot)}{=} a $$
Then the unit diagram commutes by
\[\alpha|_{\cset^{\downarrow}(A)}(\eta^{\cset^{\downarrow}}(a)) = \alpha(\conv\{\dirac{a}, \dirac{\point}\}) = a.\]
The multiplication diagram commutes because it is a restriction of the multiplication diagram for $(A, \alpha)$.

It is left to show that $\alpha|_{\cset^{\downarrow}} \circ \bcl_A = \alpha$, i.e., that $\alpha(\bcl(S)) = \alpha(S)$ for any $S \in \cset(A+\mathbf{1})$. 

We need the following lemma generalizing Lemma \ref{lem-intuitioneqncbot}.
\begin{lemma}\label{lem-techprescbot}
The following equation is derivable in $\etcsb$:
$$\bigpplus_{0\leq i\leq n} p_{i} x_{i}= \bigoplus_{F\subseteq \{1,...,n\}} \Big((\bigpplus_{i\in F} p_{i} x_{i}) + (1-(\sum_{i\in F} p_{i}))\point \Big)$$
\end{lemma}
\begin{proof}
First, we note that, by iterating the distributivity axiom (D), we derive in the theory of convex semilattices that:
\[\bigpplus_{1\leq i\leq k} p_{i}\, (t^{i}_{1}\oplus ...\oplus t^{i}_{n_{i}}) =
\bigoplus_{(t_{1},...,t_{k})\in \{(t_{1},...,t_{k})|\, t_{i} \in \{t^{i}_{1},...,t^{i}_{n_{i}}\}\}} (\bigpplus_{1\leq i\leq k} p_{i} \,t_{i})\]
and this law can be alternatively written as follows, whenever for each $i$ we have a set of terms $S_{i}$:
\[\bigpplus_{1\leq i\leq k} p_{i}\, (\bigoplus _{t\in S_{i}} t)= 
\bigoplus_{f\in \{f: \{1,...,k\} \to \terms X \sigpcs |\, f({i}) \in S_{i}\}} (\bigpplus_{1\leq i\leq k} p_{i}\, f(i))\]
where $\{f: \{1,...,k\} \to \terms X \sigpcs |\, f({i}) \in S_{i}\}$ is the set of functions choosing one term in each $S_{i}$.

By the bottom axiom, we have
$$\bigpplus_{0\leq i\leq n} p_{i} x_{i}=\bigpplus_{0\leq i\leq n} p_{i} (x_{i} \oplus \point)$$
Then by applying the iterated version of the distributivity axiom (D) shown above we have
$$\bigpplus_{0\leq i\leq n} p_{i} (x_{i} \oplus \point)
=\bigoplus_{f\in \{f:\{1,...n\} \to \terms X \sigpcs \mid f(i)\in\{x_{i},\point\}\}} \big(\bigpplus_{0\leq i\leq n} p_{i} f(i)\big)
$$
which can be in turn proved equal to
$$\bigoplus_{F\subseteq \{1,...,n\}} \Big((\bigpplus_{i\in F} p_{i} x_{i}) + (1-(\sum_{i\in F} p_{i}))\point \Big).$$

\end{proof}

Moreover, we have the following characterization of $\bcl_X(S)$, which explicits a finite base for the set.
\begin{lemma}\label{lem:kset}
    Let $S = \conv (\bigcup_{0\leq i \leq n} \{\distr_{i}\}) \in \cset(X+\mathbf{1})$. Then 
    $$\bcl_X(S)= \conv\Big(\bigcup_{0\leq i \leq n} \bigcup_{F\subseteq \support(\distr_{i})\backslash \{\star\}}\{\distr_{i|_{F}}\}\Big)$$
    where for any $\distr$ and $F\subseteq \support(\distr)\backslash \{\star\}$ we define
    $$\distr_{|_{F}}=(\sum_{x\in F} \distr(x) x) + (1-(\sum_{x\in F} \distr)) \point.$$
\end{lemma}
\begin{proof}
We first prove that for any $\distr$, 
$$\bcl_X(\{\distr\})=\conv\Big(\bigcup_{F\subseteq \support(\distr)\backslash \{\star\}}\{\distr_{|_{F}}\}\Big).$$
By Lemma \ref{lem:kclosure1}, this is equivalent to proving
$$\left\{ \distrb \in \dset(X+\mathbf{1}) \mid \forall x \in X, \distrb(x) \leq \distr(x)\right\}
        =\conv\Big(\bigcup_{F\subseteq \support(\distr)\backslash \{\star\}}\{\distr_{|_{F}}\}\Big).$$
We first prove the right-to-left set inclusion.
For any $F \subseteq \support(\distr)\backslash \{\star\}$, $\distr_{|_{F}}$ is such that for all $x\in X$,
\[\distr_{|_{F}}(x) =\begin{cases}
    \distr(x) &x \in F\\
    0 &\text{o/w}
\end{cases} \leq \distr(x).\]
Then by $\bcl_X(\{\distr\})$ being convex we conclude the $\supseteq$ inclusion.
For the converse inclusion,
note that for any $p \in [0,1]$, $\distrb_1, \distrb_2 \in \dset(X+\mathbf{1})$ and $x \in X$, we have \[p\distrb_1(x)+ (1-p)\distrb_2(x) \leq \max\{\distrb_1(x), \distrb_2(x)\}.\]
        Hence, all convex combinations of elements in $\{\distr|_F : F \subseteq \support(\distr)\backslash \{\star\}\}$ have less weight at $x$ than $\distr(x)$ for any $x \in X$. This implies the $\subseteq$ inclusion.

Then, for $S = \conv (\bigcup_{0\leq i \leq n} \{\distr_{i}\})$, as $\bcl_X$ commutes over convex union (by being a homomorphism) and by $\conv(\cup_{i} S_{i})=\conv (\cup_{i}\conv (S_{i}))$,
we conclude that
\begin{align*}
    \bcl_X(S)&= \conv\Big(\bigcup_{0\leq i \leq n} \bcl_X(\{\distr_{i}\})\Big)\\
    &= \conv\Big(\bigcup_{0\leq i \leq n} \conv(\bigcup_{F\subseteq \support(\distr_{i})\backslash \{\star\}}\{\distr_{i|_{F}}\})\Big)\\
    &= \conv\Big(\bigcup_{0\leq i \leq n} \bigcup_{F\subseteq \support(\distr_{i})\backslash \{\star\}}\{\distr_{i|_{F}}\}\Big)
\end{align*}
\end{proof}

Now, let $S = \conv (\bigcup_{0\leq i \leq n} \{\distr_{i}\}) \in \cset(A+\mathbf{1})$, with $\{\distr_{i}\}$ the unique base for $S$. By applying the lemmas and the definition of $\alpha$, we can now show that $\alpha(\bcl_A(S)) = \alpha(S)$:
\begin{align*}
    \alpha(S) &= \Big(\bigoplus_{0\leq i \leq n} \big(\bigpplus_{x\in \support (\distr_{i})} \distr_{i}(x) x\big)\Big)^{\alga}& \text{definition of $\alpha$}\\
    &= \Big(\bigoplus_{0\leq i \leq n}\big(\bigoplus_{F \subseteq \support(\distr_{i})} \big((\bigpplus_{x\in F} \distr_{i}(x) x) + (1-(\sum_{x\in F} \distr_{i}(x)))\point\big)\big)\Big)^{\alga} &\text{Lemma \ref{lem-techprescbot}}\\
    &= \Big(\bigoplus_{0\leq i \leq n}\big(\bigoplus_{F \subseteq \support(\distr_{i})\backslash\{\star\}} \big((\bigpplus_{a\in F} \distr_{i}(a) a) + (1-(\sum_{a\in F} \distr_{i}(a)))\point\big)\big)\Big)^{\alga} &\text{idempotency axiom $(I)$}\\
    &= \alpha\Big(\conv\Big(\bigcup_{0\leq i \leq n} \bigcup_{F\subseteq \support(\distr_{i})\backslash \{\star\}}\{\distr_{i|_{F}}\}\Big)\Big) &\text{definition of $\alpha$}\\
    &= \alpha(\bcl_A(S)) &\text{Lemma \ref{lem:kset}}
\end{align*}
Note that in the above derivation $x$ ranges over $\support{(\distr_{i})}\subseteq A+\oneset$.

%% file: App_presentationsinmet.tex
\subsection{Proofs for Section \ref{sec:met:cplusone}}

We prove that the multiplication $\mu^{\moncb}$ of the $\Sets$ monad $\moncb$ is not non--expansive. 

\begin{lemma}
Given a metric space $(X,d)$, the function $\mu^{\moncb}_{(X,d)}: (\moncbp {\moncbp X}, \hk(\hk(d)+\donemet)+\donemet) \to ( \moncbp{X}, \hk(d)+\donemet)$ is not non-expansive. 
\end{lemma}
\begin{proof}
We give a counterexample to non-expansiveness. Let $X$ be endowed with the discrete metric and take $S_{1}=\{ \onehalf\,\{\dirac x\} +\onehalf\, \star\}$ and $S_{2}= \{\dirac {\{\dirac x\}}\}$. Then we obtain
$$\hk(d)+\donemet (\mu^{\moncb}_{(X,d)}(S_{1}), \mu^{\moncb}_{(X,d)}(S_{2}))= \hk(d)+\donemet (\star, \{\dirac x\})=1$$
and 
$$\hk(\hk(d)+\donemet)+\donemet (S_{1}, S_{2}) =\onehalf$$
\end{proof}

\textbf{Proof of Theorem \ref{thm:bhtrivial}}

The deductive system of quantitative equational logic is the one of \cite[\S 3]{DBLP:conf/lics/BacciMPP18}.

Let $p\in (0,1)$. From the axioms $\vdash x=_{0} x $ and $\vdash x=_{1} \star$, we derive by the (K) rule $\vdash x\pplus p \star =_{1-p} x\pplus p x$. By $(I_{p})$, we have $\vdash x\pplus p x =_{0} x$, so by triangular inequality we derive $\vdash x\pplus p \star =_{(1-p)} x$. Now, by the $\qbh$ axiom ($\vdash x\pplus p \star =_{0} \star$), symmetry and triangular inequality, we have $\vdash \star=_{(1-p)} x$. Since $p \in (0,1)$ was arbitrary, we have equivalently derived that $\vdash \star=_{p} x$ belongs to $\qet$. For any $y$, we analogously obtain $\vdash \star =_{p} y$. Then, by symmetry and triangular inequality we derive $\vdash x=_{p} y$ for all $p\in (0,1)$, and by (Max) we have $\vdash x=_{\epsilon} y$ for all $\epsilon >0$. We conclude by applying (Arch) $\{x=_{\epsilon} y\}_{\epsilon>0} \vdash x=_{0} y$.

\subsection{Proofs for Section \ref{sec:met:cdownarrow}}

\textbf{Proof of Lemma \ref{lem:met:kne}}

By definition, $\bcl_{X}$ is the unique pointed semilattice homomorphism extending $f:X\to \cset(X+\oneset)$, with $f(x)=\conv(\{\dirac x, \dirac \point\})$. We first show that the function $\hat f: (X,d) \to (\cset(X+\oneset), \hk(d)+\oneset)$, defined as $f$ on $X$, is an isometry, i.e., $(\hk(d)+\oneset)(f(x),f(y))=d(x,y)$. To see this, note that 
$$f(x)=\{p\,x +(1-p)\, \star\mid p\in [0,1]\}\qquad f(y)=\{p\,y +(1-p)\, \star\mid p\in [0,1]\}$$ 
and that for every $p,q\in [0,1]$, 
\begin{equation}\label{eq:kne}
\kant(d)(p\,x +(1-p)\, \star, q\,y +(1-q)\, \star) \geq p\cdot d(x,y)
\end{equation}
Indeed, if $p\leq q$ then 
\begin{align*}
\kant(d)(p\,x +(1-p)\, \star, q\,y +(1-q)\, \star)&= p \cdot d(x,y) +(q-p) \cdot (d(\star, y)) + (1-q) \cdot d(\star, \star) \\
&=p \cdot d(x,y)+(q-p) \\
&\geq p\cdot d(x,y)
\end{align*} 
and if $p=q+r > q$ then 
\begin{align*}
\kant(d)(p\,x +(1-p)\, \star, q\,y +(1-q)\, \star)&= q \cdot d(x,y) +(p-q) \cdot (d(x, \star)) + (1-q) \cdot d(\star, \star) \\
&= q\cdot d(x,y) +(p-q)\\
&= q\cdot d(x,y) + r \\
&\geq q\cdot d(x,y) + r \cdot d(x,y)\\
&= p \cdot d(x,y).
\end{align*}
We derive by \eqref{eq:kne} that for any $\distr=p\,x +(1-p)\, \star \in f(x)$ it holds
$$\inf_{\distrb \in f(y)} \kant(d) (\distr,\distrb)= \kant(d) (p\,x +(1-p)\, \star,p\,y +(1-p)\, \star)=p\cdot d(x,y)$$
and thus 
$$\sup_{\distr \in f(x)}\inf_{\distrb \in f(y)} \kant(d) (\distr,\distrb)= \kant(d) (\dirac x, \dirac y)= d(x,y).$$

Symmetrically, we obtain 
$$\sup_{\distrb\in f(y)}\inf_{\distr\in f(x)} \kant(d) (\distr,\distrb)= \kant(d) (\dirac x, \dirac y)= d(x,y).$$

We can now conclude
\begin{align*}
(\hk(d)+\oneset)(f(x),f(y))&= \max\{\sup_{\distr\in f(x)}\inf_{\distrb\in f(y)} \kant(d) (\distr,\distrb), \sup_{\distrb\in f(y)}\inf_{\distr\in f(x)} \kant(d) (\distr,\distrb)\}\\
&= d(x,y).
\end{align*}

Since $\hat f$ is an isometry, it is non-expansive, and thus a morphism in $\Met$.
Given a metric space $(X,d)$, the metric space $((\cset(X+\oneset), \hk(d)+\oneset)$ equipped with the operations of convex union, weighted Minkowski sum and $\{\dirac \point\}$ (respectively interpreting $\oplus$, $\pplus p$, and $\point$) is the free quantitative pointed convex semilattice on $(X,d)$.
As $((\cset(X+\oneset), \hk(d)+\oneset)$ is free, there is a unique quantitative pointed convex semilattice homomorphism extending $\hat f$. 
It follows from the uniqueness of $\bcl_{X}$ and the definition of  $\hat \bcl_{(X,d)}$ that $\hat \bcl_{(X,d)}$ is the unique quantitative pointed convex semilattice homomorphism extending $\hat f$. Hence, as $\hat \bcl_{(X,d)}$ is a morphism in $\Met$, it is non-expansive.\qed

\textbf{Proof of Theorem \ref{thm:downarrowmonadmet}}

First, we prove that $\eta^{\ldownarrow}$ and $\mu^{\ldownarrow}$ are natural transformations in $\Met$, i.e., that the naturality diagrams commute and that for any $(X,d)$,  $\eta^{\ldownarrow}_{(X,d)}$ and $\mu^{\ldownarrow}_{(X,d)}$ are non-expansive. 
As the unit $\eta^{\ldownarrow}$  and multiplication $\mu^{\ldownarrow}$ are respectively defined as the unit $\eta^{\cset^\downarrow}$  and multiplication $\mu^{\cset^\downarrow}$ of the $\Sets$ monad $\moncd$, and as we know that the naturality diagrams  commute for $\eta^{\ldownarrow}$ and $\mu^{\ldownarrow}$, we derive that they also commute for $\eta^{\ldownarrow}$  and $\mu^{\ldownarrow}$. 
As $\eta^{\ldownarrow}_{(X,d)}=\hat \bcl_{(X,d)} \circ \eta^{\lcset(+\oneset)}_{(X,d)}$, non-expansiveness of $\eta^{\ldownarrow}_{(X,d)}$ follows directly from non-expansiveness of $\hat \bcl_{(X,d)}$ (Lemma \ref{lem:met:kne}) and non-expansiveness of $\eta^{{\lcset(+\oneset)}}_{(X,d)}$.
As $\mu^{\ldownarrow}$ is defined as $\mu^{\cset^\downarrow}$, which in turn is the restriction of $\mu^{\cset(+\oneset)}$ to $\bot$--closed sets, and as 
$\mu^{{\lcset(+\oneset)}}$ is defined as $\mu^{\cset(+\oneset)}$, we have that $\mu^{\ldownarrow}$ is the restriction of $\mu^{\lcset(+\oneset)}$ to metric spaces whose sets are $\bot$--closed.
Then non-expansiveness of $\mu^{\ldownarrow}_{(X,d)}$ follows from non-expansiveness of $\mu^{\lcset(+\oneset)}_{(X,d)}$.
To conclude, it remains to verify that $\eta^{\ldownarrow}$  and $\mu^{\ldownarrow}$ satisfy the monad laws \eqref{diag-unitmonad} and \eqref{diag-multmonad}. This follows as $\eta^{\ldownarrow}$  and $\mu^{\ldownarrow}$ are respectively defined as the unit $\eta^{\cset^\downarrow}$  and multiplication $\mu^{\cset^\downarrow}$ of the $\Sets$ monad $\moncd$, and as we know that monad laws \eqref{diag-unitmonad} and \eqref{diag-multmonad} hold for $\eta^{\ldownarrow}$ and $\mu^{\ldownarrow}$, we derive that the laws also hold for $\eta^{\ldownarrow}$  and $\mu^{\ldownarrow}$. \qed

\textbf{Proof of Theorem \ref{thm:maincdownarrowmet}} 

The structure of the proof of Theorem \ref{thm:maincdownarrowmet} is very similar to that of Theorem \ref{main:theorem:moncd} and is based on the following technical lemmas.

\begin{lemma}\label{lem:kmonadmapmet}
    The family $\hat \bcl_{(X,d)}: ((\cset(X+\oneset), \hk(d)) \rightarrow (\lmoncd(X), \hk(d))$ is a monad map from the monad ${(\lcset(+\oneset))}$ to the monad $\lmoncd$.
\end{lemma}
\begin{proof}
By Lemma \ref{lem:met:kne}, $\hat \bcl_{(X,d)}$ is non-expansive, so it is a morphism in $\Met$. As $\hat \bcl_{(X,d)}$ is defined as $\bcl_{X}$ on $X$, by Lemma \ref{lem:kmonadmap} it satisfies the monad map laws \eqref{diag-monmap1} and \eqref{diag-monmap2}. 
\end{proof}

\begin{lemma}\label{embedding_lemma_cbot_met}
    There is a functor $U^{\hat \bcl}: \EM(\lmoncd) \rightarrow \EM({\lcset(+\mathbf{1})})$ defined on objects by $((A, d), \alpha) \mapsto ((A, d), \alpha \circ \hat \bcl_{(A, d)})$ and acting as identity on morphisms which is an embedding.
\end{lemma}
\begin{proof}
    The fact that $U^{\hat \bcl}$ is a functor follows from Lemma \ref{lem:kmonadmapmet} and  Proposition \ref{prop-MonmapFunctor}. Fully faithfulness follows as $U^{\hat \bcl}$ acts like the identity on morphisms. Injectivity on objects follows from surjectivity of $\hat \bcl_{(X,d)}$ for any metric space $(X,d)$, which in turn follows from surjectivity of $\bcl_{X}$ (Theorem \ref{thm:Kclosure}). Indeed, if $U^{\hat \bcl}(((A, d), \alpha))=U^{\hat \bcl}((A', d'), \alpha')$ then $A=A'$, $d=d'$, and $\alpha \circ \hat \bcl_{(A, d)} = \alpha' \circ \hat \bcl_{(A, d)}$, which in turn implies by surjectivity of $\hat \bcl_{(A,d)}$ that $\alpha = \alpha'$.
\end{proof}

 And lastly we obtain the isomorphism of the two categories  $\EM(\lmoncd)$ and $\qacat( \qetcsb)$ by restricting the isomorphisms of $\EM(\lcset(+\oneset))$ and $\qacat( \qetcs)$ witnessing the presentation of the $\Met$ monad $\lcset(+\oneset)$ with the theory of quantitative pointed convex semilattices. This amounts to proving the following two points:

\begin{lemma}\label{lmoncd-final-lemma}The following hold:
\begin{enumerate}
    \item Given any $((A,d), \alpha)\in \EM(\lmoncd)$, which is embedded via $U^{\hat \bcl}$ to $((A, d), \alpha\circ \hat \bcl_{(A, d)}) \in \EM({\lcset(+\oneset)})$, the quantitative pointed convex semilattice $P(((A, d), \alpha\circ \hat \bcl_{(A,d)}))$ satisfies the $\qbot$ quantitative equation, and therefore it belongs to $\qacat(\qetcsb)$.
    \item Given $\mathbb{A}\in \qacat(\qetcsb)$, which is embedded via $\iota$ to  $\mathbb{A}\in\qacat(\qetpcs)$, the Eilenberg-Moore algebra $P^{-1}(\mathbb{A})\in \EM({\lcset(+\oneset)}) $ belongs to the subcategory  $\EM(\lmoncd)$, i.e., it is in the image of $U^{\hat \bcl}$.
    \end{enumerate}
    \end{lemma}
    \begin{proof}
For item (1), let $((A,d), \hat \alpha)\in \EM(\lmoncd)$. Then $(A, \alpha)\in \EM(\moncd)$, where $\alpha$ is $\hat\alpha$ seen as a $\Sets$ function, and by Lemma  \ref{moncd-final-lemma}.1 we know that the pointed convex semilattice $P((A, \alpha\circ \bcl_A))$ satisfies the $\bot$ equation.
By definition of $\hat P$, the interpretation of the pointed convex semilattice operations in $\hat P(((A, d), \hat\alpha\circ \hat \bcl_{(A,d_{A})}))$ is the same as in $P((A, \alpha\circ \bcl_A))$, thus $\hat P(((A, d), \hat\alpha\circ \hat \bcl_{(A,d_{A})}))$ satisfies the $\qbot$ quantitative equation as well.

For item (2), let $\mathbb{A}= (A, \oplus^{\alga}, {\pplus p}^{\alga}, \star^{\alga}, d)\in \qacat(\qetcsb)$, which we see as a quantitative pointed convex semilattice via the embedding $\iota$, and let $\hat P^{-1}(\mathbb{A})=((A,d),\hat\alpha)\in \EM({\lcset(+\oneset)})$.
We show that $((A,d_{A}),\hat\alpha)$ is in the image of $U^{\hat \bcl}$ by proving that $\hat\alpha=\hat\alpha|_{\cset^{\downarrow}(A)} \circ \hat \bcl_{(A,d)}$ and $((A, d), \hat \alpha|_{\cset^{\downarrow}(A)}) \in \EM(\lmoncd)$, with $\hat\alpha|_{\cset^{\downarrow}(A)}: (\cset^{\downarrow}(A), \hk(d)) \to (A,d)$ defined as the restriction of $\hat\alpha$ to $\bot$--closed sets.

First, note that $(A, \oplus^{\alga}, {\pplus p}^{\alga}, \star^{\alga})$ is a pointed convex semilattice. 
By the definition of $\hat P^{-1}$, we have that $P^{-1}((A, \oplus^{\alga}, {\pplus p}^{\alga}, \star^{\alga}))=(A,\alpha)$, where $\alpha$ is $\hat\alpha$ seen as a $\Sets$ function. By the proof of Lemma \ref{moncd-final-lemma}.2, we know that $\alpha=\alpha|_{\cset^{\downarrow}(A)} \circ \bcl_{A}$, with $(A, \alpha|_{\cset^{\downarrow}(A)}) \in \EM(\moncd)$.
As $\hat \bcl_{(A,d)}$ is defined as $\bcl_{A}$, we derive from $\alpha=\alpha|_{\cset^{\downarrow}(A)} \circ \bcl_{A}$ in $\Sets$ that $\hat \alpha=\hat \alpha|_{\cset^{\downarrow}(A)} \circ \hat \bcl_{(A,)}$ in $\Met$. 
From $(A, \alpha|_{\cset^{\downarrow}(A)}) \in \EM(\moncd)$ we derive that $((A, d), \hat \alpha|_{\cset^{\downarrow}(A)})$ satisfies the laws \eqref{diag-algunit} and \eqref{diag-algmult} for $\moncd$-algebras. Moreover, as $\hat\alpha$ is non-expansive, also its restriction $\hat\alpha|_{\cset^{\downarrow}(A)}$ is non-expansive. 
Hence, $((A, d), \hat\alpha|_{\cset^{\downarrow}(A)}) \in \EM(\lmoncd)$.
\end{proof}

%% file: App_examples.tex

We recall and establish some results needed for the proof of soundness and completeness of the proof system (Theorem \ref{thm:soundcomp}).
In what follows, we fix a theory $\et\in  \{\etpcs, \etcsbb, \etcsb \}$ and the corresponding monad $M \in \{ \cset(+\oneset), \cset +\oneset, \cset^\downarrow \}$ it presents, with isomorphism $P: \EM(M) \cong \acat(\et) :P^{-1}$ (following our presentation results for $\Sets$ monads). As standard, we let $M$ denote both the monad and the functor underlying the monad.

As the theory $\et$ presents the monad $M$, the monad $(M, \eta^{M},\mu^{M})$ and the term monad $(\termmon, \eta^{\termmon}, \mu^{\termmon})$ are isomorphic.
This means that there is a monad map $\sigma: \termmon \Rightarrow M$ which is an isomorphism, with $\sigma_{X}$ mapping each equivalence class $[t]_{\!/\et}\in \termmon(X)$ to the corresponding element of $M(X)$.

Hence, for the considered functor $M$, the coalgebra $\tau_{M}$ maps a process term $P$ to the element of $M(\pccs)$ defined as $\tau_{M}(P)=\sigma_{\pccs}([\tau(P)]_{\!/\et})$
and we can instantiate the definition of behavioural equivalence (Definition \ref{def:bis}) on process terms as follows:
$R$ is a behavioural equivalence if for all $P,Q\in R$ it holds that
$M(q_{R}) (\sigma_{\pccs}([\tau(P)]_{\!/\et})) = M(q_{R}) (\sigma_{\pccs}([\tau(Q)]_{\!/\et})).$

By the properties of the isomorphism monad map $\sigma$ we have the following lemma.
\begin{lemma}
\label{lem:isomm}
Let $f:X\to Y$ and let $\sigma: \termmon \Rightarrow M$ be an isomorphism monad map. Then for all $t,t'\in \terms {X} {\sigpcs}$ it holds:
\[ \termmon (f) ([t]_{\!/\et})=\termmon (f) ([t']_{\!/\et}) \text{ iff } M(f)(\sigma_{X}([t]_{\!/\et}))= M(f)(\sigma_{X}([t']_{\!/\et}))\]
\end{lemma}
\begin{proof}
By the naturality of $\sigma$ it holds that for any $f:X\to Y$ and for any term $t$,  
\[\sigma_{X}(\termmon (f) ([t]_{\!/\et}))=M(f)(\sigma_{X}([t]_{\!/\et})).\]
As $\sigma_{X}$ is an isomorphism we have that for all $t,t' \in \terms X \sigpcs$:
\[[t]_{\!/\et}=[t']_{\!/\et} \text{ iff }\sigma_{X}([t]_{\!/\et})=\sigma_{X}([t']_{\!/\et}) \]
which allows us to conclude:
\[\termmon (f) ([t]_{\!/\et})=\termmon (f) ([t']_{\!/\et})
\text{ iff }\sigma_{X}(\termmon (f) ([t]_{\!/\et}))=\sigma_{X}(\termmon (f) ([t']_{\!/\et})) 
\text{ iff }M(f)(\sigma_{X}([t]_{\!/\et}))= M(f)(\sigma_{X}([t']_{\!/\et})).\]
\end{proof}

Given an equivalence relation $R\subseteq \pccs\times\pccs$ and a term $t\in \terms {\pccs} {\sigpcs}$, where $P_{1},...,P_{n}$ are the $\pccs$ terms occurring in $t$, we let $t_{\!/R} \in \terms {\pccs_{\!/R}} {\sigpcs}$ denote the term obtained by substituting $P_{i}$ with the equivalence class $[P_{i}]_{\!/R}$, for every $i$ from $1$ to $n$.

\begin{lemma}
\label{lem:prov1}
For all $t,t'\in \terms {\pccs} {\sigpcs}$,
\[\et \ \uplus \ \prov \ \vdash_E \  t = t' \text{ iff } \et  \ \vdash_E \  t_{\!/\prov} = t'_{\!/\prov} \]
\end{lemma}
\begin{proof}
For the right-to-left implication, we show that any derivation of $t_{\!/\prov} = t'_{\!/\prov}$ in equational logic from the axioms of $\et$ can be turned into a derivation of $t = t'$ from the axioms $\et \ \uplus \ \prov$ as follows. First, we choose for every equivalence class $\pccs_{\!/\prov}$ a unique representative and we consider the terms $t[P_{1},...,P_{n}],t'[Q_{1},...,Q_{m}]$ obtained by substituting to each occurrence of elements of $\pccs_{\!/\prov}$ in $t_{\!/\prov},t'_{\!/\prov}$ the chosen representative of the equivalence classes. Then, by substituting in the proof of $t_{\!/\prov} = t'_{\!/\prov}$ any occurrence of an equivalence class with the corresponding representative, we obtain a proof of $t[P_{1},...,P_{n}]=t'[Q_{1},...,Q_{m}]$ from axioms $\et \ \uplus \ \prov$ (note that all occurrences of $P_{\!/\prov} = Q_{\!/\prov}$ as an axiom are now substituted by $P=Q$, which is an axiom in the theory $\et \ \uplus \ \prov$).
Then, by using the axioms $\prov$ and the congruence deductive rule of equational logic, we derive in the theory with axioms $\et \ \uplus \ \prov$ that $t=t[P_{1},...,P_{n}]$.
Analogously, we derive $t'[Q_{1},...,Q_{m}]=t'$, so by transitivity we obtain a proof of $t=t'$ in the theory with axioms $\et \ \uplus \ \prov $.

For the left-to-right implication we prove, by structural induction on the derivation, that any derivation of $t = t'$ in the theory with axioms $\et \ \uplus \ \prov$ becomes a derivation of $t_{\!/\prov} = t'_{\!/\prov}$ in the theory  with axioms $\et$ by substituting all occurrences of $P\in \pccs$ in the proof of $t = t'$ with the corresponding equivalence class $P_{\!/\prov}$. In the proof, axioms $P=Q$ in $\prov$ are substituted by the reflexivity axiom $P_{\!/\prov}=Q_{\!/\prov}$.
\end{proof}

From the previous lemmas we derive the following result which is at the basis of the proof of Theorem \ref{thm:soundcomp}.

\begin{lemma}
\label{lem:prov2}
For all $t,t'\in \terms {\pccs} {\sigpcs}$,
\[\et \ \uplus \ \prov \ \vdash_E \  t = t' \text{ iff } M(q_{\prov})(\sigma_{X}([t]_{\!/\et})) = M(q_{\prov})(\sigma_{X}([t]_{\!/\et})) \]
\end{lemma}
\begin{proof}
By Lemma \ref{lem:prov1} we have
\[\et \ \uplus \ \prov \ \vdash_E \  t = t' \text{ iff } \et  \ \vdash_E \  t_{\!/\prov} = t'_{\!/\prov} \]
and by the definition of $t_{\!/\prov}$ we have
\[\et  \ \vdash_E \  t_{\!/\prov} = t'_{\!/\prov} \text{ iff } \termmon(q_{\prov})([t]_{\!/\et}) = \termmon(q_{\prov})([t']_{\!/\et})\]
By Lemma \ref{lem:isomm} it holds
\[\termmon(q_{\prov})([t]_{\!/\et}) = \termmon(q_{\prov})([t']_{\!/\et}) \text{ iff } M(q_{\prov})(\sigma_{X}([t]_{\!/\et})) = M(q_{\prov})(\sigma_{X}([t']_{\!/\et}))\]
and so we conclude.
\end{proof}

For the completeness result, we also use the following inductive characterisation of process terms of depth at most $n$.

For any $n\geq0$, we define the set $\pccsn n $ of process terms of depth at most $n$, by induction on $n$.
\[\pccsn 0=\emptyset \quad \pccsn {n+1}= \{P\in \pccs \mid \tau(P) \in \terms {\pccsn n} {\sigpcs}\}\cup \pccsn n
\]

Note that for $n\leq m$ it holds $\pccsn {n} \subseteq \pccsn m$.

\begin{lemma}
For all $P\in \pccs$ there exists an $n\geq 0$ such that $\tau(P) \in \terms {\pccsn n} {\sigpcs}$.
\end{lemma}
\begin{proof}
We show by induction on the definition of the process grammar that for all $P\in \pccs$ there exists an $n\geq 0$ such that $\tau(P) \in \terms {\pccsn n} {\sigpcs}$. 
For $P=\nil$ we have $\tau(\nil)=\star\in \terms {\emptyset} {\sigpcs}= \terms {\pccsn 0} {\sigpcs}$.
Now assume by inductive hypothesis that $\tau(P_{1}) \in \terms {\pccsn n} {\sigpcs}$ and $\tau(P_{2}) \in \terms {\pccsn m} {\sigpcs}$ for some $n,m\geq 0$.
If $P=P_1  \ndplusccs P_2$ we have $\tau(P_1  \ndplusccs P_2)= \tau(P_{1}) \oplus \tau(P_{2})$, which is in $\terms {\pccsn {\max\{n,m\}}} \sigpcs$ by the inductive hypothesis.
The case $P=P_{1} \pplusccs p P_{2}$ follows analogously.
For $P=a.P_{1}$ we have $\tau(a.P_{1})=P_1 \in \terms {\pccsn n} {\sigpcs}$.
\end{proof}


This allows us to derive that:


\begin{corollary}
\label{cor:pccsn}
$\pccs=\cup_{n\geq0} \pccsn n$.
\end{corollary}

We are finally ready to prove Theorem \ref{thm:soundcomp}.

{\bf Proof of Theorem \ref{thm:soundcomp}}

We first prove the soundness of the proof system, i.e.,
\[P \prov Q \text{ implies } P \simeq_M Q\]
by showing that $\prov$ is a behavioural equivalence.
By definition $P \prov Q$ implies $\et \ \uplus \ \prov  \ \vdash_E \  \tau(P) = \tau(Q)$. By Lemma \ref{lem:prov2}, this is equivalent to 
$M(q_{\prov})(\sigma_{X}([\tau(P)]_{\!/\et}))) = M(q_{\prov})(\sigma_{X}([\tau(Q)]_{\!/\et})))$.
Hence, by Definition \ref{def:bis}, $P \prov Q$ is a behavioural equivalence and the proof is completed.

For the completeness proof, we need to show that for all $P,Q \in \pccs$ and for any behavioural equivalence $R$, 
$$P \,R\, Q \text{ implies } P\prov Q.$$

Using Corollary \ref{cor:pccsn}, we proceed by induction on $n\geq 0$ showing that the implication holds for all $P,Q \in \pccsn n$.

The case $n=0$ is trivial, as $\pccsn 0=\emptyset$.

Suppose $P,Q \in \pccsn {n+1}$ and let $R$ be a behavioural equivalence. If $P \,R \,Q$ then by the definition of behavioural equivalence it holds $M(q_{R})(\sigma_{X}([\tau(P)]_{\!/\et})) = M(q_{R})(\sigma_{X}([\tau(Q)]_{\!/\et}))$. Now, as by the definition of $\pccsn {n+1}$ all process terms $P'$
occurring in $\tau(P)$, $\tau(Q)$ have depth at most $n$, we can apply the inductive hypothesis to derive that on all such process terms $P'$ it holds $q_{R}(P')\subseteq q_{\prov}(P')$, i.e., the set of process terms $R$--equivalent to $P'$ is included in the set of terms $q_{\prov}$--equivalent to $P'$.
Hence, it follows from $M(q_{R})(\sigma_{X}([\tau(P)]_{\!/\et})) = M(q_{R})(\sigma_{X}([\tau(Q)]_{\!/\et}))$
 that $M(q_{\prov})(\sigma_{X}([\tau(P)]_{\!/\et})) = M(q_{\prov})(\sigma_{X}([\tau(Q)]_{\!/\et}))$.
 By Lemma \ref{lem:prov2}, this is equivalent to $\et \ \uplus \ R \ \vdash_E \  \tau(P) = \tau(Q)$, which implies $P \prov Q$.